\documentclass[preprint,11pt]{article}
\usepackage[ruled,vlined]{algorithm2e}
\usepackage{algorithmic}

\usepackage{todonotes}
\usepackage{xcolor}
\usepackage{ulem}
\usepackage{cancel}

\newcommand{\Dave}[1]{\todo[inline,color=red!40]{Donoho: #1}}
\newcommand{\Sun}[1]{\todo[inline,color=blue!40]{Sun: #1}}

\newcommand{\bR}{{\bf R}}

\usepackage{amssymb,amsfonts,amsmath,amsthm,amscd,dsfont,mathrsfs}
\usepackage{graphicx,float,psfrag,epsfig}
\usepackage{wrapfig}
\usepackage{relsize}
\usepackage{color}
\usepackage{pict2e}
\usepackage[tight]{subfigure}

\usepackage{comment}

\usepackage[utf8]{inputenc}

\DeclareMathAlphabet{\mathpzc}{OT1}{pzc}{m}{it}

\newcommand{\BEAS}{\begin{eqnarray*}}
\newcommand{\EEAS}{\end{eqnarray*}}
\newcommand{\BEA}{\begin{eqnarray}}
\newcommand{\EEA}{\end{eqnarray}}
\newcommand{\BEQ}{\begin{equation}}
\newcommand{\EEQ}{\end{equation}}
\newcommand{\BIT}{\begin{itemize}}
\newcommand{\EIT}{\end{itemize}}
\newcommand{\BNUM}{\begin{enumerate}}
\newcommand{\ENUM}{\end{enumerate}}

\newcommand{\BA}{\[\begin{array}{ll}}
\newcommand{\EA}{\end{array}\]}

\newcommand{\reals}{{\mbox{\bf R}}}
\newcommand{\integers}{{\mbox{\bf Z}}}

\newcommand{\Expect}{{\mathbb E}}

\newcommand{\Var}{\mathop{\bf var{}}}

\newtheorem{theorem}{Theorem}

{\begin{quote}}{\end{quote}}

\makeatletter
\long\def\@makecaption#1#2{
   \vskip 9pt 
   \begin{small}
   \setbox\@tempboxa\hbox{{\bf #1:} #2}
   \ifdim \wd\@tempboxa > 5.5in
        \begin{center}
        \begin{minipage}[t]{5.5in}
        \addtolength{\baselineskip}{-0.95pt}
        {\bf #1:} #2 \par
        \addtolength{\baselineskip}{0.95pt}
        \end{minipage}
        \end{center}
   \else 
	\hbox to\hsize{\hfil\box\@tempboxa\hfil}  
   \fi
   \end{small}\par
}
\makeatother

\newcounter{oursection}

\newcounter{lecture}

\newtheorem{propo}{Proposition}[section]
\newtheorem{lemma}[propo]{Lemma}
\newtheorem{definition}[propo]{Definition}
\newtheorem{coro}[propo]{Corollary}

\def\cF{{\cal F}}

\def\cB{{\cal B}}

\def\cT{{\cal T}}

\def\integers{{\mathbb Z}}
\def\reals{{\mathbb R}}

\def\eps{{\varepsilon}}

\def\E{{\mathbb E}}
\def\Var{{\rm Var}}

\def\L0{{L_0}}

\def\<{\langle}
\def\>{\rangle}

\def\bX{{\mathbf X}}

\def\F{{\sf F}}

\def\F{{\sf F}}

\def\v*{v_0}
\def\T*{T_0}

\def\u*{u_0}
\def\F*{F_0}

\definecolor{olivegreen}{rgb}{0,0.6,0.4}

\newcommand{\bitem}{\begin{itemize}}
\newcommand{\eitem}{\end{itemize}}
\newcommand{\goto}{\to}
\newcommand{\beq}{\begin{equation}}
\newcommand{\eeq}{\end{equation}}

\def\Var{{\rm Var}}

\def\cI{{\cal I}}

\newcommand{\ajcomment}[1]{}

\makeatletter
\newcommand{\labitem}[2]{%
\def\@itemlabel{\text{#1}}
\item
\def\@currentlabel{#1}\label{#2}}

\usepackage{bibentry}
\newcommand{\ignore}[1]{}
\newcommand{\nobibentry}[1]{{\let\nocite\ignore\bibentry{#1}}}

\makeatother

\addtocontents{toc}{\protect\setcounter{tocdepth}{2}}

 \newcommand{\bx}{{\bf x}}
 \newcommand{\by}{{\bf y}}
 \newcommand{\be}{{\bf e}}
  \newcommand{\ba}{{\bf a}}
 \newcommand{\tbe}{\widetilde{\be}}
\newcommand{\bw}{{\bf w}}
\newcommand{\bv}{{\bf v}}

\title{Convex Sparse Blind Deconvolution}
\author{Qingyun Sun \and David Donoho
}

\begin{document}

\date{}

\maketitle

\begin{abstract}
In the {\it blind deconvolution problem}, we observe the convolution $\by=\ba \star \bx$ of an unknown filter $\ba$ and unknown signal $\bx$ and attempt to reconstruct the filter and signal. The problem  seems impossible in general, since there are seemingly many more unknowns in $\bx$ and $\ba$ than knowns  in $\by$. Nevertheless, this problem arises -- in some form -- in many application fields; and empirically, some of these fields have had success using heuristic methods -- even economically very important ones, in wireless communications and oil exploration.

Today's fashionable heuristic formulations pose 
non-convex optimization problems which are then attacked heuristically as well. 
The fact that blind deconvolution can be solved under some 
repeatable and naturally-occurring circumstances poses a theoretical puzzle.

To bridge the gulf between reported successes and theory's limited understanding, we exhibit a {\it convex} optimization problem that \-- assuming the signal to be recovered is {\it sufficiently sparse} \-- can convert a \it crude approximation to the true filter into a \it high-accuracy recovery of the true filter. 

Our proposed formulation is based on $\ell^1$
minimization of inverse filter outputs:
\[
 \begin{array}{ll}
 \underset{\bw \in l_1^k}{\mbox{minimize}}   & \| \bw \star \by\|_{\ell_1^N}\\
\mbox{subject to}  
& \langle \widetilde{\ba},  \bw^{\dagger} \rangle =1.
\end{array}
  \]
Minimization inputs include: the observed blurry signal $\by \in \bR^N$; and
$\widetilde{\ba} \in \bR^k$, an initial approximation of the true unknown filter
$\ba$. Here $\bw^{\dagger}$ denotes the time-reverse of $\bw$.
Let $\bw^*$ denote the minimizer. 


We give sharp guarantees on performance of $\bw^*$
assuming sparsity of $\bx$,
showing that, under favorable conditions,
our proposal precisely recovers the true inverse filter 
$\ba^{-1}$, up to shift and rescaling.

Specifically, in a large-$N$ analysis where $\bx$ is 
an $N$-long  realization of an IID Bernoulli-Gaussian signal with expected sparsity level $p$, we measure the approximation quality of the initial approximation $\widetilde{\ba}$ 
by considering $\tbe= \widetilde{\ba} \star \ba^{-1}$, 
which would be a Kronecker (aka $delta$) sequence if our approximation were perfect.
Under the gap condition
\[
\frac{|\tbe|_{(2)}}{|\tbe|_{(1)}} \leq 1-p,
\]
we show that, in the large-$N$ limit,
the $\ell^1$ minimizer $\bw^*$
perfectly recovers $\ba^{-1}$ to shift and scaling. 

Here the multiplicative gap 
$\frac{|\tbe|_{(2)}}{\tbe|_{(1)}} \leq 1$ denotes 
the ratio of the  first and second largest entries 
of $|\tbe|$,and is a natural measure of closeness 
between our approximate  $\delta$ sequence $\tbe=\widetilde{\ba} *\ba^{-1}$ 
and a true $\delta$ sequence .

In words, there is a sparsity/initial accuracy tradeoff: {\it the less accurate the initial approximation $\widetilde{\ba} \approx \ba$, the greater we rely on sparsity of $\bx$   to enable exact recovery.} To our knowledge this is the first reported tradeoff of this kind.  We consider it surprising that this tradeoff is independent of dimension $N$, i.e. that the gap condition does not demand increasingly stringent accuracy with increasing $N$.

We also develop finite-$N$ guarantees of the form $N\geq O(k\log(k))$,
for highly accurate reconstruction with high probability.
We further show stable approximation when the true inverse filter is infinitely long
(rather than  finite length $k$). And we extend our guarantees to the case where
the observations  are contaminated by  stochastic or adversarial noise, and show that the error is linearly bounded by the noise magnitude.




%
\end{abstract}

\tableofcontents
\pagebreak

\section{Introduction}
\subsection{Blind Deconvolution}
Suppose we are interested in an underlying time series $\bx = (x(t))$ which we cannot observe directly. What we can observe is  $\by = \ba * \bx$ where $\ba$ is an unknown `blurring' filter. Blind deconvolution is the problem of recovering $\bx$ merely from the observed $\by$ without knowing either $\ba$ or $\bx$. 

This problem occurs naturally in seismology and
digital communications
as well as astronomy, satellite imaging, and computer vision.

In its most ambitious form, the problem is literally impossible; there are simply too few data and too many unknowns. Indeed, imagine that $\bx$, $\by$ and $\ba$ all have $N$ entries; we observe only $N$ pieces of information ($\by$) but there are $2N$ unknowns ($\bx$ and $\ba$). 
Nevertheless, in  some (not all) fields, heuristic approaches
have occasionally 
led to consistent  success in isolated applications;
presumably such success stories exploit specialized assumptions -- although not always in
an explicit or recognized way, and not with rigorous understanding.

 This paper, in contrast,  will exhibit a set of 
assumptions enabling practical algorithms for blind deconvolution,
backed by rigorous theoretical analysis.

\subsection{The Promise of Sparsity}

{Central to our approach
is an assumption about the}
{\it sparsity of the signal $\bx$ to be recovered}, 
in which case the problem 
{can be called} \emph{sparse blind deconvolution}. Sparse signals {-- i.e. signals} having relatively few nonzero entries, arise frequently in many fields, including seismology, microscopy, astronomy, neuroscience spike identification. Even in more abstract settings such as representation learning for computer vision, 
it surfaces in recently popular research trends, 
such as single-channel convolutional dictionary learning   \cite{bristow2013fast, heide2015fast,zhang2017global,zhang2018structured}. 

The sparsity of $\bx$ \-- if it holds \-- would constrain the recovery problem significantly; and so possibly, sparsity can play a role
in enabling useful solutions to an otherwise hopeless problem.

An inspiring precedent can be found in modern commercial medical imaging, 
where sparsity of an image's wavelet coefficients enables 
MRIs from fewer observations than unknowns.
Taking fewer observations speeds up data collection,
a principle known as compressed sensing, which
already benefits tens of millions of patients yearly.

\subsection{Translating Heuristics into Effective Algorithms}

For sparsity to reliably enable blind deconvolution, there are two
apparent hurdles. First, develop an objective function which promotes
sparsity of the solution. Second, develop an algorithm 
which can reliably optimize the objective.

Many sparsity-promoting objectives have 
been proposed over the years; typically they 
imply non-convex optimization problems. Indeed
sparsity is quantified by the $\ell_0$
pseudo norm $\|\bx\|_0 = \#\{t: x(t) \neq 0\}$, which is the limit of
$\| \bx \|_p^p$ as $p \goto 0$
of concave $\ell_p$ pseudo-norms where $p <1$.

Traditionally, non-convex problems have been viewed
by mathematical scientists with skepticism;
for such problems, gradient descent and its various refinements
lack any guarantee of effectiveness. 
Still, the lack of guarantees has not stopped engineers from trying!

A noticeable success in blind signal processing
was scored in digital communications,
where blind equalization today 
benefits billions of smartphone users.
Blind equalization is a form of blind deconvolution 
where one exploits the known discrete-valued 
nature of the signal $\bx$ (for example the signal 
entries might take only two values $\{-1,1\}$).
Practitioners found that if an initial guess of 
the equalizer (i.e. our inverse filter $\ba^{-1}$) is `fairly good' 
(in engineer-speak `opening the eye' so that a `hint' 
of the `digital constellation` becomes
`visible'), then certain `discreteness-promoting' on-line gradient methods
can reliably `focus' the result better and better 
and allow reliable recovery.

Our work identifies an 
analogous phenomenon in the sparsity-promoting blind deconvolution
setting, however it exposes and crystallizes the phenomenon in a rigorous 
and dependably exploitable form.

Namely, we show that if sparsity of $\bx$ holds,
and if an initial guess of the filter $\ba$
is `fairly good' in a precise sense, 
then a specific convex optimization
algorithm will accurately recover 
both the filter and the original signal\footnote{As we explain below, recover means: recover up to rescaling and time shift.}.

In retrospect, our insights on sparsity-promoting blind deconvolution
can be cross-applied to explain the major successes of discreteness-promoting
blind equalization in modern digital communications.
Namely, a direct variation of our arguments 
provide a related convex optimization problem 
for discrete-valued signals which rigorously converts a 
a rough initial approximation into precise recovery.

In our view these new arguments clear away some persistent
fog, mystery and misunderstandings in blind signal processing; and pave
the way for future success stories.


\subsection{Prior Work}

\paragraph{Searching for an inverse filter that promotes desired output properties}
Instead of trying to recover $\ba$ and $\bx$ together from $\ba*\bx$, 
we could formulate this problem as looking for an approximate inverse  filter $\bw$ so that the output $\bw*\by$ exhibits {extremal} properties. 
 {Under this formulation, 
our goal is to find} $\bw\neq 0$, so that  
\begin{equation}
 \begin{array}{ll}
 \underset{\bw\neq 0}{\mbox{optimize}}   & J(\bw \star \by) \label{JoptBD}
\end{array}
\end{equation}
where the functional $J$ quantifies the properties we seek to promote.
(Depending on $J$, we might  {either prefer its maximum or minimum}  ).
 
 \newcommand{\bz}{{\bf z}}
 {Working} in
exploration seismology, 
Wiggins \cite{wiggins1978minimum} adopted this approach
with the normalized $4$-norm, $J(\bz) = J_{4,2} = {\|\bz\|_4}/{\|\bz\|_2}$
and gave a few successful data-processing case studies. 
His {successful} examples all
clearly exhibit sparsity, although this was not discussed at the time.
Other objectives 
considered at that time included $J_{2,1}(\bz) = \|\bz\|_2/\|\bz\|_1$
and $J_{\infty,2}(\bz) = \|\bz\|_\infty/\|\bz\|_2$ \cite{cabrelli1985minimum}.

It was fully understood at that time
that output property optimization could succeed in principle,
if one did not have to worry about an effective algorithm.
\cite{donoho1981minimum} showed that 
if the signal $\bx$ is a realization of
independent and identically distributed
entries from any nonGaussian distribution,  optimizing
$J$ of the output $\bw \star \by$ is successful in the large-$N$ setting \-- as long as the  functional  $J$ belongs to a certain large family of non-Gaussianity measures,
for example including $J_{4,2}$ and $J_{2,1}$ as well as many others. 
\footnote{Sparsity is of course a form of non-Gaussianity,
this very explicitly in the Bernoulli-Gaussian mixture model considered below.}

The issue left unresolved in those days was how to solve such
optimization problems. Indeed optimizations like $J_{4,2}$ are badly nonconvex,
as we see clearly by rewriting the $J_{4,2}$ problem as
\[
 \begin{array}{ll}
 \underset{\bw}{\mbox{maximize}}   &  
 \|\bw * \by\|_4 
 \\
\mbox{subject to} 
& \|\bw* \by\|_2 = 1.
\end{array}
  \]
The theory cited above derived favorable properties of a would-be procedure which truly finds the optimum of a badly nonconvex objective. It clarifies that  blind deconvolution is possible in principle but does not by itself help us algorithmically, i.e. in practice. In the intellectual climate of the time, solving badly nonconvex optimization problems was considered a pipe dream, a time-wasting charade for non-serious people.

Even  today, blind deconvolution 
continues to be studied as an non-convex optimization problem;
see recent work
\cite{kuo2019geometry, kuo2020geometry, lau2019short}, who study the  problem of recovering short $a$ and sparse $x$ 
  
  \paragraph{Blind equalization}
  In digital communications, 
  the transmitted signal $\bx$ can be viewed as 
  {\it discrete-alphabet valued;}
  for example, in PAM signaling, where $x(t) \in \{\pm 1\}$, and QAM signaling, where the signal alphabet has equally spaced points on unit sphere in complex space \cite{kennedy1992blind,ding2000fast}.  
  
  \cite{vembu1994convex} considered the non-convex problem
   \[
 \begin{array}{ll}
 \underset{\bw}{\mbox{maximize}}   &  
 \|\bw * \by\|_8 
 \\
\mbox{subject to} 
& \|\bw* \by\|_2 = 1
\end{array}
  \]
If the data $\by$ were preprocessed to be serially uncorrelated,
this optimization is effectively of the earlier form $J_{8,2}$.

The authors attack this nonconvex problem
using projected gradient descent and give suggestive 
experimental results. They apparently view the $\ell_8$ norm objective as an approximation of the  $\ell_\infty$ norm objective $\|\bw * \by\|_\infty $.
Later, \cite{ding2000fast} used linear programming to solve the $\infty$ norm problem directly.

 \paragraph{Searching for a projection with desired output properties} 
Here is another setting for property-promoting
output optimization.
We have a data matrix $Y \in \bR^{n,p}$ \-- which we think of as
$n$ points in $\bR^p$,
and we have a unit vector $w \in \bR^p$
called the projection direction.  Our
output vector $Y\cdot w$ contains the 
projection of the $n$-points on the projection
direction q; it has $n$ entries.
We seek `interesting' projections; i.e. directions where the 
projection displays some structure. We
adopt a functional $J$ which measures properties we seek to promote
in the output, and we seek to solve:
\begin{equation}
 \begin{array}{ll}
 \underset{\|w\|=1}{\mbox{optimize}}   & J(Y \cdot w). \label{JoptPP}
\end{array}
\end{equation}

This was implemented by  [Friedman and Tukey, 1974], who proposed a
functional that promotes `clumping' or 'clustering' of the output.
They called it {\it projection pursuit} and
were motivated by exploratory  high  dimensional  data  analysis;
for $p$-dimensional data involve $p>2$ and we can't 
easily get a visual sense of what's in the data. 
It was hoped at the time that looking at selected
low-dimensional projections might lead to better insights.
For the most part, such hopes for exploration of high-dimensional data
never materialized.

However, output optimization of this type has proven to be
useful in important problems in blind signal processing,
where $Y$ has known structure that can be exploited systematically.

In {\it blind source separation}
we observe $Y = XA$, $Y \in \bR^{n,p}$, $X \in \bR^{n,p}$ and
$A \in \bR^{p,p}$ is an invertible matrix. 
(We don't observe $X$ or $A$ separately).

Think of the $X$ matrix as containing columns giving
successive time samples of $p$ clean source signals, for example $p$
individual acoustic signals. These signals arrive
at array of $p$ spatially distributed acoustic sensors, each
of which records the acoustic information it receives.
The matrix $Y$ contains in its
columns what was obtained by each of the $p$ different sensors.
In general, each sensor receives information from each of the sources.
This is colorfully called the {\it cocktail party problem}, referring to
the setting where $X$ records the sources are human speakers at a cocktail party,
and $Y$ records what is heard at various locations in a room.
Each column of $Y$ then contains a
superposition of different speakers; while
we would prefer to separate these and pay attention to 
just the ones of most immediate interest to us.

In this separation problem, sparsity of the signal might be valuable.
Suppose that each individual speaker is listening quite a bit
and so not speaking much of the time. Then the columns of $X$
are each sparse. On the other hand, if there are many people
in the room, the room
as a whole may still always be noisy, and so each
column of $Y$ may be fully dense.
Assuming $A$ is invertible, and the vector $w$
obeys $Aw= e_j$, then $Y \cdot w$ extracts column $j$ of $X$,
which will be sparse. Hence, we may hope that the projection
pursuit principle, with an appropriate measure $J$, might
identify the `sparse' projections we seek.  

Ju Sun, Qu Qing, Yu Bai, John Wright, Zibulevsky and Pearlmutter \cite{sun2015complete,sun2016complete,bai2018subgradient,zibulevsky2000blind} proposed the optimization
problem
 \[
 \begin{array}{ll}
 \underset{w }{\mbox{minimize}}   & \| Y w\|_{1}\\
\mbox{subject to}  & \|w\|_2 =1.
\end{array}
  \]
Assuming the $Y$ data pre-processed so that $\frac{1}{n} Y'Y=I_p$ 
this is equivalent to projection pursuit (\ref{JoptPP}) 
applied to the objective $J_{2,1}$. Notably, this is 
again a highly non-convex optimization problem..
The authors proposed projected gradient descent
and recited some favorable empirical results.
 
There is a close relation between the
 blind deconvolution optimization (\ref{JoptBD}) 
 and the projection pursuit optimization (\ref{JoptPP}). Indeed, if the $Y$ matrix 
 is filled in from an observed time series appropriately, then 
 output optimization in 
 blind deconvolution and in projection pursuit are
 essentially identical. Namely, let $\by^{ser} = (y^{ser}(t))_{t=1}^N$ denote a time series of interest to us,
 and  $Y^{mat} = (Y^{mat}_{i,j})$ 
 denote an $n \times p$ matrix where $n = N-p$ constructed 
 using $\by^{ser}$ like so:
 \[
  Y^{mat}_{i,j} = y^{ser}(p+i-(j-1)) , \qquad 1 \leq i \leq n = N-p; \; 1 \leq j \leq p.
  \]
  Now suppose the filter vector 
  $\bw$ in the blind deconvolution optimization 
  and the projection direction $w$ in the projection pursuit optimization
  are chosen identically. Then the blind deconvolution output objective 
  $J(\bw \star \by^{ser})$ is identical to the projection pursuit objective
  $J(Y^{mat} w)$, except for
  possibly different treatment of the first $p$ entries of $\by^{ser}$.
 
 \paragraph{Convex Projection Pursuit}
 
 In view of the connection between blind deconvolution and
 projection pursuit, and in view of our results in this paper,
 it is quite interesting to consider the
 work of \cite{spielman2012exact} and \cite{gottlieb2010matrix}.
They propose to solve the following 
linear-constrained convex optimization problem.
Given an $n \times p$ 
data matrix $Y$ and a constraint vector $u$,
they propose to solve:
\[
 \begin{array}{ll}
 \underset{w }{\mbox{minimize}}   & \|Y w\|_{1}\\
\mbox{subject to}  & u^T w =1.
\end{array}
  \]
As it turns out,  when the matrix $Y$ in this problem and
$\by$ in the time series deconvolution problem
are related by the  $Y^{mat}$-$\by^{ser}$
construction just mentioned, the objective we propose in this paper
is essentially identical, when $\tilde{\ba} = u$. As we will show,
our setting permits much more thorough studies and more penetrating analyses, {and we find that success in 
sparse blind deconvolution is more
broadly prevalent than one might have expected, based on
earlier analyses such as \cite{spielman2012exact}
or \cite{sun2015complete,sun2016complete,bai2018subgradient,zibulevsky2000blind} }.

 

\subsection{Mathematical setup}
\paragraph{Sequence Space, and Filtering }
To make our results concrete, let's discuss things formally.
Let $\bX$ denote the collection of {\it bilaterally infinite sequences}
$\bx = (x(t): t = 0, \pm 1, \pm 2, \dots)$; for short
we call such objects {\it bisequences}. Then $\ell_1(\integers) \subset \bX$
denotes the collection of bisequences 
obeying $\| \bx \|_1 = \sum_t |x(t)| < \infty.$
For a bisequence $\bx$ we denote time reversal operator
$t \leftrightarrow -t$ by $\bx^{\dagger}$.
For whole number $k > 0$ let $\bX_k$ denote the subspace
of bisequences supported in $-k \leq t \leq k$. We sometimes
abuse notation: for a bisequence $\bx$ we might write
$\bx = (1, .3)$ when we really mean $\bx = (\ldots,0,0,1,.3,0,0,\ldots)$.

Let $\star$ denote the convolution 
product on pairs of bisequences in $\ell_1(\integers)$
\-- $( \bx \star \by ) = \sum_u x(t-u)y(u)$. 
Let $\be_0$ denote the  `delta'
or `Kronecker' bisequence: $\be_0(t) = 1_{\{t=0\}}$;
$\be_0$ is the unit of convolution. The convolution inverse
of $\bx$ $\--$ $\bx^{-1}$ $\--$ is a bisequence
obeying $ \bx \star \bx^{-1}  = \be_0$. For example,
the filter $\bx = (\dots,0,1,1/2,1/4,1/8,\ldots)$ 
anchored at the time origin so $\bx(0)=1$, has inverse
$\bx^{-1} = (\dots,0,1,-1/2,0,\dots)$, again anchored at the time origin.
Abusing notation we may simply write $\bx^{-1} = (1,-1/2)$.

Our approach to blind deconvolution searches 
among candidates for a filter ($\equiv$ bisequence) that extremizes a certain objective function. We then
show that the extremal is in fact the desired inverse filter
to our (unknown) true underlying filter. Hence, it helps know
conditions under which an inverse filter actually exists! 

For a bisequence $\bx \in \ell_1(\integers)$,
we define the Fourier transform 
$\hat{\bx}(w) = \mathcal{F} \bx (w) \equiv \sum_t x(t) \exp\{ i 2\pi w t \}$.
For the inverse transform, we use
$ x(t) = (\mathcal{F}^{-1} \hat{\bx} ) (t) =  \frac{1}{2\pi} \int_0^{2\pi} \hat{\bx}(w) \exp\{ -i 2\pi w t \} dw$.

\begin{lemma}[Wiener's lemma]
If $\ba\in \ell_1(\integers)$, and also  
$(\mathcal{F} \ba)(\widetilde{w}) \neq 0$, $\forall \widetilde{w}\in T$,
then an inverse filter exists in $\ell_1(\integers)$. 
The bilaterally infinite sequence defined formally by
$$\ba^{-1} := \mathcal{F}^{-1}(\frac{1}{\mathcal{F} \ba})$$
exists as an element of $\ell_1(\integers)$
and obeys $\ba^{-1}*\ba = \be_0.$
\end{lemma}
In the engineering literature, we
say that the bisequence $\ba$ has so-called $Z$-transform 
$A(z)$, defined by:
\[
 \begin{array}{ll}
A(z) =  \sum_{t=-\infty}^{\infty} a_{t} z^{-t}.
\end{array}
  \]

Evaluating $A$ on the unit circle in the complex plane, at $z$ of the form $z=\exp\{ -i 2\pi w \}$, we see that  the 
$Z$-transform is effectively the Fourier transform $A = \cF \ba$. 
Applying Wiener's lemma, we see that, if 
$\ba \in \ell_1$ and $ \min_w |A(\exp\{ - 2\pi i w \}) | > 0$,
i.e. $A$ is never zero on the unit circle, then $\ba^{-1} \in \ell_1$ and $1/A$ is the $Z$-transform of $\ba^{-1}$.

\paragraph{Finite-sample observation model and finite-length inverse filter }
In searching for an inverse filter, our  initial  results assume
existence of a {\it finite-length} inverse. Namely, we assume that $\ba\in \ell_1(\integers)$ is a forward filter with an inverse filter $\ba^{-1} \in \ell_1(\integers)$ supported in a centered window of radius $k$.

In practice we only have a finite dataset! Suppose that 
there is an underlying bisequence $\by \in \bX$
of the form $\by = \ba* \bx$, and let $\by^{[N]}$ denote the
restriction to an $N$-long centered window 
$\cT = \{ -T,\dots, T\}$ of radius $T$ and size $N=2T+1$.

Our goal is seemingly to 
recover $\bx$ or $\ba$ from the observed data $\by^{[N]}$.
However, in statistical theory we generally 
don't expect to exactly recover the true underlying 
representation (i.e. the generating $\ba$ and $\bx$) exactly.
Our goal is instead to find an inverse filter 
$\bw \in \ell_1^k$ such that the convolution $\bw*\by$ 
would be close\footnote{modulo time shift and rescaling}  to $\bx$, and where the closeness improves with
increasing data size $N \goto \infty$.

\paragraph{Finite-length filtering and practical algorithms}

While our analysis framework concerns bisequences (bilaterally infinite sequences), our data have finite length (as just mentioned). The algorithms we discuss are often {\it motivated by} convolutions on bisequences; however, they reduce {\it in practice} to truncated convolutions involving finite data windows. A certain ambiguity is helpful for efficient communication. Suppose we have $\by^{[N]}$, an N-long observed window of $\by$, and we also have $\bw$, a k-long filter, by which we mean a bisequence nonzero only within a 
fixed $k$-long window. We might encounter discussion both of 
$\bw \star \by^{[N]} $ as well as
$\bw \star \by$, both using the same filter $\bw$. In the first case, we would actually be thinking of $\by^{[N]}$ as zero-padded out to a bisequence, so that both situations involve bisequence convolutions. 

Finite $N$ effects are important in practice but tedious to discuss. It can be important to account for {\it end effects} in truncated convolution. In a setting where we initially
think to consider a norm $\| \bw \star \by^{[N]} \|_{\ell_p(\integers)}$, we might instead next think to rather consider the windowed norm
$\| \bw \star \by^{[N]} \|_{\ell_p(\{-T,\dots,T\})}$,
while finally we realize $\| \bw \star \by^{[N]} \|_{\ell_p(\{-T+k,\dots,T-k\})}$ is more correct for our purposes, as it includes only the terms which do not suffer from truncation of the convolution.



\paragraph{Algorithmic formulation}

Under assumptions we will be making, the sequence $\bx$ underlying our observed data will be either exactly sparse -- having few nonzeros -- or approximately so. Moreover, there will be either an exact length $k$ inverse filter $\bw$, or approximately such. It follows that the filter output $\bw \star \by$ is sparse.
This suggests the would-be optimization principle 
\[
 \begin{array}{ll}
 \underset{\bw\neq 0}{\mbox{minimize}}   & \|\bw \star \by^{[N]}\|_{\ell_0(\{-T+k,\dots,T-k\}})
\end{array}
  \]
where the $\ell_0$ quasi-norm simply 
counts the number of nonzero entries. Unfortunately, this objective, though well-motivated, is not suitable for numerical optimization.

Inspired by this, we perform convex relaxation of the $\ell_0$ norm, replacing it with the $\ell_1$ norm, 
which is convex.

We also need to fix the scale to get a unique output. One might think to constrain $\|\bw\|_2 =1$, however, this would give a non-convex constraint and again is not suitable for effective algorithms.


We instead suppose given a rough initial approximation  
$\widetilde{\ba}$ of the forward filter $\ba$, and impose 
an $\ell_\infty$ constraint on the 
`pseudo-delta' $\bw\star\widetilde{\ba}$,
forcing it to `peak' at target entry $t$.
\BEQ
 \begin{array}{ll}
&  (\widetilde{\ba}*\bw)_t =1, \quad \|\widetilde{\ba}*\bw\|_\infty \leq 1.
\end{array}
\EEQ
Combining these steps, we obtain a convex optimization problem 
associated to each possible target coordinate $t$:
\BEQ
 \begin{array}{ll}
 \underset{\bw \in \ell_1^k}{ \mbox{minimize} }   & \frac{1}{N-2k} \| \bw \star \by^{[N]}\|_{\ell_1(\{-T+k,T-k\})}\\
\mbox{subject to}  
&  (\widetilde{\ba}*\bw)_t =1, \quad \|\widetilde{\ba}*\bw\|_\infty \leq 1.
\end{array}
\EEQ

In the case of perfect shift symmetry of convolution ( linear convolution on bisequence, or circular convolution for finite observation ), we could design the problem to fully utilize the shift symmetry of objective  $\frac{1}{N} \| w*Y\|_{\ell_1(\cT)}$, so that the convex problem is equivalent to a convex optimization problem with a simple linear constraint as the $\ell_\infty$ box constraint 
with a fixed location of largest entry at $0$ coordinate under shift:
\BEQ
 \begin{array}{ll}
\widetilde{\ba}^{T} w^{\dagger} = \langle \widetilde{\ba}, w^{\dagger}\rangle = (\widetilde{\ba}*w)_0  =1, 
\end{array}
\EEQ
where $w^{\dagger}$ is the time reversal of $w$.
\end{comment}

It will be convenient to reformulate slightly,  
hide consideration of end effects,
and force the peak to occur at target coordinate $t=0$.
Abusing notation somewhat, we then write:
\BEQ
 \begin{array}{ll}
 \underset{\bw \in \ell_1^k}{ \mbox{minimize} }   & \frac{1}{N} \| \bw* \by \|_{\ell_1^N}\\
\mbox{subject to}  
& \langle \widetilde{\ba},  \bw^{\dagger} \rangle =1, 
\end{array}
\tag{${C}_{\widetilde{\ba}}$ }
\EEQ

(Again $\bw^\dagger$ denotes time-reversal of $\bw$). In practice we
might truncate the convolution due to end effects, or truncate the window over which we take the norm, but we will hide such practical details in the coming material, for ease of exposition; they would not change our results.

\paragraph{Stochastic models for sparse signals}
Although our algorithms make sense in the absence of any theory,
our theoretical results concern properties of our algorithm for data generated under a probabilistic generative model i.e. a {\it stochastic signal} model. 

Let $X = (X_t)$ be a bisequence
of independent identically distributed random variables indexed by $t \in \integers$, having a common marginal CDF $F = F_X$, such that $F(x) = 1 - F(-x)$. One realization is then a sequence $\bx$ of the type discussed in earlier paragraphs.

We assume that $F$ has an atom at $0$ \-- $F = (1-p) H +p G$, where $H$ is the standard Heaviside distribution and $G$ is the standard Gaussian distribution.  We say that $F$ follows the \emph{Bernoulli$(p)$-Gaussian} model. 
Equivalently, $X_t$ is sampled IID from the \emph{Bernoulli$(p)$-Gaussian} distribution
$p N(0,1) +(1-p) \delta_0$. 

The iid process $X$ is of course ergodic. 
If $\bx$ denotes one realization of $X$, then in a window $(x_t)_{-T}^T$ of length $N$,  $\approx p N$ nonzero values  will occur,
for large $N$. Consequently, if $p \ll 1$,
realizations from $X$
will empirically be sparse.

\newcommand{\bb}{{\bf b}}
Let $Y = \ba \star X$ denote the random bisequence produced as the output of convolution of the random signal $X$
with deterministic filter $\ba \in \ell_1(\integers)$. More explicitly,
\[
    Y_t = \sum_u a(u) X_{t-u}.
\]
This defines formally a so-called {\it stationary linear process},  a classical object for which careful foundational results are well established.\footnote{In our case, we assume that $X$ is iid Bernoulli-Gaussian, so
$E |X_0| = p \cdot E |N(0,1)|$ is finite. Using this,
we can see that the sum in the display above, 
even though possibly containing an 
infinite number of terms, converges in various natural senses.}
Consider now filtering $Y$, by a length-$k$ filter 
$\bw$, producing the random bisequence $V = \bw \star Y$.
The {\it end-to-end} filter $\bb = \bw \star \ba$ 
is a well-defined element of $\ell_1(\integers)$; 
using it, we can represent 
the filtered output in terms of the underlying iid process $X$:
$V  = \bb \star X$. 
This representation shows that the filtered output series $V$ is itself a well-defined stationary linear process, and moreover, since $E |X_0| < 1$ and $\| \bb \|_{\ell_1} \leq \|\bw \|_{\ell_1} \cdot  \| \ba \|_{\ell_1} < \infty$, we have $E | V_0| < \| \bb \|_{\ell_1} < \infty$.

Any such stationary linear process is ergodic.
By the ergodic theorem, the
large-$N$ limit of the objective 
will almost surely  be an expectation over $X$:
\BEQ
 \begin{array}{ll}
\lim_{N\rightarrow \infty}\frac{1}{N} \| \bw*Y\|_{\ell_1^N}  = \lim_{N\rightarrow \infty}\frac{1}{N} \sum_{t\in \cT^N}|(\bb \star X)_t|  =  \Expect_X|(\bb \star X)_0|=\Expect|(\bw*Y)_0|.
\end{array}
\EEQ
Consequently, 
the large-$N$ properties of our proposed algorithm
are driven by properties of the following optimization problem
{\it in the population}:
 \[
 \begin{array}{ll}
 \underset{\bw }{\mbox{minimize}}   & \Expect|(\bw*Y)_0|\\
\mbox{subject to}  &  \langle \widetilde{\ba},  \bw^{\dagger} \rangle =1. 
\end{array}
  \]

\section{Main Results Overview}



\subsection{Main Result $1$: Phase Transition Phenomenon for Sparse Blind Deconvolution}

We have at last defined a convex optimization problem 
at the population level, which we will now want to study in detail.
In our studies, we can make various choices of the sparsity parameter $p$, of the underlying forward filter $\ba$, and of the guess $\tilde{\ba}$. The tuple $(p,\ba,\tilde{\ba})$ defines in this way a kind of phase space. We can then study performance of the algorithm at different points in phase space. 

\newcommand{\ERP}{{\sc Exact}{\sc Recovery}}
Consider this performance property: 
\begin{quotation}
\ERP \, $\equiv$ \,  {\sl ``There is an unique solution of the population-based optimization problem --- modulo time shift and output rescaling --- and
this solution {\bf exactly} solves the blind deconvolution problem correctly. ''}
\end{quotation}
It probably seems too much to ask that such a property could {\sl ever} be true, i.e. could ever be true for even one 
choice of phase space tuple. After all the optimization problem doesn't have any apparent connection to blind deconvolution -- instead only to some sort of relaxation of the search for sparse output filters.

We will see that this phase space can be partitioned into two regions: one where the exact recovery property holds and its complement where the exact recovery property fails. Surprisingly the region where \ERP \  holds is nonempty and can be appreciable. And it can be described in a clear and insightful way.

\paragraph{Surprising phase transition for a special case of sparse blind deconvolution}

We first illustrate the phase transition phenomenon in possibly the most elementary situation. Consider the special case when the forward filter $\ba$ is the exponential decay filter:  $\ba = (1, s, s^2, s^3, \ldots )$ with $|s|\leq 1$; then the inverse filter is a basic short filter  $\ba^{-1} = ( 1, -s )$. 
\begin{theorem}[Population phase transition for exponential decay filter]
\label{thm: single root PT}
Consider a linear process $Y = \ba * X$ with
\BIT
\item
$\ba = (1, s, s^2, s^3, \ldots )$ where $|s|\leq 1$, so $\ba^{-1} = ( 1, -s )$;
\item 
$X_t$ is IID \emph{Bernoulli$(p)$-Gaussian}
$p N(0,1) +(1-p) \delta_0$.
\EIT
Consider as initial approximation 
$\widetilde{\ba} = \be_0 = ( 1, 0, 0,0, \ldots)$,
and the resulting 
fully specified population
optimization problem, 
with parameter tuple $(p,\ba,\tilde{\ba})$:
 \[
 \begin{array}{ll}
 \underset{\bw }{\mbox{minimize}}   & \Expect|(\bw*Y)_0|\\
\mbox{subject to}  & \bw_0 =1.
\end{array}
\tag{${P}_1(e_0)$ }
\label{bd:population}
  \]
Let $\bw^\star$ denote the (or simply some) solution.
Define the threshold
\[
p^\star  = 1-|s| .
\]
The property \ERP experiences a phase transition at $p = p^\star$:
\BIT
    \item provided $p<p^\star$, then $\bw^\star$ is uniquely defined and equal to $\ba^{-1}$ up to shift and scaling ; and
    \item provided $p> p^\star$, then $\bw^\star$ is not $\ba^{-1}$ up to shift and scaling .
\EIT
  \end{theorem}

  \begin{figure}[!tb]
\centering
\includegraphics[width=.8\textwidth]{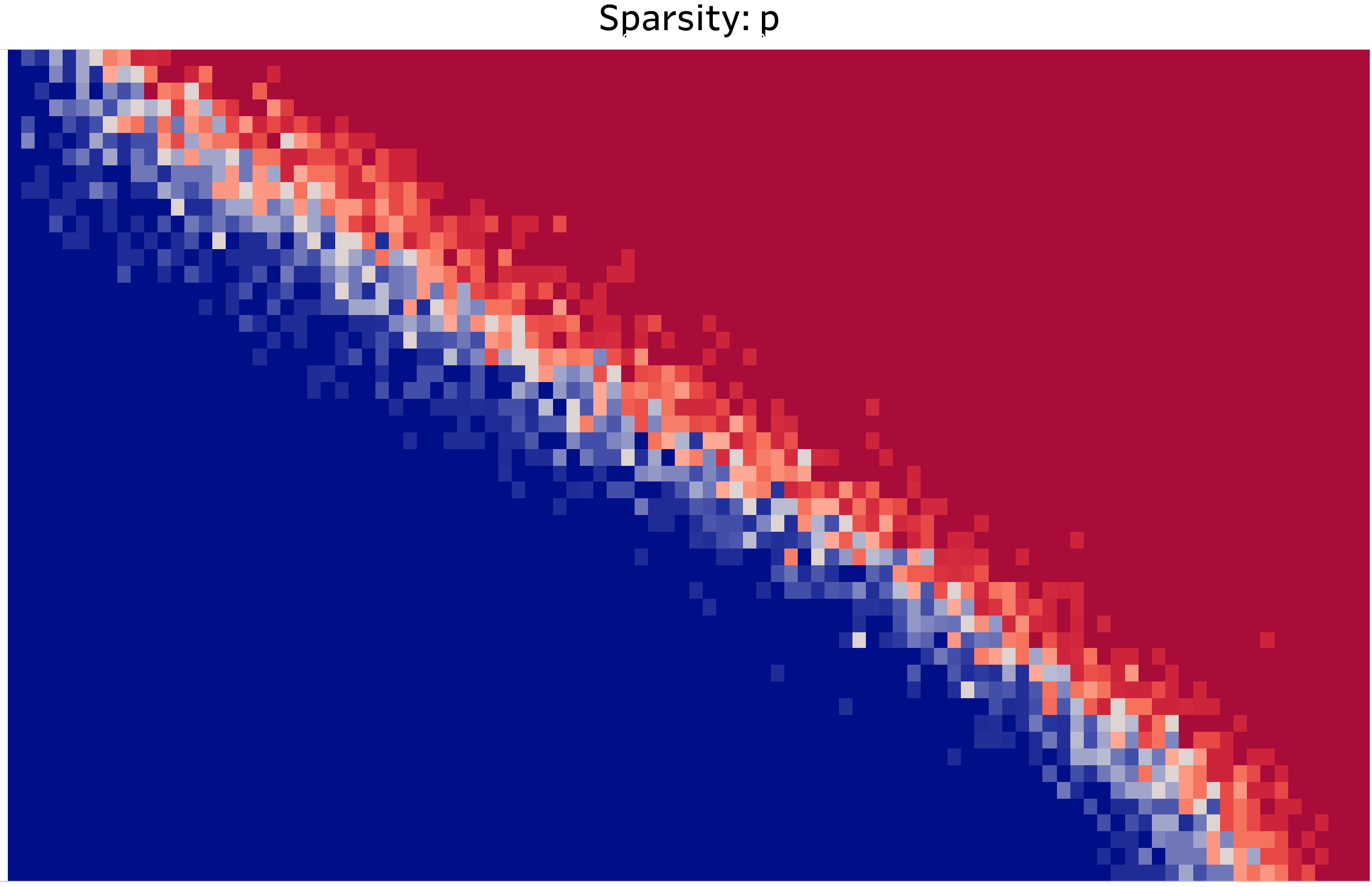}
\caption{Finite-sample phase transition diagram; ground truth filter is $\bw^\star =  ( 0, 1, -s )$ with $T=200$. Horizontal axis:  sparsity level $p$, ranging from $0.01$ to $0.99$; Vertical axis:  filter parameter $|s|$, ranging from $0$ to $0.99$.  The shaded attribute indicates the observed fraction of  
successful experiments at the given parameter combination from deep blue for 1.0 down to deep red for 0.0. Thus, the red region indicates failure of recovery and the blue region indicates success. The transition from success to failure is quite abrupt.}
\label{fig-PT-finite}
\end{figure}

We can empirically  observe that the population phase transition described in Theorem 1 
describes accurately the situation with finite-length signals.   The setting of Theorem 1,
defines a phase space by the numbers $(p,s)$. We can
sample the phase space according to a grid and then at each grid point, conducting a sequence of experiments like so:
\BIT
\item sample a realization of synthetic data $Y = \ba \star X$ according to the stochastic signal;
\item extract a window  $\by^{[N]}$ of size $N$ from within 
each generated $Y$; and
\item solve the resulting finite-$N$ optimization problem:
\[
 \begin{array}{ll}
 \underset{\bw }{\mbox{minimize}}   & \frac{1}{N}\|\bw*Y\|_1\\
\mbox{subject to}  & \bw_0 =1.
\end{array}
\tag{${P}^N_1(e_0)$ }
\label{bd:finite}
  \]
\EIT
Tabulating the fraction of instances with numerically precise recovery of the correct underlying inverse filter $\ba^{-1}$ 
and sparse signal $X$ across
grid points, we can make a heatmap of empirical success probability.

Specifically, we numerically run $20$ independent experiments, and choose an accuracy $\epsilon = 10^{-3}$ and count the number of successfully accurate recovery where the convex optimization solution $w$ satisfy $\|w^\star -w\|_1/\|w^\star\|_\infty<\epsilon$.

We do this in Figure~\ref{fig-PT-finite}; the reader will see there an empirical phase transition curve, produced by a logistic-regression calculation of the location in $p$ where $50\%$ success probability is achieved. 
We observe empirical behavior entirely consistent 
with $p^\star  = 1-|s|$.

\paragraph{Phase transition for sparse blind deconvolution with general filter: upper bound from deltaness discrepancy}
 {Consider now a more general situation where $\ba$ is a quite general filter, and $X$ is Bernoulli Gaussian. 
Our formal assumptions are:}

{\begin{enumerate}
 \item  [{\bf A1}]
$\ba\in l_1(\integers)$ is invertible:  $(\mathcal{F} \ba)(\widetilde{w}) \neq 0$, $\forall \widetilde{w}\in T$; thus $\ba^{-1}$ exists in $\ell_1(\integers)$.
\item [{\bf A2}] $X = (X_t)_{t \in \integers}$ is IID with marginal distribution  $pN(0,1) + (1-p) \delta_0$; and
\item [{\bf A3}] $Y = (Y_t)_{t \in \integers}$ is a linear process
obeying $Y = \ba \star X$.
\end{enumerate}}

{Consider the convex optimization problem 
 \[
 \begin{array}{ll}
 \underset{\bw \in l_1(\integers)}{\mbox{minimize}}   & \Expect|(\bw*Y)_0|\\
\mbox{subject to}  
& \langle \widetilde{\ba},  \bw^{\dagger} \rangle =1,
\end{array}
\tag{${P}_1(\widetilde{\ba})$ }
  \]
and let $\bw^\star$ denote any solution of the optimization problem. }

\newcommand{\tba}{\widetilde{\ba}}
\newcommand{\mgap}{\mbox{Gap}}
\newcommand{\lOneZ}{l_1(\integers)}
\newcommand{\lInfZ}{l_\infty(\integers)}
 {Our results establish the existence of a general
phase transition phenomenon and a precise quantification of it,
in terms of a {\it phase transition functional} .}

 {Define the following {\bf deltaness discrepancy} $\Delta(\bv)$:
\[
 \Delta(\bv) =  \inf_{\alpha,k} \frac{\| \bv - \alpha \delta_k \|_\infty}{\|\bv\|_\infty}.
\]
}

 {This discrepancy  indeed obeys
$\Delta(\be_0) = 0$. However, it is not restricted only to
deltaness concentrating at the origin; in fact,
$\Delta( \alpha \cdot \be_k) = 0$ for
each $\alpha > 0$, $k \in \integers$.
Also, $0 \leq  \Delta(\bv) \leq 1$.
At the other extreme,
$\Delta([...,0,0,1/2,1/2,0,0,...]) = 1$. 
Finally,
$\Delta$ is a continuous function on $\lOneZ$:
so that, 
setting $\eps(\bv,\bw) \equiv \frac{\|v-w\|_\infty}{\|w\|_\infty}$,
\[
   | \Delta(\bv) - \Delta(\bw) | \leq \frac{2 \cdot \eps}{1-\eps}.
\]
}
Indeed, we can give a more explicit form
for $\Delta$;  let $\widetilde{\be}:= \widetilde{\ba}*\ba^{-1}$,
\[
 \Delta(\tbe) = \frac{|\tbe|_{(2)}}{|\tbe|_{(1)}},
\]
where $|\tbe|_{(1)}$ denotes the largest entry
in $\tbe$ and $|\tbe|_{(2)}$ the second largest.
$\frac{|\tbe|_{(2)}}{|\tbe|_{(1)}}$ 
is a kind of multiplicative gap functional,
measuring the extent to which the second-largest entry in $\tbe$
is small compared to the largest entry. It is a natural measure of closeness 
between our approximate Kronecker delta $\tbe=\widetilde{\ba} * \ba^{-1}$ 
and true Kronecker delta  $\be_0$. 

	\begin{theorem}
	[ { Population (large-$N$) phase transition}: upper bound from {\bf deltaness discrepancy}]
	\label{theorem:AltPopulationPhaseTransition}

 {Under assumptions (A1)-(A3),}
\BIT
\item  {There is a functional
$\Pi^*$ defined on tuples $(\ba,\tba)$ in 
$\ell_1(\integers) \times \ell_1(\integers)$
with the property that, for $p > \Pi^*(\ba,\tba)$,
every solution of the optimization problem ${P}_1(\widetilde{\ba})$
is exactly the correct answer $\ba$, up to lateral shift and scaling.}

\item  {We have the upper bound:
\begin{equation}
    \label{eq:coro-upperbnd}
\Pi^*(\ba,\tba) \leq 1-\Delta(\tba*\ba^{-1}), \qquad (\ba,\tba) \in \ell_1(\integers) \times \ell_1(\integers).
\end{equation}
}
In words, the theorem is stating that {\it the less accurate the initial approximation $\widetilde{\ba} \approx \ba$, the greater we rely on sparsity of $X$ to allow exact recovery.}


\EIT
\end{theorem}

 {The reader might compare this with our earlier example
in Theorem \ref{thm: single root PT}, the special case of
geometrically decaying filters. 
In that example $\ba = (\dots,0,0,1,s,s^2,\dots)$ while
$\tba = (\dots,0,0,1,0,0,\dots)$. $\Delta(\tba*\ba^{-1})= \Delta((\dots,0,0,1,s,s^2,\dots)) = |s|$. In short, the upper bound $\Pi^* = 1- |s|$ deriving from this general viewpoint agrees precisely with the the exact answer $p^* = 1- |s|$ given in that earlier Theorem. }

When does equality hold in (\ref{eq:coro-upperbnd})?
For a vector $\bv$, let $\bv'$ denote the same vector,
except the largest-amplitude  entry is replaced by $0$. We will see
that a sufficient condition for equality is:
\[
   \Delta(\tbe') \leq \Delta(\tbe).
\]
Put another way, equality holds if 
\BEQ \label{eq:suffcondsharp}
|\tbe|_{(3)} \leq  \Delta(\tbe) \cdot |\tbe|_{(2)}.
\EEQ
The set of situations where this occurs is ample but not overwhelming.
It has relative
Lebesgue measure at least  $\Delta$. In the situation covered by  Theorem \ref{thm: single root PT},
$|\tbe|_{(3)}  = s^2$, $ |\tbe|_{(2)} = s$,
$\Delta = s$, and so equality holds in (\ref{eq:suffcondsharp}).
Hence, the general-filter result of Theorem 
\ref{theorem:AltPopulationPhaseTransition} along with the sufficient condition (\ref{eq:suffcondsharp})
imply as a special case
Theorem \ref{thm: single root PT}.

 {Now we state the following corollary to show that there is a substantial
region in the space of $(\ba,\tba)$ pairs where a meaningful 
sparsity-accuracy of initialization tradeoff exists, such that
 any sufficiently accurate initial guess results in
exact recovery \-- provided that the sparsity of the underlying object $X$ exceeds a threshold.}

\begin{coro}
 {Let $r \in (0,1/2)$ and suppose 
that assumptions (A1)-(A3)
hold with $p > \frac{r}{1-r}$. 
Normalize the problem so that 
$\| \ba^{-1} \|_\infty =1$.}

 {Consider $\tba$ in the $\ell_1$ metric ball
(\ref{eq:lonemetricball}) of radius $r$
about $\ba$:
\BEQ \label{eq:lonemetricball}
    \| \tba - \ba \|_1 \leq r .
\EEQ
}
 {Every solution of the optimization problem $P_1(\tba)$ 
achieves exact recovery of $\ba$ up to shift and rescaling.}
\end{coro}

 {\textbf{Proof}:
Since 
\[
  \| \tbe - \be_0\|_\infty \leq \| \ba^{-1} \|_\infty \cdot \| \tba-\ba\|_1 = r,
\]
we have $|\tbe|_{(1)} \geq 1-r$ and $|\tbe|_{(2)} \leq r.$ This 
implies that $\Pi^* \leq \frac{r}{1-r}$. Now apply theorem
\ref{theorem:AltPopulationPhaseTransition}.
\qed
}

\paragraph{Phase transition for sparse blind deconvolution with general filter from optimization point of view}
Theorem 
\ref{theorem:AltPopulationPhaseTransition} has provided a upper bound of the threshold $p^\star$ for phase transition with clear mathematical meaning.   Now we present the main phase transition theorem with the exact $p^\star$ for blind deconvolution of general inverse filter. This statement is from optimization point of view that connects blind deconvolution problem to a classical projection pursuit problem. 


Let $I$ denote an iid Bernoulli(p) bisequence.
For a bisequence $\bw$ let $ \bw \cdot I$ denote the 
elementwise multiplication
of $\bw$ by $I$.
Define the optimization problem
\[
 \begin{array}{ll}
 \underset{\bw}{\mbox{minimize}}   & \E_I \|\bw \cdot I\|_2\\
\mbox{subject to}  
& \langle   \bv, \bw\rangle =1
\end{array}
\tag{$Q_1(\bv)$ }
  \]

	\begin{theorem}[Population (large-$N$) phase transition: upper and lower bound]
	\label{thm:Population phase transition}


	Under assumptions (A1)-(A3),consider the solution $\bw^\star$ of the convex optimization problem ${P}_1(\widetilde{\ba})$. There is a threshold $p^\star > 0$,
\BIT
\item  $\bw^\star$ is $\ba^{-1}$ up to time shift and rescaling provided $p<p^\star$; and
\item   $\bw^\star$ is not $\ba^{-1}$ up to time shift and rescaling, provided $p> p^\star$.
\EIT

The threshold $p^\star$ obeys
 \[
\frac{p}{1-p}= val(Q_1(\tbe')).
\] 
where $\widetilde{\be}:= \widetilde{\ba}*\ba^{-1}$.
As defined in theorem \ref{theorem:AltPopulationPhaseTransition}, $p^\star = \Pi^*(\ba,\tba) $.

We have an upper bound and lower bound of $val(Q_1(\tbe'))$,
\[
 \frac{p}{\Delta(\tbe)}  \geq val(Q_1(\tbe')) \geq  p \cot\angle(\tbe, \be_0) 
\]
explicitly, the upper bound can be expressed as
\[
1- \tan\angle(\tbe, \be_0)   \leq p^\star \leq 1- \frac{|\tbe|_{(2)}}{|\tbe|_{(1)}} = 1-\Delta(\tbe)
\]
Namely,
\begin{equation}
    \label{eq:coro-upperbnd-lowerbnds}
1- \tan\angle(\tba*\ba^{-1}, \be_0) \leq \Pi^*(\ba,\tba) \leq 1-\Delta(\tba*\ba^{-1}), \qquad (\ba,\tba) \in \ell_1(\integers) \times \ell_1(\integers).
\end{equation}
Additionally, the upper bound 
takes equality if
	\[
   \Delta(\tbe') \leq \Delta(\tbe).
\]

	\end{theorem}
	
We comment that theorem 	\ref{thm:Population phase transition} imply previous two theorems. Clearly theorem \ref{thm: single root PT} is a special case for exponential decay forward filter, and theorem \ref{theorem:AltPopulationPhaseTransition} can also be viewed as the natural upper bound side of the functional
$\Pi^*$  in theorem 	\ref{thm:Population phase transition}.



\subsection{Main Result $2$: Finite Observation Window, Finite-length Inverse}

With a finite observation window of length $N$ we (surprisingly) still can have exact recovery, starting as soon as $N \geq \Omega(k\log(k))$.

\paragraph{Finite-observation phase transition}

Let 
$\ell_1^k$
denote the collection of bisequences vanishing outside a centered window of radius $k$.

	\begin{theorem}[Finite observation window, finite-length inverse filter]
	\label{thm:fin_guarantee}
	Suppose $Y = \ba \star X$, where:
	\begin{enumerate}
 \item  [{\bf A1}]
$\ba\in l_1(\integers)$ is invertible and the inverse has finite-length:  $(\mathcal{F} \ba)(\widetilde{w}) \neq 0$, $\forall \widetilde{w}\in T$; thus $\ba^{-1}$ exists in $\ell_1(\integers)$; furthermore, $\ba^{-1} \in \ell_1^k$ vanishes off a centered window of length $k$.
\item [{\bf A2}] $X = (X_t)_{t \in \integers}$ is IID with marginal distribution  $pN(0,1) + (1-p) \delta_0$; 
\item [{\bf A3}] $Y = (Y_t)_{t \in \integers}$ is a linear process
obeying $Y = \ba \star X$. Suppose we observe a window $(Y_t)_{t\in \cT}$ of length $N = |\cT|$. 
\end{enumerate}

Consider the convex optimization problem
\BEQ
 \begin{array}{ll}
 \underset{\bw\in \ell_1^k}{  \mbox{minimize} }   & \frac{1}{N}\|\bw*Y\|_{\ell_1^N}\\
 \mbox{subject to}  
& (\widetilde{\ba}* \bw)_0 =1.
\end{array}
\tag{${P}_{1}^{N,k}(\widetilde{\ba})$}
\label{P1_finite}
\EEQ
and let $\bw^\star$  denote any solution of the optimization problem.

	Our result establishes that there exist $\epsilon>0, \delta>0$, so that when the number of observations $N$ satisfies
		\[
N \geq  k\log(\frac{k}{\delta}) ( \frac{C\kappa_\ba}{\epsilon}  )^2,
\]
and the sparsity level $p$ obeys $\frac{1}{N}\leq p < p^\star - \delta_p(N,\epsilon)$,
		then 	with probability exceeding $1-\delta$, $\bw^\star$   is $\ba^{-1}$ up to rescaling and time shift. In the statement, $\kappa_\ba$ is the condition number of the circular matrix with its first column being $\ba$; $C$ is a positive constant independent of $N$ and $k$.
	\end{theorem}

\subsection{Main Result $3$: Stability Guarantee with Finite Length Inverse Filter}
The results so far concern the ideal setting when true inverse filter has a known length $k$ and we use $k$ to set up a correctly matched optimization problem. In practice we do not know $k$ and $k$ might even be infinite.

We can provide practical guarantees even when the inverse filter is an infinite length inverse. To develop these, we must be more technical about the situation. We assume that $\ba$ has a Z-transform having $N_{-}$ roots and $N_{+}$ poles $(s_{i})$ inside the unit circle and we construct a finite length approximation $\bw$ to  $\ba^{-1}$, in fact of length
$r(N_{-}+N_{+})$.
This approximation has error 
$\|\bw*\ba-\be_0\|_2= O(\max_i |s|_{i}^{r}  )$. 

Since the objective value $\E_I\|\bw\cdot I\|_2$ of this approximation $\bw$ is an upper bound of the optimal value of the optimization solution $\bw^\star$, we could use the objective value upper bound  to derive a upper bound for $\|\bw^\star*\ba-\be_0\|_2= O(
\max_i |s|_{i}^{r}  )$ when $p<p^\star$, where the constant of this upper bound is determined by the Bi-Lipschitz constant of the finite difference of objective $\E_I\|\bw\cdot I\|_2-\E_I\|(\be_0)\cdot I\|_2$.

\paragraph{Approximation theory for infinite length inverse filter}
\begin{theorem}[Approximation theory for infinite length inverse filter based on roots of Z-transform]
\label{thm:RootNorm}
Let the finite-length forward filter $\ba$ have a Z-transform with roots inside the unit circle, namely $s_{k}:= e^{-\rho_k + i \varphi_k }$ with $|s_k|<1$ and $\rho_k>0$
for $k \in \{-N_{-}, \ldots -1, 1, \ldots N_{+}\}$.  Let $\cI = \{-N_{-},\ldots, -1, 1, \ldots, N_{+} \}$ as the set of all the possible indexes.
\[
 \begin{array}{ll}
A(z) &=  \sum_{i=-N_{-}}^{N_{+}} a_{i} z^{-i}
\\
&= c_0\prod_{j=1}^{N_{-}} (1- s_{-j} z) \prod_{i=1}^{N_{+}} (1- s_{i} z^{-1});
\end{array}
  \]
here $c_0$ is a constant ensuring $a_0=1$.

Then for a scalar $r$
, we construct an approximate inverse filter $\bw^r$ with Z-transform
\[
 \begin{array}{ll}
W(z) 
&=  \frac{1}{c_0}\prod_{j=1}^{N_{-}} (\sum_{\ell_j=0}^{r-1} s^{\ell_j}_{-j} z^{-\ell_j} ) \prod_{i=1}^{N_{+}} (\sum_{\ell_i=0}^{r-1} s^{\ell_i}_{i} z^{\ell_i} )\\
&=  \frac{1}{c_0}\prod_{j=1}^{N_{-}} (1- (s_{-j} z)^{r})(1- s_{-j} z)^{-1} \prod_{i=1}^{N_{+}}(1- (s_{i}  z^{-1})^{r}) (1- s_{i} z^{-1})^{-1}.
\end{array}
\]
We have
\[
 \begin{array}{ll}
 \|\bw^r*\ba-\be_0\|_2^2  
 &= \sum_{n\in \{1,2,3,\ldots,|\cI| \} }\sum_{ k_1, \ldots, k_n  \in \cI } \prod_{i\in [n]}|s_{k_i}|^{2r }
\end{array}
  \]
as $r \rightarrow \infty$, this converges to zero at an exponential rate, determined by the slowest decaying term,
\[
 \begin{array}{ll}
\|\bw^r*\ba-\be_0\|_2
&=  O( |s|_{(1)}^{r}  ) , \quad r \rightarrow \infty.
\end{array}
  \]
  where $|s|_{(1)}$ is the largest absolute value root.
 \end{theorem}
 
 
 \paragraph{Stability for infinite length inverse filter}
 \begin{theorem}[Stability for infinite length inverse filter]
 \label{thm:bd_stab}
 Let $\ba\in V_{N_-, N_+} $ be a forward filter with all the roots of Z-transform strictly in the unit circle. Let $\bw^\star \in V_{(r-1)N_-, (r-1)N_+} $ be the solution of the convex optimization problem. Let $\bw^r$ be the constructed filter in previous theorem with a uniform vector index $(r, \ldots, r, r,\ldots, r)$.
  Then provided $p<p^\star$, as 
  $r \rightarrow \infty$, 
  \[
 \begin{array}{ll}
  \| \bw^\star*\ba -\be_0\|_2 
   \leq
    O(
 |s|_{(1)}^{r}  ), \quad r \rightarrow \infty.
\end{array}
  \]
  where $|s|_{(1)}$ is the largest absolute value root.
In words, \textbf{the Euclidean distance between $\bw^\star*\ba$ and $\be_0$ converges to zero at an exponential rate as the approximation length is allowed to increase.} 
 \end{theorem} 
 

\subsection{Main Result $4$: Robustness Against Stochastic Noise and Adversarial Noise}
We now extend the previous analysis of an exactly sparse model of $X$, exactly observed, to the more practical setting of approximate sparsity and observation noise. We consider two cases: first, we add stochastic noise as an independent Gaussian linear process, and second, we add adversarial noise with bounded $\ell_\infty$ norm.  

In each scenario, since the noisy objective value at $\be_0$ is an upper bound on the optimal value of the optimization solution $\bw^\star$, we use this upper bound on the objective value  to derive an upper bound of $\|\bw^\star*\ba-\be_0\|_2$ when $p<p^\star$.
The upper bound shows that the distance $\|\bw^\star*\ba-\be_0\|_2$ is bounded by the (appropriately measured)
magnitude of input noise in both cases.

\paragraph{Robustness under stochastic noise }
 \begin{theorem}[Robustness against Gaussian Linear Process noise]
 \label{thm:rand_exact_rob}
 We consider a Gaussian linear process $Z = \sigma \cdot \bb \star G$ where: $\sigma > 0$ denotes the noise level;   $\bb$ is
 a bisequence having  unit $\ell_2$ norm $\|\bb\|_2=1$; and $G$ is a standard Normal iid bisequence. Suppose that we observe:
 \[
\begin{array}{ll}
Y&= \ba* (X+ Z).
\end{array}
\]
 
 \newcommand{\wpsi}{{\bf w}}
 Let $\wpsi^\star$ be the solution of the convex optimization problem in Eq.(4.1);

\label{thm:Bound_Normal_Noise}
\[
 \E_I  \| (\wpsi^\star)\cdot I \|_2 -\E_I  \| (\be_0)\cdot I \|_2 \leq  (1-p) \sigma +
 p(\sqrt{1+ \sigma^2 }-1).
\]
When $p<p^\star$, $\sigma\leq 1$, there exists a constant $C$,
\[
  \| \bw^\star*\ba -\be_0\|_2 = \|\wpsi^\star -\be_0\|_2 \leq C \sigma.
\]
In words, \textbf{the Euclidean distance between $\bw^\star \star \ba$ and $\be_0$ is bounded linearly by the magnitude of stochastic noise}.
\end{theorem}  

\newcommand{\bc}{{\bf c}}
\paragraph{Robustness under adversarial noise }
 \begin{theorem}[Robustness under adversarial noise with $\ell_p$ norm bound]
\label{thm:adv_exact_rob}
Suppose that an adversary chooses a disturbance bisequence
$\bc$ subject to the constraint:	
\[
	\| \bc\|_\infty \leq \eta;
\]
and perturbs the observation process $Y$ via:
	\[
	\begin{array}{ll}
	Y&= \ba* (X+  \bc) .
	\end{array}
	\]
\newcommand{\apsi}{{\bf v}}
Let $\bw^\star$ denote the solution of the convex optimization problem~\ref{eqn:P1}.
Define $\apsi^\star = (\ba*\bw^\star)^{\dagger}$, then $\apsi^\star$ satisfies the following bound on objective difference:
 \[
 \begin{array}{ll}
 \E_I  \| (\apsi^\star)\cdot I \|_2 -\E_I  \| (\be_0)\cdot I \|_2 \leq 
 p\sqrt{\frac{2}{\pi}}(\mathcal{R}(\eta) - 1)+(1-p)\eta.
\end{array}
  \]  
 
Therefore, when $p<p^\star$, there exists a constant $C'$, so that
 \[
 \| \bw^\star*\ba -\be_0\|_2 = \|\apsi^\star -\be_0\|_2  \leq C'  \eta, \forall \eta > 0.
 \]
 In words, \textbf{the Euclidean distance between $\bw^\star \star \ba$ and $\be_0$ is at most proportional to the magnitude of the adversarial noise}.

\end{theorem}  

\textbf{Remark}: Here $\mathcal{R}(\eta)$
is the folded Gaussian mean, for standard Gaussian $G$:
\[
\mathcal{R}(\eta) :=  \sqrt{\frac{\pi}{2}} E_G |\eta+  G| = \exp\{-\eta^2/2\} + \sqrt{\frac{\pi}{2}}\eta \left(1 - 2\Phi\left(-\eta\right) \right).
\]  
$\mathcal{R}(\eta)-1$ is an even function that is monotonically non-decreasing for $\eta \geq 0$ with quadratic upper and lower bound: there exists constants $C_1\leq C_2$,
  \[
C_1 \eta^2 \leq \mathcal{R}(\eta)-1 \leq C_2 \eta^2, \forall \eta.
\]

\newcommand{\bu}{{\bf u}}

\subsection{Non-convex Blind deconvolution initialization method}
Now we study how to get the initial guess $\widetilde{\ba}$.

\paragraph{}
First, let $C_Y$ denote the circular embedding of $Y$;  we can define $\Bar{Y} = (C_Y C_Y^T)^{-1/2}Y$, then from $Y=\ba*X$, we get $C_Y = C_a C_X$, then
\[
C_{\Bar{Y}} =  (C_Y C_Y^T)^{-1/2}C_Y = (C_a C_X C_X^T  C_a^T)^{-1/2}C_a C_X
\]
as $X$ are IID Bernoulli-Gaussian, 
we know approximately 
\[
C_{\Bar{Y}} =  (C_Y C_Y^T)^{-1/2}C_Y = (C_a C_X C_X^T  C_a^T)^{-1/2}C_a C_X \approx (C_a C_a^T)^{-1/2}C_a C_X
\]
We can also define $C_{\Bar{a}}=(C_a C_a^T)^{-1/2}C_a$.

Let $S_{\kappa_1 \times \kappa_2, K}$ be Riemannian manifold of $K$-channel $\kappa_1 \times \kappa_2$ convolution dictionary with row norm $1$. 

Our non-convex algorithm objective is
\[
 \begin{array}{ll}
 \underset{A, X}{\mbox{minimize}}   & \|Y-(A*X+b) \|_2^2+ \lambda \|X\|_1\\
\mbox{subject to}  
& A \in S_{\kappa_1 \times \kappa_2, K},
\end{array}
\label{eqn:NC}
  \]
  And the algorithm is alternating between proximal gradient method on solving $X$ given $A$ with Backtracking line search for updating step size of proximal gradient on $X$, and  projected gradient method on solving $A$ given $X$ with line search for updating the step size of Riemannian gradient on $A$.
\begin{algorithm}[ht]
\LinesNumbered
\SetAlgoLined
Input: $Y$ of size $n_1 \times n_2$  estimated upper bound of kernel size of forward filter $a$: $\kappa_1 \times \kappa_2$, estimated number of kernel channels $K$.\;

Create zero tensor
$A^{\mbox{init}}$ be of size $3\kappa_1 \times 3\kappa_2  \times K$, \;

for $k \in 1, \ldots, K:$\;

random uniformly choose index $i_{1,k},i_{2,k}$; assign
\[
A^{\mbox{init}}[\kappa_1: 2*\kappa_1, \kappa_2: 2*\kappa_2, k] = 
Y[i_{1,k}:i_{1,k}+\kappa_1,i_{2,k}: i_{2,k}+\kappa_2]/\| Y[i_{1,k}:i_{1,k}+\kappa_1,i_{2,k}: i_{2,k}+\kappa_2]\|_2
\]
  \;
  
Let $X^{\mbox{init}}$ be zeros of shape $n_1 \times n_2 \times K$; 
$b^{\mbox{init}}$ be $\mbox{mean}(\Vec{Y})$.  \;

for $\mbox{iter} \in 1, \ldots,  \mbox{MaxIter}:$\;

Alnternating the two steps:\;

Given $A$ fixed, take a descent step on $X$ via proximal gradient descent;
Backtracking for update $X$ and update stepsize $t$;\;

Given $X$ fixed, take a Riemannian gradient step on $A$, use line search for updating the stepsize of Riemannian gradient on $A$. \;

Given $A, X$ fixed, update the bias $b$ \;

if $\|A -A^{\mbox{prev}\|_2}<\epsilon$ and $\|X -X^{\mbox{prev}\|_2}<\epsilon$, stop.
 \caption{Non-convex Blind deconvolution initialization method}
\end{algorithm}

\begin{algorithm}[ht]
\LinesNumbered
\SetAlgoLined
Input: $Y$ of size $n_1 \times n_2$, estimated upper bound of kernel size of forward filter $a$: $\kappa_1 \times \kappa_2$.\;

\eIf{Initialization of forward filter $\Tilde{a}$ not provided}{
Let $\Tilde{a}$ be the solution of the previous algorithm with estimated number of kernel channels $K=1$ and forward kernel size $\kappa_1 \times \kappa_2$. 
}
{Use initialization guess of forward filter $\Tilde{a}$ \;}
Find inverse filter $\bw$ by
\[
 \begin{array}{ll}
 \underset{\bw \in l_1(\integers)}{\mbox{minimize}}   & \|\bw*Y\|_1\\
\mbox{subject to}  
& \langle \widetilde{\ba},  \bw^{\dagger} \rangle =1,
\end{array}
\tag{${P}_1(\widetilde{\ba})$ }
\label{eqn:P1}
  \]
  \;
 \caption{Full blind deconvolution algorithm}
\end{algorithm}

\section{Numerical Experiments}




Now we look at numerical experiments of $2-$D image blind deconvolution for different choices of $X^\star$ and $a^{-1,\star}$.

\clearpage
\begin{figure} 
\centering
\includegraphics[width=.2\textwidth]{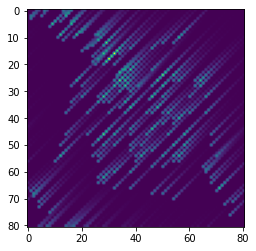}
\includegraphics[width=.2\textwidth]{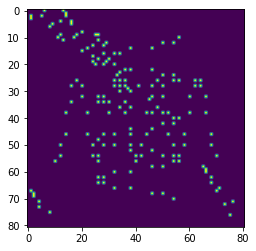}
\includegraphics[width=.2\textwidth]{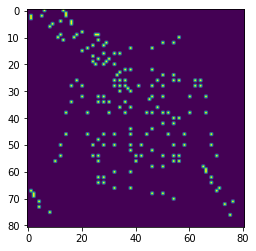}
\includegraphics[width=.2\textwidth]{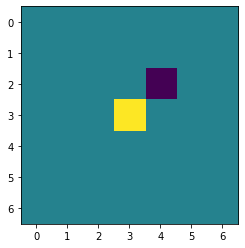}
\caption{
Sparse point: $Y$, solved $X$, true $X$, solved $f$. 
}
\label{fig-2D-Sparse}
\end{figure}
\BIT
\item $X$: $2-$D Bernoulli Gaussian IID of size $80 \times 80$.
\item $f^\star$: $1-$D filter $(1, -s, 0)$, $s=0.9$ centered and rotated by $45$ degree.
\item $Y =a^\star*X$, where the filter inverse $a^\star = f^{\star,-1} $ is defined by discrete Fourier transform on $80 \times 80$ grid.  
\item
The plots in order, first row:
$Y = a^\star * X^\star$, $X^{\mbox{cvx}}$,  $X^\star$, $f^{\mbox{cvx}}= f^\star$.  
\EIT 
\clearpage

\begin{figure} 
\centering
\includegraphics[width=.2\textwidth]{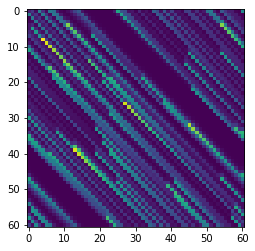}
\includegraphics[width=.2\textwidth]{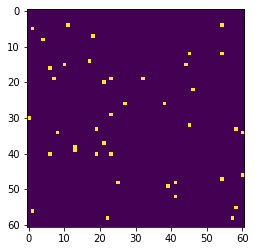}
\includegraphics[width=.2\textwidth]{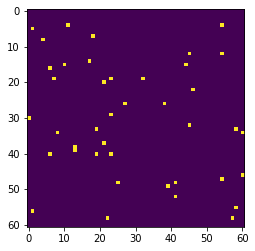}
\includegraphics[width=.2\textwidth]{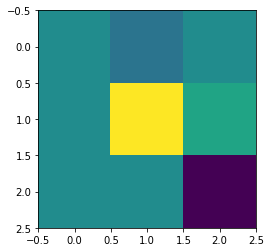}
\centering
\includegraphics[width=.2\textwidth]{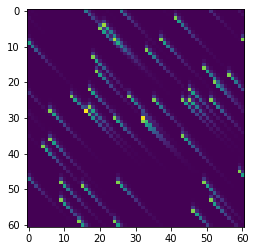}
\includegraphics[width=.2\textwidth]{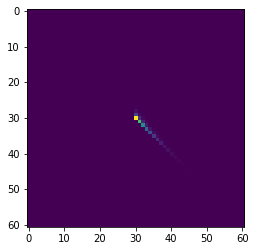}
\includegraphics[width=.2\textwidth]{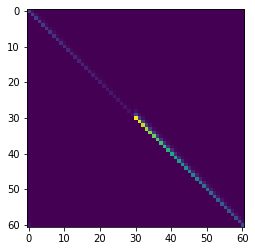}

\centering
\includegraphics[width=.2\textwidth]{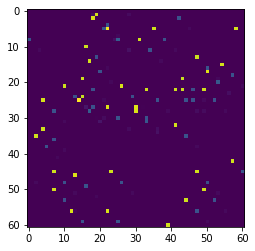}
\includegraphics[width=.2\textwidth]{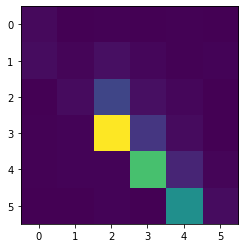}
\includegraphics[width=.2\textwidth]{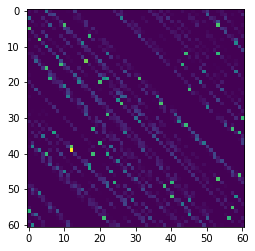}
\includegraphics[width=.2\textwidth]{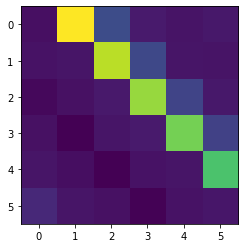}
\caption{Sphere: $Y$, solved $X$, true $X$, solved $f$;
previous $X$, previous $a$ as initialization, $a$ truth;
$X$ solved by non-convex alternating formulation, $a$ solved by non-convex alternating  formulation.
}
\label{fig-2D-Sphere-BD}
\end{figure}

\BIT
\item $X$: $2-$D Bernoulli Gaussian IID of size $60 \times 60$.
\item $f^\star$: $1-$D filter $(1, -s, 0)$ with $s=0.99$ centered and rotated by $45$ degree.
\item $Y =a^{\star}*X$, where the filter inverse is defined by discrete Fourier transform on $60 \times 60$ grid.
\item
$a^{\mbox{init}}$ = $f^{\mbox{init},-1}$, where $f^{\mbox{init}}$ is $1-$D filter $(1, -s_{\mbox{init}}, 0)$ with $s_{\mbox{init}}=0.8$ centered and rotated by $45$ degree.
\item $X^{\mbox{init}} = a^{\mbox{init}} * X$ is the initial guess $a^{\mbox{init}}$ convolve with $X^\star$.
\item $a^{\mbox{ncvx}}$ is the forward solved by alternating non-convex algorithm.
\item $X^{\mbox{ncvx}}$ is the initial guess of $X$ solved by alternating non-convex algorithm.
\item
The plots in order, first row:
$Y = a^\star * X^\star$, $X^{\mbox{cvx}}$,  $X^\star$, $f^{\mbox{cvx}}= f^\star$.  
\item
second row:
$X^{\mbox{init}}$, $a^{\mbox{init}}$, $a^{\star}$;
\item
third row:
$X^{\mbox{nc, init}}$, $a^{\mbox{nc,init}}$, $X^{\mbox{nc}}$, $a^{\mbox{nc}}$;
\EIT 

\clearpage

\begin{figure} 
\centering
\includegraphics[width=.2\textwidth]{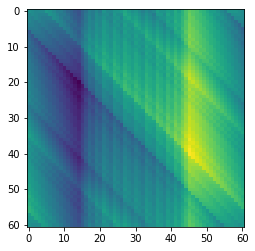}
\includegraphics[width=.2\textwidth]{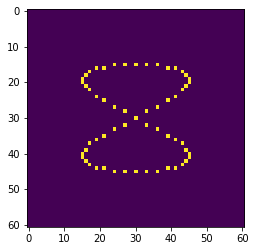}
\includegraphics[width=.2\textwidth]{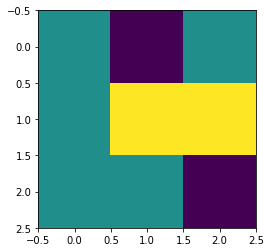}
\includegraphics[width=.2\textwidth]{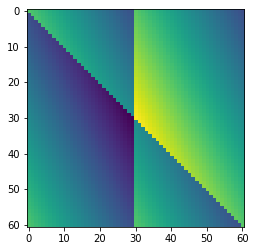}
\centering
\includegraphics[width=.2\textwidth]{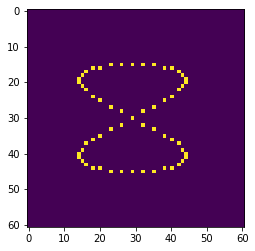}
\includegraphics[width=.2\textwidth]{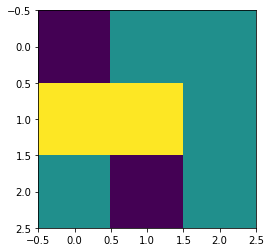}
\centering
\includegraphics[width=.2\textwidth]{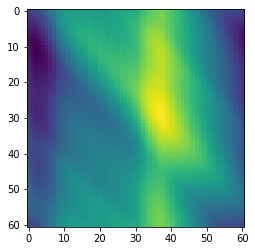}
\includegraphics[width=.2\textwidth]{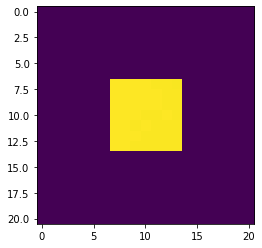}
\caption{Sphere: $Y$, solved $X$, true $X$, solved $f$;
previous $X$, previous $a$ as initialization, $a$ truth;
$X$ solved by non-convex alternating formulation, $a$ solved by non-convex alternating  formulation
}
\label{fig-2D-Lis-new2}
\end{figure}

\BIT
\item $X$: $2-$D centered shaped $8$ with frame size $15$ on size $60 \times 60$ grid.
\item $f^\star$: convolution of two $1-$D filters: 
$(1, -s_1, 0)$ with $s_1=0.9995$ centered and rotated by $45$ degree and 
$(1, -s_2, 0)$ with $s_2=0.9995$ centered and rotated by $0$ degree.
\item $Y =a^{\star}*X$, where the filter inverse is defined by discrete Fourier transform on $60 \times 60$ grid.
\item $a^{\mbox{nc}}$ is the forward solved by alternating non-convex algorithm.
\item $X^{\mbox{nc}}$ is the initial guess of $X$ solved by alternating non-convex algorithm.
\item
The plots in order, first row:
$Y = a^\star * X^\star$,   $X^\star$, $ f^\star$, $a^\star$.  
\item
second row:
$X^{\mbox{cvx}}$, $f^{\mbox{cvx}}$ ,$X^{\mbox{nc}}$, $a^{\mbox{nc}}$;

\EIT 
\clearpage

\begin{figure} 
\centering
\includegraphics[width=.2\textwidth]{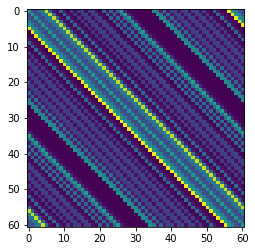}
\includegraphics[width=.2\textwidth]{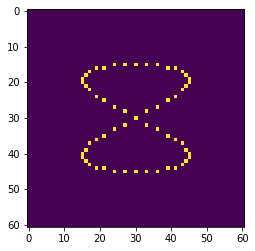}
\includegraphics[width=.2\textwidth]{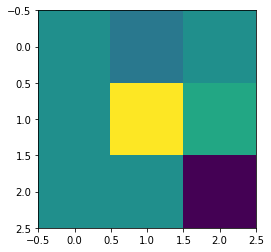}
\centering
\includegraphics[width=.2\textwidth]{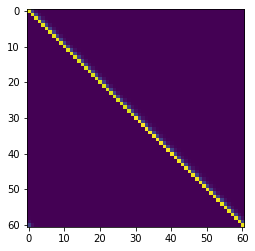}
\includegraphics[width=.2\textwidth]{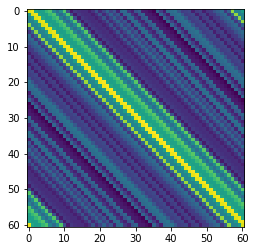}
\includegraphics[width=.2\textwidth]{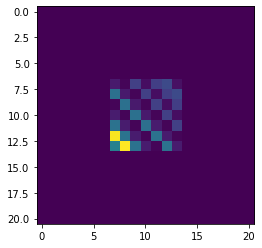}
\caption{ $Y$, solved $X$, true $X$, solved $f$;
previous $X$, previous $a$ as initialization, $a$ truth;
$X$ solved by non-convex alternating formulation, $a$ solved by non-convex alternating  formulation.
}
\label{fig-2D-Lis-new1}
\end{figure}

\BIT
\item $X$: $2-$D centered shaped $8$ on size $60 \times 60$ grid.
\item $f^\star$: convolution of two $1-$D centered filters:
$(1, -s_1, 0)$, $s_1 = 0.999$, rotated by $45$ degree, .
$(1, -s_12, 0)$, $s_2 = 0.9$, rotated by $0$ degree.
\item $Y =f^{\star,-1}*X$, where the filter inverse is defined by discrete Fourier transform on $60 \times 60$ grid.
\item $a^{\mbox{nc}}$ is the forward solved by alternating non-convex algorithm.
\item $X^{\mbox{nc}}$ is the initial guess of $X$ solved by alternating non-convex algorithm.
\item
The plots in order, first row:
$Y = a^\star * X^\star$,   $X^\star$, $ f^\star$, $a^\star$.  
\item
second row:
$X^{\mbox{cvx}}$, $f^{\mbox{cvx}}$ ,$X^{\mbox{nc}}$, $a^{\mbox{nc}}$;

\EIT 
\clearpage

\begin{figure} 
\centering
\includegraphics[width=.2\textwidth]{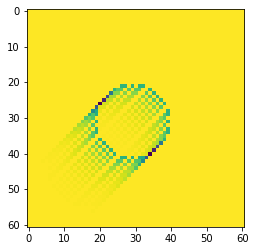}
\includegraphics[width=.2\textwidth]{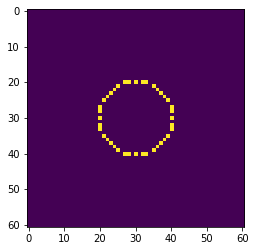}
\includegraphics[width=.2\textwidth]{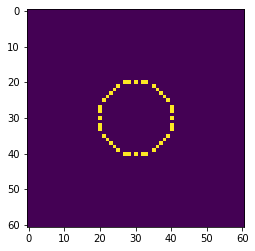}
\includegraphics[width=.2\textwidth]{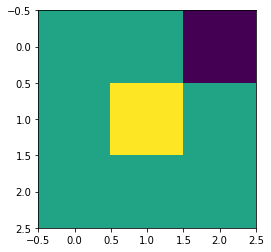}
\centering
\includegraphics[width=.2\textwidth]{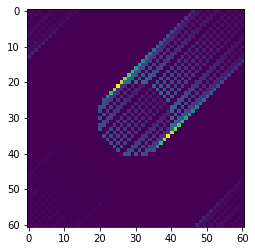}
\includegraphics[width=.2\textwidth]{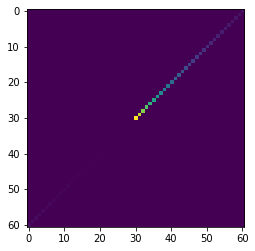}
\includegraphics[width=.2\textwidth]{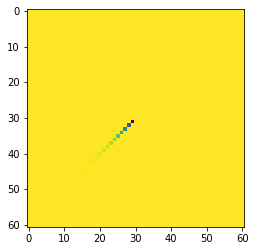}

\centering
\includegraphics[width=.2\textwidth]{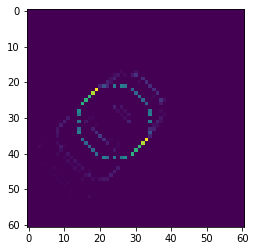}
\includegraphics[width=.2\textwidth]{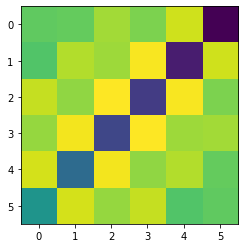}
\caption{
Sphere: $Y$, solved $X$, true $X$, solved $f$;
previous $X$, previous $a$ as initialization, $a$ truth;
$X$ solved by non-convex alternating formulation, $a$ solved by non-convex alternating  formulation.
}
\label{fig-2D-Sphere-BD2}
\end{figure}

\BIT
\item $X$: $2-$D circle centered with radius $10$, on size $60 \times 60$ grid.
\item $f^\star$: $1-$D filter $(1, -s, 0)$ centered and rotated by $45$ degree.
\item $Y =f^{\star,-1}*X$, where the filter inverse is defined by discrete Fourier transform on $60 \times 60$ grid.
\item $a^{\mbox{ncvx}}$ is the forward solved by alternating non-convex algorithm.
\item $X^{\mbox{ncvx}}$ is the initial guess of $X$ solved by alternating non-convex algorithm.
\item
The plots in order, first row:
$Y = a^\star * X^\star$, $X^{\mbox{cvx}}$,  $X^\star$, $f^{\mbox{cvx}}= f^\star$.  
\item
second row:
$X^{\mbox{init}}$, $a^{\mbox{init}}$, $a^{\star}$;
\item
third row:
$X^{\mbox{nc}}$, $a^{\mbox{nc}}$, $f^{\mbox{nc}}$;
\EIT

\clearpage

\begin{figure} 
\centering
\includegraphics[width=.2\textwidth]{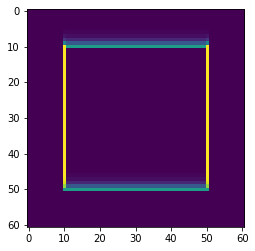}
\includegraphics[width=.2\textwidth]{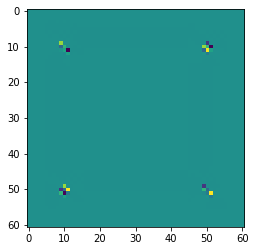}
\includegraphics[width=.2\textwidth]{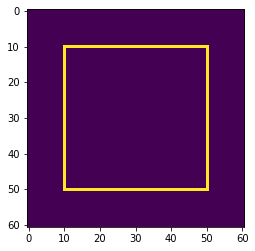}
\includegraphics[width=.2\textwidth]{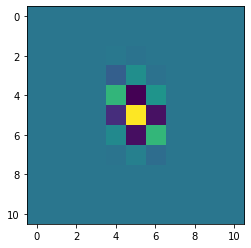}
\caption{Failure mode-square, TV norm: $Y$, solved $X$, true $X$, solved $f$,  }
\label{fig-2D-Square}
\end{figure}

\BIT
\item $X$: $2-$D centered shaped square with frame size $40But $ on size $60 \times 60$ grid.
\item $f^\star$: convolution of two $1-$D filters: 
$(1, -s_1, 0)$ with $s_1=0.9995$ centered and rotated by $45$ degree and 
$(1, -s_2, 0)$ with $s_2=0.9995$ centered and rotated by $0$ degree.
\item $Y =a^{\star}*X$, where the filter inverse is defined by discrete Fourier transform on $60 \times 60$ grid.
\item $a^{\mbox{nc}}$ is the forward solved by alternating non-convex algorithm.
\item $X^{\mbox{nc}}$ is the initial guess of $X$ solved by alternating non-convex algorithm.
\item
The plots in order, first row:
$Y = a^\star * X^\star$,   $X^\star$, $ f^\star$, $a^\star$.  
\item
second row:
$X^{\mbox{cvx}}$, $f^{\mbox{cvx}}$ ,$X^{\mbox{nc}}$, $a^{\mbox{nc}}$;

\EIT 
\clearpage

\begin{figure} 
\centering
\includegraphics[width=.2\textwidth]{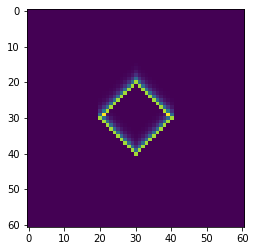}
\includegraphics[width=.2\textwidth]{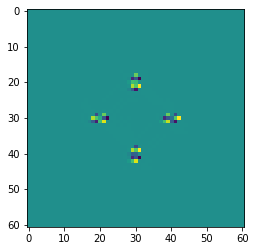}
\includegraphics[width=.2\textwidth]{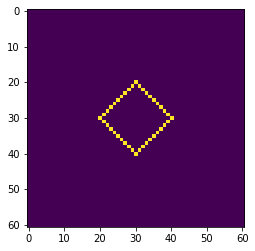}
\includegraphics[width=.2\textwidth]{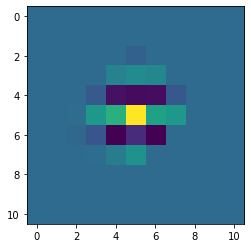}
\caption{Failure mode-diamond, TV norm: $Y$, solved $X$, true $X$, solved $f$,  }
\label{fig-2D-diamond}
\end{figure}

\clearpage





\begin{figure} 
\centering
\includegraphics[width=.2\textwidth]{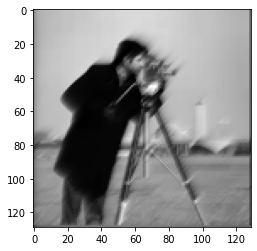}
\includegraphics[width=.2\textwidth]{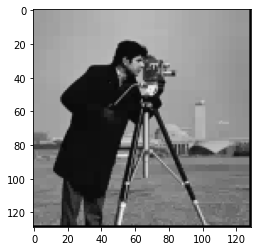}
\includegraphics[width=.2\textwidth]{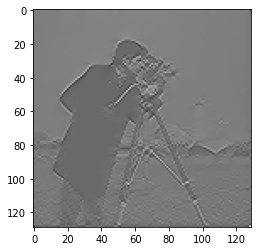}
\includegraphics[width=.2\textwidth]{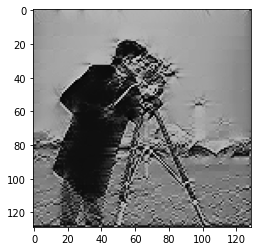}
\caption{`Cameraman', using TV norm and weighted Haar L1 norm ($5$ levels with weights $2^j, j=0, \ldots 4$: $Y$, true $X$, TV norm solved $X$,Weighted Haar L1 norm solved $X$, true $f$, TV norm solved $f$,Weighted Haar L1 norm solved $f$ ,  }
\label{fig-2D-Camera}
\end{figure}
\BIT
\item $X$: camera man figure, on size $128 \times 128$ grid.
\item $f^\star$: $1-$D filter $(1, -s, 0)$ centered and rotated by $45$ degree.
\item $Y =f^{\star,-1}*X$, where the filter inverse is defined by discrete Fourier transform on $60 \times 60$ grid.
\item
$X^{\mbox{TV}}$ is the solution of convex problem with TV norm objective.
\item
$X^{\mbox{Haar}}$ is the solution of convex problem use weighted Haar L1 norm with $5$ levels as objective, where the weighted Haar L1 norm  weights  level  $j$ coefficient by with weights $2^j, j=0, \ldots 4$.
\item
The plots in order, first row:
$Y = a^\star * X^\star$, $X^\star$,$X^{\mbox{TV}}$,  $X^{\mbox{Haar}}$.  
\EIT 
\clearpage

\section{ Techinical Overview and Proof Sketch for Phase Transition Theorems }



\subsection{Main Result $1$: Population (Large-$N$) Phase Transition }

\paragraph{Sketch of proof ideas for Theorem~\ref{thm:Population phase transition}}

Here we highlight some of the key ideas in the proof:
\BIT
\item
\textbf{Change of variable.}
Rewrite the population version of our convex sparse blind deconvolution problem, with the population objective 
$\Expect\frac{1}{N} \| \bw*Y\|_{\ell_1(\cT)} =\Expect|(\bw*Y)_0|=\Expect|(\bw*\ba*X)_0|=\Expect_X|\langle X, (\bw*\ba)^{\dagger} \rangle|$ due to  the ergodic property of stationary process and shift invariance, and $(\widetilde{\ba}*\bw)_0=((\widetilde{\ba}*\ba^{-1})*(\ba*\bw))_0  = \langle (\widetilde{\ba}*\ba^{-1})^{\dagger}, \bw*\ba\rangle$, the convex problem becomes 
\[
 \begin{array}{ll}
 \underset{\bw}{\mbox{minimize}}   & \Expect_X|\langle X ,  (\ba*\bw)^{\dagger} \rangle|\\
\mbox{subject to}  
& \langle \widetilde{\ba}*\ba^{-1} , (\ba*\bw)^{\dagger}\rangle =1,
\end{array}
 \]
 \newcommand{\apsi}{{\bf v}}
Let $\apsi$ denote  the time reversed version of $\ba*\bw$: $\apsi := (\ba*\bw)^{\dagger}$, and
  let $\tbe:= \widetilde{\ba}*\ba^{-1}$, then by previous assumptions, 
  $\tbe(0)=1$, $\tbe' = \tbe-\be_0$.
  
 Now we arrive at a simple and fundamental population convex problem: 
 \[
 \begin{array}{ll}
 \underset{\apsi}{\mbox{minimize}}   &\Expect_X|\langle X ,  \apsi  \rangle|\\
\mbox{subject to}  
& \langle \tbe ,  \apsi  \rangle  =1.
\end{array}
  \]
  \item
\textbf{Expectation using Gaussian.  }
Since $X$ follows Bernoulli-Gaussian IID probability model $X_t = I_t G_t$, we nest the expectation over $I_t$ outside the expectation over Gaussian $G_t$, for which we use $E |N(0,1)| = \sqrt{\frac{2}{\pi}}$:
\[
\Expect_X|\langle X ,  \apsi  \rangle|
= \E_I \E_G {|\sum_{t\in \integers}  I_t G_t \apsi(t)| }
= \sqrt{\frac{2}{\pi}} \cdot\E_I \|\apsi\cdot I\|_2
\]
\item
\textbf{KKT condition for $\be_0$.}
Let $\apsi^\star$ denote the solution of the optimization problem:
\[
 \begin{array}{ll}
 \underset{\psi}{\mbox{minimize}}   & \sqrt{\frac{2}{\pi}} \cdot\E_I \|\apsi\cdot I\|_2\\
\mbox{subject to}  
& \langle   \tbe, \apsi\rangle =1
\end{array}
\tag{$Q_1(\tbe)$ }
  \]
  
To prove that $\apsi^\star=\be_0$, i.e. $\be_0$ solves ($Q_1(\tbe)$), we calculate the directional finite difference at $\be_0$. Then $\be_0$ solves this convex problem if the directional finite difference at $\apsi = \be_0$ is non-negative at every direction $\beta$ on unit sphere where $\tbe^T \beta =0$:
	 \[
	 \Expect\|(\be_0+\beta)\cdot I\|_2 - \Expect\|(\be_0)\cdot I\|_2 \geq 0.
	 \]

\item
\textbf{Conditional expectation at one sparse element $X_0$.} 
We decompose the objective into a sum of terms, 
conditioning on whether $I_0 = 1_{\{X_0\neq 0\}}$ is zero or not:
	 \[
	 \begin{array}{ll}
	    \Expect\|(\be_0+\beta)\cdot I\|_2 - \Expect\|(\be(0))\cdot I\|_2
	    &= p (1+\beta_0)  + (1-p) \nabla_{\beta}{\Expect_I[\|(\be_0+\beta)'\|_{\ell_2(I-\{0\})} \mid I_0 =0]} -p\\
	    &= p \beta_0  + (1-p) \Expect_{I'}[\|\beta'\|_{\ell_2(I')} ].
	 \end{array}
	\]
This will be non-negative in case either $\beta(0) > 0$, or else $\beta(0) < 0$ but
	\[
	\frac{p}{1-p} \leq  \frac{\Expect_{I'}\|\beta'\|_{\ell_2(I')}}{ |\beta_0|}
	\]
	for all $\beta $ that satisfy $ \tbe^T \beta =0$.
\item
\textbf{Reduction to $val(Q_1 (\tbe')).$}	
	We normalize the direction sequence $\beta$ so that $\beta(0)= -1$; 
	using $\tilde{e}(0)=1$, we obtain a lowerbound:
\[
 \inf_{\beta(0)=-1, \langle  \tbe, \beta\rangle  =0} \Expect_{I'}\|\beta'\|_{\ell_2(I')}
 =\inf_{ \beta(0)=-1, \beta(0) \tilde{e}(0)-\langle  \tbe', \beta'\rangle  =0} \Expect_{I'}\|\beta'\|_{\ell_2(I')}
= \inf_{\langle   \tbe', \beta'\rangle  =1} \E_{I'} \|\beta'\|_{\ell_2(I')} = val(Q_1 (\tbe'))
\]	
Here $Q_1(\tbe')$ is the optimization problem:
\[
 \begin{array}{ll}
 \underset{\beta \in l_1(\integers)}{\mbox{minimize}}   &\E_{I'} \|\beta'\|_{\ell_2(I')}\\
\mbox{subject to}  
& \langle  \tbe', \beta'\rangle  =1
\end{array}
  \]

\item 
\textbf{The explicit phase transition condition with upper and lower bound.}
We have shown the existence of  $p^\star$ so that for all $p<p^\star$, the KKT condition is satisfied. And we have represented $p^\star$ as the optimal value of a derived optimization problem $val(Q_1 (\tbe'))$. The following lemma finds simple upper and lower bounds for $p^\star=val(Q_1 (\tbe'))$. 
	\begin{lemma}[Explicit phase transition condition with upper and lower bound]
	\label{lemma:Condition phase transition}
The threshold $p^\star$ determined by
 \[
\frac{p}{1-p}= val(Q_1(\tbe')).
\]
obeys an upper bound and lower
\[
 \frac{p}{\|\tbe'\|_\infty}  \geq val(Q_1(\tbe')) \geq  \E_I\|\tbe'\cdot I\|_{2}^{-1}
\]
where $\|\tbe'\|_\infty  = \frac{|\tbe|_{(2)}}{|\tbe|_{(1)}}  $.
	Additionally, the upper bound 
is sharp if and only if
	\[
	\frac{p}{1-p} \leq val(Q_1(\frac{\tbe'}{\|\tbe'\|_\infty}))=val(Q_1(\tbe''))/\|\tbe'\|_\infty
	\]
	therefore, the upper bound 
holds with equality 
	\[
	p^\star= 1- \frac{|\tbe|_{(2)}}{|\tbe|_{(1)}}
	\]
if 
		\[
	p\leq 1- \frac{|\tbe|_{(3)}}{|\tbe|_{(2)}}
	\]
	Therefore, if 
		\[
	\frac{|\tbe|_{(3)}}{|\tbe|_{(2)}} \leq \frac{|\tbe|_{(2)}}{|\tbe|_{(1)}}
	\]
	then
	\[
	p^\star= 1- \frac{|\tbe|_{(2)}}{|\tbe|_{(1)}}.
	\]
	
	\end{lemma}

The above narrative gives a sketch of our result and its proof. The upper bound $ 1- \frac{|\tbe|_{(2)}}{|\tbe|_{(1)}}$ of $p^\star$ generalized the previous special case of exponential decay filter in theorem~\ref{thm: single root PT} with $p^\star = 1-|s|$.
\item
\begin{lemma}[Geometric lower bound $val(Q_1(\tbe'))$ for the phase transition condition]
\label{lemma: geo lower}
    \[
val(Q_1(\tbe'))
\geq p \cot\angle(\tbe, \be_0)
\]
where
\[
\cot\angle(\tbe, \be_0) = \frac{1}{\|\tbe'\|_2} 
\]
\end{lemma}
\item
Using the technical background provided in theorem~\ref{thm:BGPropLemmaNorm}, we can further provide a tighter upper bound to compute phase transition $p^\star$.

\subsection{Technical Tools:  Landscape of  Expected Homogeneous Function over Bernoulli Support on Sphere}

Let 
\[
V_k(\bu) := \frac{\E_J\|\bu\cdot J\|_2^k}{\|\bu\|^k_{2}} 
\]
where $\bu\in \reals^N$, $J$ is a Bernoulli sequence indexed from $1$ to $N$.

\paragraph{Expectation $V_1(\bu):= \E_J \|\bu\cdot J\|_2^{1}$ on sphere.}



\begin{theorem} \label{thm:BGPropLemmaNorm}
Let $\bu \in \reals^N$, then
\[
p \leq V_{1}(\bu) \leq V_{1}(\frac{\sum_{j\in [N]}\pm \be_j}{\sqrt{N}} ) \leq \sqrt{p} .
\]
the lower bound is approached by on-sparse vectors $\bu \in \{ \pm \be_i, i\in [N]\}$, and the upper bound is approached by 
$\bu  \in \{  \frac{1}{\sqrt{N}}\sum_{j\in [N]}\pm \be_j\}$.

Furthermore, all the stationary points of $V_{1}(\bu)$ are $\{\frac{\sum_{i\in J} \pm \be_i}{\sqrt{N_J}}\}$ for different support $J\subset \cT$, where $\{ \pm \be_i, i\in [N]\}$ are the global minimizers, and $\{  \frac{1}{\sqrt{N}}\sum_{j\in [N]}\pm \be_j\}$ are the global maximizers. And for $J$ with $1<N_J<N$, $\{\frac{\sum_{i\in J} \pm \be_i}{\sqrt{N_J}}\}$ are saddle points with value
\[
V_1(\frac{\sum_{i\in J} \pm \be_i}{\sqrt{N_J}}) = \E_I \sqrt{\frac{\sum_{i\in J} 1_{I_i}}{N_J}} = \sum_{j=0}^{N_J} (1-p)^{N_J-j} p^{j} {N_J \choose j}\sqrt{\frac{j}{N_J}}.
\]
\end{theorem}
To support geometric intuition, Figure~\ref{fig-landscape-V1}
visualizes $V_1$ on the $2-$dimensional sphere.

\begin{figure} 
\includegraphics[width=.7\textwidth]{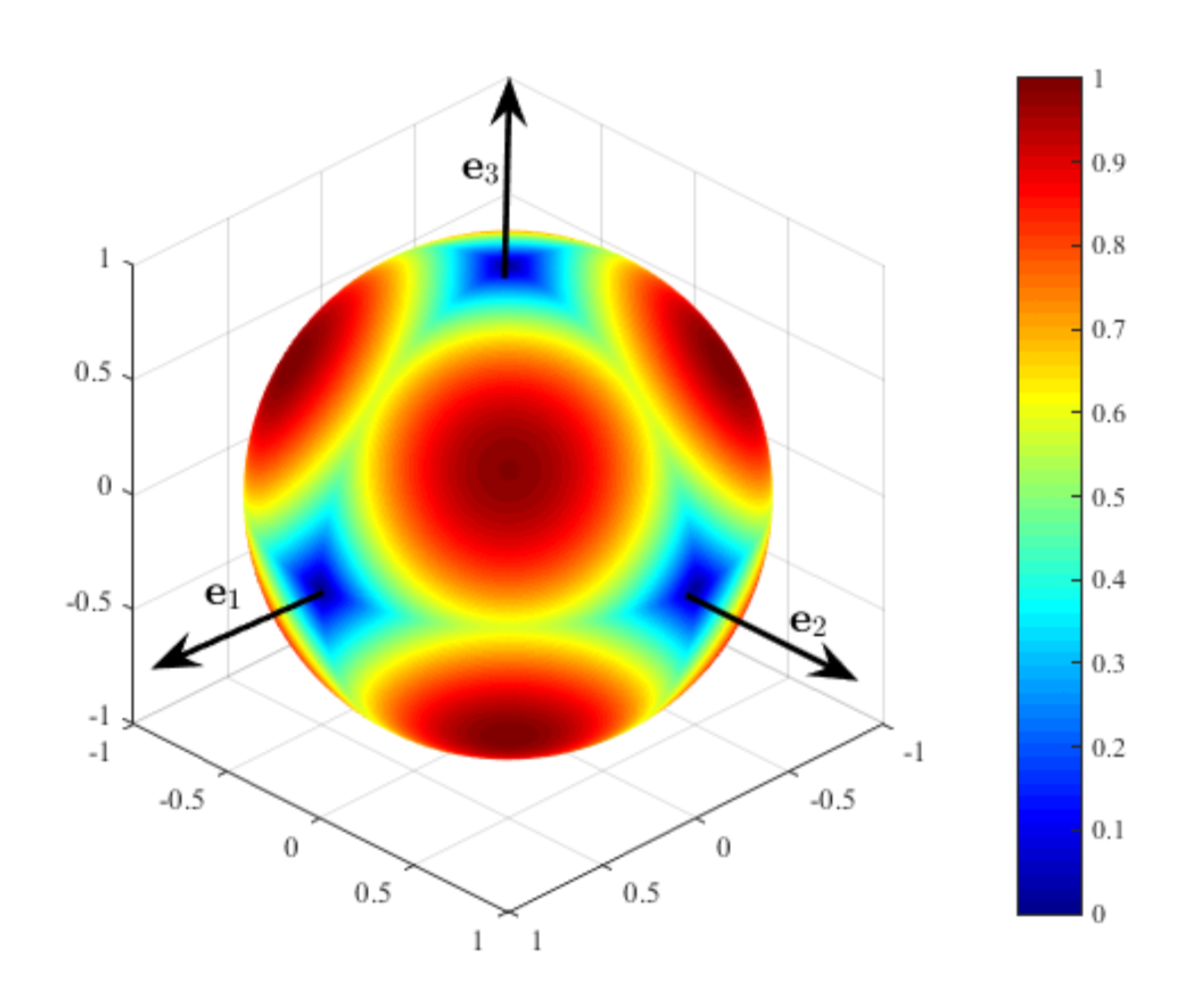}
\caption{The value of $V_1$ on a two-dimensional sphere in $\reals^3$, normalized by the affine transform to send the value in $[0,1].$  
}
\label{fig-landscape-V1}
\end{figure}

\paragraph{Expectation $V_{2k}(\bu):= \E_J \|\bu\cdot J\|_2^{2k}$ on sphere.}
\begin{theorem} \label{thm:BGPropLemmaSub_k}
For $k\geq 2$, let $\bu$ denote a vector in $\reals^N$, then
\[
p^k \leq V_{2k}(\frac{\sum_{j\in [N]}\pm \be_j}{\sqrt{N}} ) \leq V_{2k}(\bu) \leq p.
\]
the upper bound is approached by one-sparse vectors $\bu \in \{ \pm \be_i, i\in [N]\}$, and the lower bound is approached by 
$\bu  \in \{  \frac{1}{\sqrt{N}}\sum_{j\in [N]}\pm \be_j\}$.

Furthermore, all the stationary points of $V_{2k}(\psi)$ are $\{\frac{\sum_{i\in J} \pm \be_i}{\sqrt{N_J}}\}$ for different support $J\subset \cT$, where $\{ \pm \be_i, i\in [N]\}$ are the set of global maximizers, and $\{  \frac{1}{\sqrt{N}}\sum_{j\in [N]}\pm \be_j\}$ are the set of all the global minimizers. And for $J$ with $1<N_J<N$, $\{\frac{\sum_{i\in J} \pm \be_i}{\sqrt{N_J}}\}$ are saddle points with value
\[
V_{2k}(\frac{\sum_{i\in J} \pm \be_i}{\sqrt{N_J}}) = \E_J (\frac{\sum_{i\in J} 1_{I_i}}{N_J})^k = \sum_{j=0}^{N_J} (1-p)^{N_J-j} p^{j} {N_J \choose j}(\frac{j}{N_J})^k.
\]
\end{theorem}

\paragraph{Expectation $V_{-1}(\bu):= \E_J \|\bu\cdot J\|_2^{-1}$ on sphere.}
\begin{theorem}
\label{thm:TaylorSqrtInvExpect}
\beq 
  V_{-1}(\bu) 
  \geq p^{-1/2} .
\eeq
\end{theorem}

\subsection{Exact representation of $val(Q_1)$}
\paragraph{Exact representation from support  analysis} 
Let $|\tbe|$ be the entry-wise absolute value of $\tbe$. We can rank the entries of $|\tbe|$ to be $|\tbe|_{(1)}, |\tbe|_{(2)}, |\tbe|_{(3)},\ldots$, then the entries of $\tbe'$ will be ranked as $|\tbe|_{(2)}, |\tbe|_{(3)},\ldots$. By scaling we have $|\tbe|_{(1)}=1$. We define $\tbe_{S_m}$ as the vector that only keep the entries $ \tbe_{(2)}, \tbe_{(3)},\ldots, \tbe_{(m+1)}$ of $\tbe$, and send the rest of entries to zero. 
\begin{lemma}[Exact representation of $val(Q_1(\tbe'))$ via support decomposition and $V_1$ function]
\label{lemma: tigher upper}
     \[
val(Q_1(\tbe'))= \inf_{ m\in \{1,2,\ldots, n_e-1\}}val(Q_1^{S_m})
\]

\end{lemma}

\EIT

\textbf{Proof: }
Assume that $\tbe$ has $n_e$ non-zero entry on support $S_{\tbe}$, we can rank the absolute value of entries of $\tbe$ to be $|\tbe|_{(1)}, |\tbe|_{(2)}, |\tbe|_{(3)},\ldots,|\tbe|_{(n_e)}$, then the entries of $\tbe'$ will be ranked as $|\tbe|_{(2)}, |\tbe|_{(3)},\ldots,|\tbe|_{(n_e)}$. 

We know the optimal solution of $(\beta')^\star$ of $\inf_{\langle   \tbe', \beta'\rangle  =1} \E_{I'} \|\beta'\|_{\ell_2(I')}$ must have support $S_{\beta^\star}$ that satisfy $S_{\star}\subset S_{\tbe'}$. From the symmetry of objective, we know if $(\beta')^\star$ is $m$ sparse, then $m\leq n_e-1$ and, its support must be on the top $m$ entries $|\tbe|_{(2)}, |\tbe|_{(3)},\ldots,|\tbe|_{(m+1)}$, we call this support $S_m$, and we know the corresponding entries of $(\beta')^\star$ would have the same sign as entries of $\tbe$. 

We can define the $m$ sparse optimization problem for a random support function on the subset of $S_m$: $J_m \subset S_m$. Let $\beta$ be supported on $S_m$ and each entry positive, then
\[
val(Q_1^{S_m}) :=\inf_{ \beta: \sum_{j=1}^{m}  |\tbe|_{(j+1)} \beta_j  =1, \beta_j> 0} \E_{J_m} \|\beta \cdot J_m\|_{2} 
\]
The linear equality constraint comes from
\[
\langle  |\tbe'_{S_m}|, \beta'\rangle =\sum_{j=1}^{m}  |\tbe|_{(j+1)} \beta_j   =1
\]
Then by breaking the case to find minimum over at most $n_e -1$ non-zero entries  to  the minimum of the cases to find minimum over exactly $m$ non-zero entries, and also use sign symmetry, we have 
\[
val(Q_1(\tbe'))= \inf_{ m\in \{1,2,\ldots, n_e-1\}}val(Q_1^{S_m})
\]


\paragraph{Tighter decomposition on $val(Q_1)$}
Now we can prove a more refine upper bound:
\begin{lemma}[Tighter upper bound on $val(Q_1)$]
    \[
val(Q_1(\tbe'))= \inf_{ m\in \{1,2,\ldots, n_e-1\}}val(Q_1^{S_m})
= \inf_{ m\in \{1,2,\ldots, n_e-1\}} [V_1(\tbe'_{S_m}) \cot\angle(\tbe_{S_m}, \be_0)]
\]
where $V_1$ function takes value in $[p, \sqrt{p})$.

More specifically, let the unit vector along the direction of $|\tbe'_{S_m}|$ be 
\[
u^m = |\tbe'_{S_m}|/\|\tbe'_{S_m} \|_2
\]
\BEAS
val(Q_1^{S_m}) 
&=& \cot\angle(\tbe_{S_m}, \be_0)  V_1(|\tbe'_{S_m}|) C(u^m)
\EEAS
where 
\[
C(u^m):=\frac{1}{V_1(u^m)}\inf_{ \beta^m:  \beta^m_j> 0} \frac{V_1(\beta^m)}{\cos\angle(\beta^m, u^m)}  = \frac{1}{E_{J_m}\|u^m\cdot J_m\|_2}\inf_{ \beta^m:  \beta^m_j> 0} \frac{E_{J_m}\|\beta^m\cdot J_m\|_2}{\langle  \beta^m, u^m\rangle }
\]
And we can prove that
\[
C(u^m)= 1
\]
\end{lemma}

\begin{proof}[Proof of lemma~\ref{lemma: tigher upper}]
From 
\[
val(Q_1^{S_m}) :=\inf_{ \beta: \sum_{j=1}^{m}  |\tbe|_{(j+1)} \beta_j  =1, \beta_j> 0} \E_{J_m} \|\beta\cdot J_m\|_{2} 
\]

For a fixed $m$, after re-ranking the entries by absolute value, let $S_m$ be the support so that only the top $m$ entries $|\tbe|_{(2)}, |\tbe|_{(3)},\ldots,|\tbe|_{(m+1)}$ are non-zero,  let 
\[
|\tbe'_{S_m}| = (0, |\tbe|_{(2)}, |\tbe|_{(3)},\ldots,|\tbe|_{(m+1)}, 0, \ldots, 0)
\]
\[
\|\tbe'_{S_m}\|_2^2 = |\tbe|^2_{(2)}+ |\tbe|^2_{(3)}+\ldots+|\tbe|^2_{(m+1)} 
\]
let the unit vector along the direction of $|\tbe'_{S_m}|$ be 
\[
u^m = |\tbe'_{S_m}|/\|\tbe'_{S_m} \|_2
\]
so that $u^m_j>0$, $j \in S_m$.
then 
\[
val(Q_1^{S_m}) := \frac{1}{\|\tbe'_{S_m}\|_2}\inf_{ \beta:  \beta_j> 0} \{ \E_{J_m} \| (\beta/\|\beta \|_{2}) \cdot J_m\|_{2} \cdot \frac{1}{\sum_{j=1}^{m}  u^m_j \frac{\beta_j }{\|\beta \|_{2}}} \}
\]
From previous definition, since $\beta$ is supported on $S_m$, we denote it as $\beta^m$, then
$$V_1(\beta^m)= \E_{J_m} \| (\beta/\|\beta \|_{2}) \cdot J_m\|_{2}, $$
$$V_1(\tbe'_{S_m}) = V_1(|\tbe'_{S_m}|) = V_1(u^m).$$

Geometrically,
\[
\cos\angle(\beta^m, u^m) = \sum_{j=1}^{m}  u^m_j \frac{\beta_j }{\|\beta \|_{2}}
\]
\[
\cot\angle(\tbe_{S_m}, \be_0) =  \frac{1}{\|\tbe'_{S_m}\|_2}
\]
\[
val(Q_1^{S_m}) =  \cot\angle(\tbe_{S_m}, \be_0) \inf_{ \beta:  \beta_j> 0} \{V_1(\beta^m) \frac{1}{\cos\angle(\beta^m, u^m)} \}
\]

Since $\beta^m = u^m$ is a feasible point of the  linear equality constraint 
\[
\langle  |\tbe'_{S_m}|, \beta'\rangle =\sum_{j=1}^{m}  |\tbe|_{(j+1)} \beta_j   =1,
\]
we have an upper bound 
\[
val(Q_1^{S_m}) \leq \cot\angle(\tbe_{S_m}, \be_0) V_1(\tbe'_{S_m})
\]
On the other hand,
\BEAS
val(Q_1^{S_m}) 
&=& \cot\angle(\tbe_{S_m}, \be_0) \inf_{ \beta:  \beta_j> 0} V_1(\beta^m) \cdot \frac{1}{\cos\angle(\beta^m, u^m)}
\\
&=& \cot\angle(\tbe_{S_m}, \be_0)  V_1(|\tbe'_{S_m}|) \inf_{ \beta:  \beta_j> 0} \frac{V_1(\beta^m)}{V_1(u^m)} \cdot \frac{1}{\cos\angle(\beta^m, u^m)}
\EEAS
The lower bound is given by finding the lower bound of
\[
C(u^m):=\frac{1}{V_1(u^m)}\inf_{ \beta^m:  \beta^m_j> 0} \frac{V_1(\beta^m)}{\cos\angle(\beta^m, u^m)}  = \frac{1}{E_{J_m}\|u^m\cdot J_m\|_2}\inf_{ \beta^m:  \beta^m_j> 0} \frac{E_{J_m}\|\beta^m\cdot J_m\|_2}{\langle  \beta^m, u^m\rangle }
\]
First, we know upper bound $C(u^m)\leq 1$ by plug in $\beta^m = u^m$ as feasible point. 

This value $C(u^m)$ can be restricted to the $m$-dimensional subspace supported on $S_m$, therefore its value is independent of ambient space dimension, only dependent on the non-zero entries of $u^m$, without loss of generosity we only need to consider the case when $u^m$ is a $m$-dimensional dense vector lie on the first  quadrant (all positive non-zero entries) of the unit sphere.

If $u^m$ approaches one-sparse point by taking the rest of coordinate close to zero, then $\beta^m$ should equal $u^m$.

On the other extreme, if $u^m$ is has $m$ equal entries $\frac{1}{\sqrt{m}}$, via numerical simulation we could check that \[
C(u^m)=1
\]
We conjecture that $C(u^m)=1$.We only need to prove
$C(u^m)\geq 1$.
To prove the lower bound, we only need to show that, on positive quadrant of $m$-dimensional unit sphere, the function 
$$R_{u^m}(\beta^m)=\frac{1}{E_{J_m}\|u^m\cdot J_m\|_2}\frac{E_{J_m}\|\beta^m\cdot J_m\|_2}{\langle  \beta^m, u^m\rangle }\geq 1 $$
and takes minimum value $1$
  at $\beta^m=u^m$.

In general, by symmetry, expect our guess $\beta^m = u^m$, for any $u^m$ on  we could infer that the $2$ other types of points to check on positive quadrant of unit sphere are $ \beta^m = (0, \sqrt{1-\epsilon^2}, \sqrt{\epsilon},0\ldots,  0)$ and $\beta^m = (\frac{1}{\sqrt{m}}, \ldots, \frac{1}{\sqrt{m}}).$ 
 
First , for $\beta^m = (\frac{1}{\sqrt{m}}, \ldots, \frac{1}{\sqrt{m}})$, then $\langle  \beta^m, u^m\rangle=  \frac{1}{\sqrt{m}} 1^T u^m \leq 1$ by Cauchy inequality,  we have
\[
R_{u^m}(\beta^m) = \frac{1}{E_{J_m}\|u^m\cdot J_m\|_2} \frac{E_{J_m}\|\beta^m\cdot J_m\|_2}{\langle  \beta^m, u^m\rangle} \geq \frac{E_{J_m}\|(\frac{1}{\sqrt{m}}, \ldots, \frac{1}{\sqrt{m}})\cdot J_m\|_2}{E_{J_m}\|u^m\cdot J_m\|_2}\geq 1
\]
The  last  inequality comes from the saddle point property of $V_1$ at point $(\frac{1}{\sqrt{m}}, \ldots, \frac{1}{\sqrt{m}})$.
 since $V_1((\frac{1}{\sqrt{m}}, \ldots, \frac{1}{\sqrt{m}})) = \sum_{j=0}^{m} (1-p)^{m-j} p^{j} {m \choose j}\sqrt{\frac{j}{m}} = \sqrt{p} - \sum_{j=m+1}^{\infty} (1-p)^{m-j} p^{j} {m \choose j}\sqrt{\frac{j}{m}}  \in [p, \sqrt{p}).$

Second, for $ \beta^m = (0, \sqrt{1-\epsilon^2}, \sqrt{\epsilon},0\ldots,  0) $, if $\beta^m$ approaches one-sparse point by taking the rest of coordinate close to zero, by continuity of function $R_{u^m}(\beta^m)$, we can consider the limit point on the boundary $ \beta^m = (0, 1, 0,0\ldots,  0) $, then
\[
R_{u^m}(\beta^m) = \frac{1}{E_{J_m}\|u^m\cdot J_m\|_2} \frac{E_{J_m}\|\beta^m\cdot J_m\|_2}{\langle  \beta^m, u^m\rangle} \geq \frac{p \|u^m\|^2_2}{E_{J_m}\|u^m\cdot J_m\|_2 \|u^m\|_\infty}
\]

We can first compute gradient of $R_{u^m}(\beta^m)$, show that its directional gradient are all non-negative at $u^m$. 



\end{proof}






\subsection{Technical Tool: Relation between Finite Difference of Objective and Euclidean Distance }
 \paragraph{Bi-Lipschitzness of finite difference of objective}
As an important proof tool, we study functional $B$ that allows us to bound the Euclidean distance $d_2(\psi, \be_0) := \| \psi - \be_0\|_2$ by the objective difference $\E_I  \| \psi\cdot I \|_2- \E_I  \| (\be_0)\cdot I \|_2$  : 
\[
B(\be_0,\phi) := \frac{\E_I  \| (\be_0+\phi)\cdot I \|_2 - \E_I  \| (\be_0)\cdot I \|_2}{\|\phi\|_2} 
\]
When we rescale $\phi$ to $\beta:= \frac{\phi}{\|\phi \|_2}$, we have
\[
B(\be_0,t \beta) = \frac{\E_I  \| (\be_0+t \beta)\cdot I \|_2 - \E_I  \| (\be_0)\cdot I \|_2}{t} 
\]
From the definition of directional derivative,
\[
\nabla_{\phi}{\E_I  \| (\be_0+\phi)\cdot I \|_2}\mid_{\phi=0}=\lim_{t\rightarrow 0^+}\frac{\E_I  \| (\be_0+t \beta)\cdot I \|_2 - \E_I  \| (\be_0)\cdot I \|_2}{t}=\lim_{t\rightarrow 0^+}B(\be_0,t \beta).
\]

We have upper and lower bound on their difference.
This upper and lower bound allows us to connect objective $\E_I  \| \psi\cdot I \|_2$ and the $2-$norm of $ \psi - \be_0$:
\begin{theorem}[Bi-Lipschitzness of finite difference of objective near $\be_0$ for linear constraint]
\label{thm:bound_diff}
We have upper and lower bound of $B(\be_0, \phi)$ :
\BEAS
0\leq B(\be_0,t \beta) -\lim_{t\rightarrow 0^+}B(\be_0,t \beta) 
\leq 
  \frac{t}{2}p\Big( 
  \beta_0^2 
 +
 p  (1-\beta_0^2) \Big) \leq \frac{pt}{2}.
\EEAS
This leads to
\[
0 \leq B(\be_0, \phi) -	 \nabla_{\phi}{\E_I  \| (\be_0+\phi)\cdot I \|_2}\mid_{\phi=0}\leq \frac{p}{2}\|\phi \|_2.
\]
Therefore, when $p<p^\star$, $\nabla_{\phi}{\E_I  \| (\be_0+\phi)\cdot I \|_2}\geq \epsilon(p,p^\star)>0$, 
we have $B(\be_0, \phi) \geq \epsilon(p,p^\star)>0$, which allows us to bound difference of objective by Euclidean distance .
Reversely, 
$1/B(\be_0, \phi) \leq  1/\epsilon(p,p^\star)$, which allows us to bound Euclidean distance by difference of objective.

\end{theorem}

\section{Conclusion}
In this paper, we proposed a novel {\it convex} optimization problem for sparse {\it blind deconvolution problem} based on $\ell^1$
minimization of inverse filter outputs:
\[
 \begin{array}{ll}
 \underset{\bw \in l_1^k}{\mbox{minimize}}   & \frac{1}{N} \| \bw \star \by\|_{\ell_1^N}\\
\mbox{subject to}  
& \langle \widetilde{\ba},  \bw^{\dagger} \rangle =1.
\end{array}
  \]
Assuming the signal to be recovered is {\it sufficiently sparse}, the algorithm can convert a crude approximation to the filter into a high-accuracy recovery of the true filter.   
  
We present four main results. 

First, in a large-$N$ analysis where $\bx$ is a realization of an IID Bernoulli-Gaussian signal with expected sparsity level $p$, we measure the approximation quality of $\widetilde{\ba}$ 
by considering $\tbe= \widetilde{\ba} \star \ba^{-1}$, 
which would be a Kronecker sequence if our approximation were perfect.
Under the condition
\[
\frac{|\tbe|_{(2)}}{|\tbe|_{(1)}} \leq 1-p,
\]
we show that, in the large-$N$ limit,
the $\ell^1$ minimizer $\bw^*$
perfectly recovers $\ba^{-1}$ to shift and scaling. 

In words {\it the less accurate the initial approximation $\widetilde{\ba} \approx \ba$, the greater we rely on sparsity of $\bx$.}

Second, we develop finite-$N$ guarantees of the form $N\geq O(k\log(k))$,
for highly accurate reconstruction with high probability.

Third, we further show stable approximation when the true inverse filter is infinitely long
(rather than length $k$), we show that the approximation error decrease exponentially as the approximation length growth.

Last, we extend our guarantees to the case where
the observation contain stochastic or adversarial noise, we show that in both stochastic or adversarial noise cases, the approximation error growth linearly as a function of noise magnitude.

\section{Supplementary: Landscape of  Expected Homogeneous Function over Bernoulli Support on Sphere}
In the following, we study the expected landscape for projection pursuit on sphere: $V_1$, $V_{2k}$, $V_{-1}$ 

 First, we calculate $\E| \psi^T X| $ for $X_t$ IID sampled from Bernoulli Gaussian $p N(0,1) +(1-p) \delta_0.$ Let $X_i = B_i Z_i$ be IID Bernoulli Gaussian, where $B_i$ is sampled from Bernoulli with parameter $p$, $Z_i$ is sampled from $N(0,1)$. Let $I$ be the support where $B_i =1$.

From now on, we denote three equivalent notation, and interchange them for the convenience of each context:
\[
\|\psi\|_{\ell_2(I)} = \|\psi_I\|_2 = \|\psi\cdot I\|_2
\]
\subsection{Expectation of Inner Product for Sparse Signal}
In the following, we study in detail the quantity  
\[
E{| \psi^T X| } = E{|\sum_{i} \psi_i X_i| }.
\]
We know when $X_i$ have variance $\sigma^2$,
\[
   E (\sum_i \psi_i X_i)^2 =  \sum_i   \psi_i^2 E(X_i)^2 = \sigma^2 \|\psi\|_2^2.
\]
The ratio indicates the sparsity level of the random variable $ \sum_i \psi_i X_i$
\[
\frac{ E{|\sum_{i} \psi_i X_i| }}{\sqrt{E (\sum_i \psi_i X_i)^2}} = \frac{ E{|\sum_{i} \psi_i X_i| }}{\sigma^2 \|\psi\|_2} 
\]
Now, we consider $E{|\sum_{i} \psi_i X_i| }$ in the general symmetric setting, where a lower bound can be derived.

\paragraph{Exact calculation for Bernoulli Gaussian}
\[
\E_X{|\sum_{i} \psi_i X_i| }  
= \E_B \E_G {|\sum_{i} \psi_i B_i G_i| }
= \E_I \E_G {|\sum_{i\in I} \psi_i G_i| }
\]
Leveraging the fact that linear transforms of Gaussians are also Gaussian, we get
\[
\E_G {|\sum_{i\in I} \psi_i G_i| } = 
\sqrt{\frac{2}{\pi}} \cdot\|\psi_I\|_2.
\]

\paragraph{Upper and lower bound for symmetric distribution}
It is worth commenting that we could still calculate the upper and lower bound of $\E{| \psi^T X| }$ in terms of $\E_I\| \psi_I\|_2$ for $X_t$ IID sampled from any Bernoulli symmetric $p G +(1-p) \delta_0$, where $G$ is a symmetric distribution. 

\begin{lemma} \label{lem:sym_Khintchine_inequality}
If we only know $\Xi_i$ are IID sampled from a symmetric distribution $F$ with unit variance, then we already have a lower bound,  
\[
     1/\sqrt{2}  \| a\|_2 \leq \E |\sum_i a_i \Xi_i| \leq \|a\|_2.
\] 
\end{lemma}
\begin{proof}[Proof of Lemma \ref{lem:sym_Khintchine_inequality}]
The upper bound come from Cauchy inequality.
 
Now we derive the lower bound. We write $\Xi_i = \sigma_i |\Xi_i|$, where $s_i =\pm 1$ with equal probability since $\Xi_i$ are symmetric RV, and $\{s_i,|\Xi_i|\}$  are all independent random variables. We could use Khintchine inequality,

\begin{eqnarray*}
\E_\Xi |\sum_i a_i \Xi_i| &=& E_{|\Xi_i|}E_{s_i} |\sum_i a_i s_i |\Xi_i||\\
&\geq&  
  1/\sqrt{2}  \sum_i a_i^2  E_{|\Xi_i|} |\Xi_i|^2 \\
  &=&   1/\sqrt{2}  \|a\|_2^2.
\end{eqnarray*}
\end{proof}

Let $\psi_I = a$, we have a corollary for Bernoulli Symmetric case,
\[
 \E_I \|\psi_I\|_2  \geq \E_X{|\sum_{i} \psi_i X_i| }  
\geq \sqrt{\frac{1}{2}} \cdot \E_I \|\psi_I\|_2.
\]

\begin{theorem}
Let $X_i = B_i \Xi_i$ be IID, where $B_i$ is sampled from Bernoulli with parameter $p$, let $I$ be the support where $B_i =1$. And $\Xi_i$ is variance $1$, sampled from: 
\begin{itemize}
    \item $N(0,1)$
    \item general symmetric distribution with variance $1$.
\end{itemize}

Then 
\begin{itemize}
    \item(for Bernoulli Gaussian:)
    \[
\E_X{|\sum_{i} \psi_i X_i| }  
= \sqrt{\frac{2}{\pi}} \cdot \E_I \|\psi_I\|_2
\]
    \item (for Bernoulli symmetric R.V. :) 
    \[
\E_I \|\psi_I\|_2 \geq \E_X{|\sum_{i} \psi_i X_i| }  
\geq \sqrt{\frac{1}{2}} \cdot \E_I \|\psi_I\|_2.
\]
\end{itemize}

\end{theorem}

\subsection{Expectation over Bernoulli Support}

We previously considered the identity
\beq \label{eq:startingPoint}
    \E | \sum_{i \in \cT}  \psi_i X_i | = \sqrt{\frac{2}{\pi}} \cdot \E_I \| \psi \|_{\ell_2(I)},
\eeq
where $X$ is a Bernoulli-Gaussian RV $BG(p;0,1)$,
and where $I$ is a random subset of the domain $\cT$
determined by Bernoulli-$p$ coin tossing.

Let $J \subset \cT$ be the support of $X$, where $  N:=|\cT|$, $N_J:=|J|$. Let $\psi_J : = \psi\cdot 1_J$ denote the elementwise product of vector $\psi$ with the indicator vector of subset $J$. Then we can consider our problem on the space $\cT$ by default, and rewrite for simplicity
\[
\| \psi_J \|_{2} : = \| \psi_J \|_{\ell_2(\cT)} = \| \psi \|_{\ell_2(J)}.
\]

This leads us to consider the following ratio:
\[
  V^J_k(\psi) := \frac{\| \psi \|_{\ell_2(J)}^{k}}{\| \psi \|_{\ell_2(\cT)}^{k}},
 \]
 and its expectation over all Bernoulli random subset,
\[
  V_k(\psi) := \frac{\E_I \| \psi \|_{\ell_2(I)}^{k}}{\| \psi \|_{\ell_2(\cT)}^{k}},
 \]
 where $I$ is again a random subset.
We remark that the lower bound of $V_k(\psi)$ is the optimal value of the optimization problem:
 \begin{eqnarray*}
    \min&   \E_I  \| \psi_I \|^k_2\\
    \mbox{ subject to }&   \| \psi \|_2 = 1,
\end{eqnarray*}
and the upper bound of $V_k(\psi)$ is the optimal value of the optimization problem
 \begin{eqnarray*}
    \max&   \E_I  \| \psi_I \|^k_2\\
    \mbox{ subject to }&   \| \psi \|_2 = 1.
\end{eqnarray*}

\begin{lemma} \label{lem:BGPropLemmaSub_fix}
Let $\psi$ denote a vector in $\reals^N$. For fixed support $J\subset [N]$ then for any $k\neq 0$, the deterministic quantity $V^J_{k}(\psi)$ satisfies:
\[
0 \leq V^J_{k}(\psi) \leq 1.
\]
The upper bound is approached by the vectors $\psi  \in \{  \frac{1}{\sqrt{N}}\sum_{j\in J}\pm e_j\}$, where  $\| \psi_J \|_2^k/\|  \psi \|_2^k =1$,
 and the lower bound is approached by 
any $\psi$ that has all of its entries on $J$ to be zero.
\end{lemma}

\begin{lemma} \label{lem:BGPropLemmaSub_2}
Let $\psi$ denote a vector in $\reals^N$.
\[
V_2(\psi) =   p.
\]
\end{lemma}

\begin{theorem} 
Let $\psi$ denote a vector in $\reals^N$, then
\[
p \leq V_{1}(\psi) \leq V_{1}(\frac{\sum_{j\in [N]}\pm e_j}{\sqrt{N}} ) \leq \sqrt{p} .
\]
the lower bound is approached by on-sparse vectors $\psi \in \{ \pm e_i, i\in [N]\}$, and the upper bound is approached by 
$\psi  \in \{  \frac{1}{\sqrt{N}}\sum_{j\in [N]}\pm e_j\}$.

Furthermore, all the stationary points of $V_{1}(\psi)$ are $\{\frac{\sum_{i\in J} \pm e_i}{\sqrt{N_J}}\}$ for different support $J\subset \cT$, where $\{ \pm e_i, i\in [N]\}$ are the global minimizers, and $\{  \frac{1}{\sqrt{N}}\sum_{j\in [N]}\pm e_j\}$ are the global maximizers. And for $J$ with $1<N_J<N$, $\{\frac{\sum_{i\in J} \pm e_i}{\sqrt{N_J}}\}$ are saddle points with value
\[
V_1(\frac{\sum_{i\in J} \pm e_i}{\sqrt{N_J}}) = \E_I \sqrt{\frac{\sum_{i\in J} 1_{I_i}}{N_J}} = \sum_{j=0}^{N_J} (1-p)^{N_J-j} p^{j} {N_J \choose j}\sqrt{\frac{j}{N_J}}.
\]
\end{theorem}
Remark:
Specifically,
\[
V_1( e_i) = p.
\]
\[
V_1(\frac{ e_i+e_j}{\sqrt{2}}) = p(\sqrt{2}+(1-\sqrt{2}p)).
\]

\begin{theorem} 
For $k\geq 2$, let $\psi$ denote a vector in $\reals^N$, then
\[
p^k \leq V_{2k}(\frac{\sum_{j\in [N]}\pm e_j}{\sqrt{N}} ) \leq V_{2k}(\psi) \leq p.
\]
the upper bound is approached by on-sparse vectors $\psi \in \{ \pm e_i, i\in [N]\}$, and the lower bound is approached by 
$\psi  \in \{  \frac{1}{\sqrt{N}}\sum_{j\in [N]}\pm e_j\}$.

Furthermore, all the stationary points of $V_{2k}(\psi)$ are $\{\frac{\sum_{i\in J} \pm e_i}{\sqrt{N_J}}\}$ for different support $J\subset \cT$, where $\{ \pm e_i, i\in [N]\}$ are the set of global maximizers, and $\{  \frac{1}{\sqrt{N}}\sum_{j\in [N]}\pm e_j\}$ are the set of all the global minimizers. And for $J$ with $1<N_J<N$, $\{\frac{\sum_{i\in J} \pm e_i}{\sqrt{N_J}}\}$ are saddle points with value
\[
V_{2k}(\frac{\sum_{i\in J} \pm e_i}{\sqrt{N_J}}) = \E_I (\frac{\sum_{i\in J} 1_{I_i}}{N_J})^k = \sum_{j=0}^{N_J} (1-p)^{N_J-j} p^{j} {N_J \choose j}(\frac{j}{N_J})^k.
\]
\end{theorem}

\subsection{Upper and Lower bound on Expectation of Norm over Bernoulli Support}
\paragraph{Tight bound on $V_1$ and $V_{2k}$ using mean and variance of $\frac{\sum_{i\in J} 1_{I_i}}{N_J}$}

Now notice that 
\[
\frac{\sum_{i\in J} 1_{I_i}}{N_J} 
\]
has mean $p$ and variance $\frac{p(1-p)}{N_J}$.
We could define a zero mean unit variance random variable $g^J$ as a function of $\{1_{I_i}\mid i\in J\}$ 
\[
g^J := \left(\frac{\sum_{i\in J} 1_{I_i}}{N_J}-p\right)/ \left( \sqrt{\frac{p(1-p)}{N_J}} \right)
\]
namely, $\E_I \left( g^I\right)=0, \E_I \left( g^I\right)^{2}=1$.
then
\[
\frac{\sum_{i\in J} 1_{I_i}}{N_J} := p + \sqrt{\frac{p(1-p)}{N_J}} g^J = p(1 + \sqrt{\frac{1-p}{p N_J}} g^J)
\]
then
\[
\E_I \sqrt{\frac{\sum_{i\in J} 1_{I_i}}{N_J}} = \E_I\sqrt{p- \sqrt{\frac{p(1-p)}{N_J}}g^I} 
= \sqrt{p} \E_I\sqrt{1+ \sqrt{\frac{1-p}{pN_J}}g^I}  
\]
\[
\E_I \left(\frac{\sum_{i\in J} 1_{I_i}}{N_J}\right)^k 
= p^k \E_I\left({1+ \sqrt{\frac{1-p}{pN_J}}g^I}\right)^k 
\]
\begin{lemma}
\label{TaylorSqrt}
For $ 0 \leq x\leq 1,$
\begin{align}
(1-x)^{1/2}
=1-\sum_{\ell=0}^\infty\frac2{\ell+1}\binom{2\ell}{\ell}\left(\frac x4\right)^{\ell+1}
\end{align}
\end{lemma}

\begin{lemma}[Taylor expansion for $V_1$]
\label{TaylorSqrt_Expect}
Using Taylor expansion, asymptotically when $N_J\rightarrow \infty$, 
\BEAS
\frac{1}{\sqrt{p}}\E_I \sqrt{\frac{\sum_{i\in J} 1_{I_i}}{N_J}} 
&=& \E_I\left({1+ \sqrt{\frac{1-p}{pN_J}}g^I}\right)^{1/2} \\
&=& 1-\sum_{\ell=0}^\infty\frac2{\ell+1}\binom{2\ell}{\ell}2^{-2(\ell+1)}\left( \sqrt{\frac{1-p}{pN_J}}\right)^{\ell+1}\E_I \left( g^I\right)^{\ell+1}\\
&=& 1- 2^{-3}({\frac{1-p}{pN_J}}) +O\left( ({\frac{1-p}{pN_J}})^2 \right).
\EEAS
\end{lemma}

\begin{lemma}[Taylor expansion for $V_{2k}$]
\label{TaylorSqrt_Expect_k}
Using Taylor expansion, asymptotically when $N_J\rightarrow \infty$, 
\BEAS
\frac{1}{p^k}\E_I ({\frac{\sum_{i\in J} 1_{I_i}}{N_J}})^k  
&=& \E_I\left({1+ \sqrt{\frac{1-p}{pN_J}}g^I}\right)^{k} \\
&=& 1 +\sum_{\ell=1}^{k} \binom{k}{\ell} ({\frac{1-p}{pN_J}})^{\ell/2} \E_I (g^I)^{\ell}\\
&=& 1 + \frac{k(k-1)}{2} ({\frac{1-p}{pN_J}}) +O\left( ({\frac{1-p}{pN_J}})^2 \right).
\EEAS
\end{lemma}
Therefore, we know what when $N_J>>1$, 
$V_1(\frac{\sum_{i\in J} 1_{I_i}}{N_J})$ is close to its upper bound $\sqrt{p}$, 
and for $k\geq 2$, 
$V_k(\frac{\sum_{i\in J} 1_{I_i}}{N_J})$ is close to its lower bound $p^k$. Namely 
\[
\E_I \sqrt{\frac{\sum_{i\in J} 1_{I_i}}{N_J}} \approx \sqrt{p},
\]
\[
\E_I (\frac{\sum_{i\in J} 1_{I_i}}{N_J})^k \approx p^k.
\]



To gain geometric intuition, we visualize $V_1$ in $2-$dimensional sphere in figure~\ref{fig-landscape-V1}.


\begin{figure} 
\centering
\includegraphics[width=.7\textwidth]{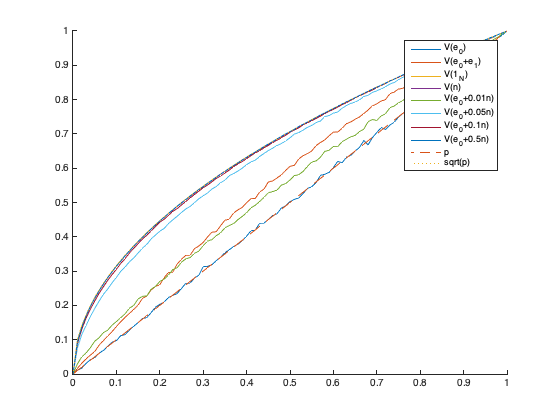}
\caption{The value of $V_1(\psi)$ for different $\psi$, $x-$axis is the sparsity parameter $p$, $y-$axis is $V_1(\psi)$. }
\end{figure}

\begin{figure} 
\centering
\includegraphics[width=.7\textwidth]{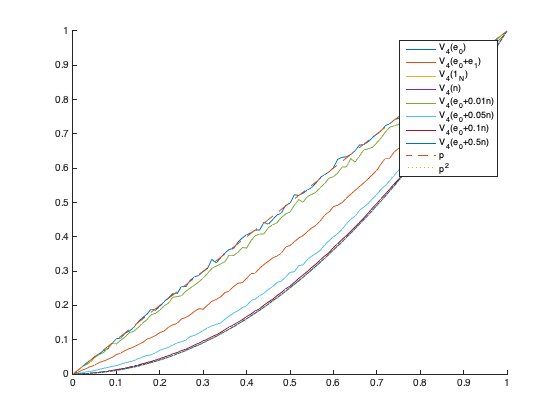}
\caption{The value of $V_4(\psi)$ for different $\psi$, $x-$axis is the sparsity parameter $p$, $y-$axis is $V_4(\psi)$. }
\end{figure}

\begin{figure} 
\centering
\includegraphics[width=.7\textwidth]{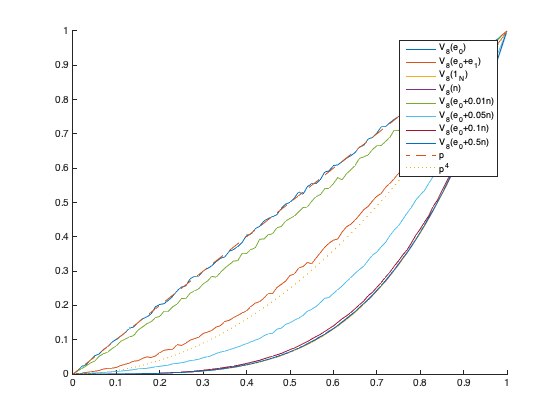}
\caption{The value of $V_8(\psi)$ for different $\psi$, $x-$axis is the sparsity parameter $p$, $y-$axis is $V_8(\psi)$. }
\label{fig-V4}
\end{figure}

\subsection{Proofs of Upper and Lower Bounds}

\begin{proof}[Proof of Theorem~\ref{thm:BGPropLemmaNorm}]
First, we show that the upper and lower bound value we give is achievable.
\BIT
\item 
When $\psi = e_0$,  
\[
\E_I  \| (e_0)_I \|_2=p\| (e_0)\|_2+(1-p)\| 0\|_2 = p.
\]
\item
When $\psi = \frac{1}{\sqrt{N}}\sum_{j\in [N]}\pm e_j,$ then  
\[
V_1(\psi) = \E_I \sqrt{\frac{\sum_{i\in [N]} 1_{I_i}}{N}} \leq   \sqrt{\E_I \frac{\sum_{i\in [N]} 1_{I_i}}{N}} = \sqrt{p}.
\]
The inequality comes from Jensen's inequality, since square root function is concave.
\item
For fixed support $J$, when $\psi = \frac{1}{\sqrt{N_J}}\sum_{i\in J} e_i,$  then $\| \psi_J \|_2/\|  \psi \|_2 =1.$
However, the expectation is going to be smaller for $p<1$. 
\[
V_1(\psi) = \E_I \sqrt{\frac{\sum_{i\in J} 1_{I_i}}{N_J}} = \sum_{k=0}^{N_J} (1-p)^{N_J-k} p^{k} {N_J \choose k}\sqrt{\frac{k}{N_J}}.
\]
\EIT

This problem can be reformulated as projection pursuit with sphere constraint,
 \begin{eqnarray*}
    \min&   \E_I  \| \psi_I \|_2\\
    \mbox{ subject to }&   \| \psi \|_2 = 1.
\end{eqnarray*}
Then from \cite{bai2018subgradient} Proposition $3.3$, we get the result. 

The main idea of the proof is to calculate the projected gradient
for $q  \in \{\frac{\sum_{i\in J} \pm e_i}{\sqrt{N_J}}\}$, when $N_J = M,$
\[
e_j^T  \E_I  \partial\| (q)_I \|_2 = q_j [\sum_{k=0}^{M} (1-p)^{M-k} p^{k} {M \choose k}\sqrt{\frac{k}{M}}] ,
\]
therefore,
\[
  \E_I  \partial\| (q)_I \|_2 =  [\sum_{k=0}^{M} (1-p)^{M-k} p^{k} {M \choose k}\sqrt{\frac{k}{M}}] q,
\]
so
\[
(I-q q^T)  \E_I  \partial\| q_I \|_2 = [\sum_{k=0}^{M} (1-p)^{M-k} p^{k} {M \choose k}\sqrt{\frac{k}{M}}]q -[\sum_{k=0}^{M} (1-p)^{M-k} p^{k} {M \choose k}\sqrt{\frac{k}{M}}]q  = 0.
\]
then $q \in \{\frac{\sum_{i\in J} \pm e_i}{\sqrt{N_J}}\}$ are stationary points.
 
The other direction (all other points are not stationary) is implied by (the proof of) Theorem $3.4$ in  \cite{bai2018subgradient}.
\end{proof}

\begin{proof}[Proof of Lemma~\ref{TaylorSqrt}]
From generalized binomial theorem, we know that for $ 0 \leq x\leq 1,$
\beq
(1-x)^{1/2}
=\sum_{\ell=0}^\infty\binom{1/2}{\ell}(-x)^\ell.
\eeq

\begin{align}
\binom{1/2}{\ell}
&=\frac{\frac12(\frac12-1)(\frac12-2)\cdots(\frac12-\ell+1)}{\ell!}\\
&=\frac{(-1)^{\ell-1}}{2^\ell\ell!}1\cdot3\cdot5\cdots(2\ell-3)\\
&=\frac{(-1)^{\ell-1}}{2^\ell\ell!}\frac{(2\ell-2)!}{2^{\ell-1}(\ell-1)!}\\
&=\frac{(-1)^{\ell-1}}{\ell2^{2\ell-1}}\binom{2\ell-2}{\ell-1},
\end{align}
then
\begin{align}
(1-x)^{1/2}
&=1-\sum_{\ell=1}^\infty\frac2\ell\binom{2\ell-2}{\ell-1}\left(\frac x4\right)^\ell\\
&=1-\sum_{\ell=0}^\infty\frac2{\ell+1}\binom{2\ell}{\ell}\left(\frac x4\right)^{\ell+1}.
\end{align}
\end{proof}

\begin{proof}[Proof of Lemma~\ref{TaylorSqrt_Expect}]
\BEAS
\E_I\sqrt{1- \sqrt{\frac{1-p}{pN_J}}g^I}  
&=& 1-\sum_{\ell=0}^\infty\frac2{\ell+1}\binom{2\ell}{\ell}2^{-2(\ell+1)}\left( \sqrt{\frac{1-p}{pN_J}}\right)^{\ell+1}\E_I \left( g^I\right)^{\ell+1}\\
&=& 1- 2^{-3}({\frac{1-p}{pN_J}}) - O(({\frac{1-p}{pN_J}})^2)
\EEAS
Using the central limit theorem, when $N_J\rightarrow \infty$, we have normal approximation for $g^J$ so that
\[
\frac{\sum_{i\in J} 1_{I_i}}{N_J} \sim N(p,  \sqrt{\frac{p(1-p)}{N_J}}),
\]
then for $G\sim N(0,1),$ asymptotically when $N_J\rightarrow \infty$,
\[
\E_I \sqrt{\frac{\sum_{i\in J} 1_{I_i}}{N_J}} \approx  \sqrt{p} \E_G\sqrt{1- \sqrt{\frac{1-p}{pN_J}}G} =
\sqrt{p} \left( 1- 2^{-3}(1-p)^{2}  N_J^{-1} + O((1-p)^{4}  N_J^{-2})\right),
\]
where
\BEAS
\E_I\sqrt{1- \sqrt{\frac{1-p}{pN_J}}g^I}  
&\approx &\E_G\sqrt{1- \sqrt{\frac{1-p}{pN_J}}G}\\
&=& 1-\sum_{\ell=0}^\infty\frac2{\ell+1}\binom{2\ell}{\ell}2^{-2(\ell+1)}\left( \sqrt{\frac{1-p}{pN_J}}\right)^{\ell+1}\E_G \left( G\right)^{\ell+1}\\
&=& 1-\sum_{\ell \mbox{ is odd}}^\infty\frac2{\ell+1}\binom{2\ell}{\ell}2^{-2(\ell+1)}\left( \sqrt{\frac{1-p}{pN_J}}\right)^{\ell+1} (\ell)!!\\
&=& 1-\sum_{s=1}^\infty\frac1{s}\binom{4s-2}{2s-1}2^{-4s} (2s-1)!! (p^{-1}-1)^{s}  N_J^{-s}\\
&=& 1-\sum_{s=1}^\infty\frac1{s}\frac{(4s-2)!}{(2s-1)!(2s-1)!}2^{-4s} \frac{(2s-1)!}{2^{s-1}(s-1)!} (p^{-1}-1)^{s}  N_J^{-s}\\
&=& 1-\sum_{s=1}^\infty\frac{(4s-2)!}{(2s-1)!(s)!}2^{-5s+1}  (p^{-1}-1)^{s}  N_J^{-s}.
\EEAS
\end{proof}

\subsection{Bound on Harmonic Expectation }
\paragraph{Bounds on  $V_{-1}(\psi)=\frac{\E_I  \| \psi \|^{-1}_{\ell_2(I)} }{\| \psi \|^{-1}_{\ell_2(\cT)}}$}
\begin{lemma}
\label{TaylorSqrtInv}
For $ 0 \leq x\leq 1,$
\begin{align}
(1-x)^{-1/2}
=\sum_{\ell=0}^\infty\frac{1}{2^{2\ell}}\binom{2\ell}{\ell}\left(x\right)^{\ell}
\end{align}
\end{lemma}

\begin{proof}
From generalized binomial theorem, we know that for $ 0 \leq x\leq 1,$
\beq
(1-x)^{-1/2}
=\sum_{\ell=0}^\infty\binom{-1/2}{\ell}(-x)^\ell
\eeq

\begin{align}
\binom{-1/2}{\ell}(-1)^\ell
&=\frac{-\frac12(-\frac12-1)(-\frac12-2)\cdots(-\frac12-\ell+1)}{\ell!}\\
&=\frac{1}{2^\ell\ell!}1\cdot3\cdot5\cdots(2\ell-1)\\
&=\frac{1}{2^\ell\ell!}\frac{(2\ell)!}{2^{\ell}(\ell)!}\\
&=\frac{1}{2^{2\ell}}\binom{2\ell}{\ell}
\end{align}
\end{proof}

\begin{theorem}
\label{thm: TaylorSqrtInvExpect}
\beq 
  V_{-1}(\psi) = \frac{\E_I  \| \psi \|^{-1}_{\ell_2(I)} }{\| \psi \|^{-1}_{\ell_2(\cT)}} \geq p^{-1/2} \geq 1+\frac{1}{2}(1-p) \geq 1 .
\eeq
\end{theorem}
\begin{proof}
For any fixed support $J\subset \cT,$ 
\beq 
 \frac{ \| \psi \|^{-1}_{\ell_2(J)} }{\| \psi \|^{-1}_{\ell_2(\cT)}} 
 = \sqrt{ \frac{ \| \psi \|^2_{\ell_2(\cT)} }{\| \psi \|^2_{\ell_2(J)}}} \geq 1 .
\eeq
Therefore, 
\beq 
\frac{\E_I  \| \psi \|^{-1}_{\ell_2(I)} }{\| \psi \|^{-1}_{\ell_2(\cT)}}  
 \geq   1 .
\eeq

On the other hand, let 
\beq 
\E_I (\frac{ \| \psi \|_{\ell_2(I)} }{\| \psi \|_{\ell_2(\cT)}} )^{-1}
 = \E_I ( 1- \frac{ \| \psi \|^2_{\ell_2(\cT-I)} }{\| \psi \|^2_{\ell_2(\cT)}})^{-1/2}.
\eeq

For a fixed true subset $J \subset \cT$ define 
the ratio $\rho_{\cT -J} \equiv  \frac{ \| \psi \|^2_{\ell_2(\cT-J)} }{\| \psi \|^2_{\ell_2(\cT)}}$; it 
obeys $0 \leq \rho_J \leq 1$. Assume $\rho_J < 1$

\begin{align}
(1-\rho_{\cT -J})^{-1/2}
=\sum_{\ell=0}^\infty\frac{1}{2^{2\ell}}\binom{2\ell}{\ell}\left(\rho_{\cT -J}\right)^{\ell}
\end{align}
Now with $J$ a random subset as earlier, we
induce a random variable $\rho_I$.
\begin{align}
\E_J(1-\rho_{\cT -J})^{-1/2}
=\sum_{\ell=0}^\infty\frac{1}{2^{2\ell}}\binom{2\ell}{\ell}\E_J\left(\rho_{\cT -J}\right)^{\ell}
\end{align}
We apply the bound on $V_{2\ell}$ for $\ell\geq 2$, where $p$ in the final formula is replaced by $ 1-p$.

We obtain for $\ell\geq 2$
\[
 (1-p)^\ell    \leq \E_J \rho_{\cT -J}^{\ell} \leq    (1-p),
\]
\begin{align}
\E_J(1-\rho_{\cT -J})^{-1/2}
&=\sum_{\ell=0}^\infty\frac{1}{2^{2\ell}}\binom{2\ell}{\ell}\E_J\left(\rho_{\cT -J}\right)^{\ell}
\\
&\geq  \sum_{\ell=0}^\infty\frac{1}{2^{2\ell}}\binom{2\ell}{\ell} (1-p)^\ell\\
&=(1-(1-p))^{-1/2}\\
&= p^{-1/2}\\
\geq 1+\frac{1}{2}(1-p) 
\end{align}

\end{proof}

\section{Supplementary: Background for Convex Blind Deconvolution Problem}
\subsection{Technical Background: Wiener's Lemma and Inverse Filter}
\paragraph{Fourier transform and inverse filter}
The discrete-time Fourier transform $\mathcal{F}$ is defined by $\mathcal{F} \ba $ where
\[
 (\mathcal{F} \ba)(\omega) : = \sum_{n=-\infty}^\infty a_n\, e^{2\pi i n \omega}, \quad  \omega\in T = [-\frac{1}{2}, \frac{1}{2}];
\]
the inverse Fourier transform $\mathcal{F}^{-1}$ is defined by
\[
(\mathcal{F}^{-1} f)(n) =  \int_{-\frac{1}{2}}^\frac{1}{2} f(\omega) e^{-2\pi i n \omega}  d\omega, \quad n\in \integers.
\]
Now, our condition on the filter  $\ba$ is:
\[a\in l_1(\integers) ,  
(\mathcal{F} \ba)(\omega) \neq 0, \forall \omega\in T.
\]
In the following theorem, we show that this condition would provide the existence of an inverse filter in  $l_1(\integers)$.

First, we present the standard Wiener's lemma.
\begin{lemma}{\textbf{Wiener's lemma on periodic functions:}}
Assume that a function $f$ on unit circle has an absolutely converging Fourier series, and $f(t)\neq 0$ for all $t \in T$, then $1/f$ also has an absolutely convergent Fourier series.
\end{lemma}
Then we can see clearly that the Fourier series version of the previous lemma would guarantee the existence of an inverse filter in  $l_1(\integers)$.
\begin{lemma}{\textbf{Wiener's lemma on $l_1(\integers)$ sequences :}}
If  $\ba\in l_1(\integers) ,  
(\mathcal{F} \ba)(\omega) \neq 0, \forall \omega\in T,$ we could define the inverse filter of  $\ba$ as
$\ba^{-1} := \mathcal{F}^{-1}(\frac{1}{\mathcal{F} \ba})$
so that  $\ba^{-1}*a = \be_0.$ Here $e_0$ is the sequence with $1$ at $0$ coordinate and $0$ elsewhere.
From Wiener's lemma on $\mathcal{F} \ba$,    $\ba^{-1} \in l_1(\integers)$ .
\end{lemma}

\subsection{Change of Variable and Reduction to Projection Pursuit}
Rewrite the population version of our convex sparse blind deconvolution problem, with the population objective 
$\Expect\frac{1}{N} \| \bw*Y\|_{\ell_1(\cT)} =\Expect|(\bw*Y)_0|=\Expect|(\bw*\ba*X)_0|=\Expect_X|\langle X, (\bw*\ba)^{\dagger} \rangle|$ due to  the ergodic property of stationary process and shift invariance, and $(\widetilde{\ba}*\bw)_0=((\widetilde{\ba}*\ba^{-1})*(\ba*\bw))_0  = \langle (\widetilde{\ba}*\ba^{-1})^{\dagger}, \bw*\ba\rangle$, the convex problem becomes 
\[
 \begin{array}{ll}
 \underset{\bw}{\mbox{minimize}}   & \Expect_X|\langle X ,  (\ba*\bw)^{\dagger} \rangle|\\
\mbox{subject to}  
& \langle \widetilde{\ba}*\ba^{-1} , (\ba*\bw)^{\dagger}\rangle =1,
\end{array}
 \]
Let $\psi$ denote  the time reversed version of $\ba*\bw$: $\psi := (\ba*\bw)^{\dagger}$, and
  let $\tbe:= \widetilde{\ba}*\ba^{-1}$, then by previous assumptions, 
  $\tbe_0=1$, $\tbe' = \tbe-\be_0$.
  
 Now we arrive at a simple and fundamental population convex problem: 
 \[
 \begin{array}{ll}
 \underset{\psi}{\mbox{minimize}}   &\Expect_X|\langle X ,  \psi  \rangle|\\
\mbox{subject to}  
& \langle \tbe ,  \psi  \rangle  =1.
\end{array}
  \]
  \textbf{Expectation using Gaussian.  }
Since $X$ follows Bernoulli-Gaussian IID probability model $X_t = I_t G_t$, we nest the expectation over $I_t$ outside the expectation over Gaussian $G_t$, for which we use $E |N(0,1)| = \sqrt{\frac{2}{\pi}}$:
\[
\Expect_X|\langle X ,  \psi  \rangle|
= \E_I \E_G {|\sum_{t\in \integers}  I_t G_t \psi(t)| }
= \sqrt{\frac{2}{\pi}} \cdot\E_I \|\psi\cdot I\|_2
\]

\subsection{Technical background: Directional Derivative and Projected Subgradient }
\paragraph{Exact calculation of subgradient for phase transition }
In this section, using sub-gradient and directional derivative, we compute the KKT condition of our problem rigorously.
\begin{lemma}
Let $J \subset [n]$ be the support of $X$, let $\psi_J : = \psi\cdot 1_J$ denote the elementwise product of vector $\psi$ with the indicator vector of subset $J$.  Let $B_J$ denote the central section of the euclidean  ball $B(\reals^n)$,
where the slice is produced the linear space $\mbox{span} \{e_i : i\in J \}$.
Alternatively, we may write
$B_J: = \{ v_J: \| v_J\|_2 \leq 1\}$.  The set-valued subgradient
operator applied to $\|\psi_J\|$
evaluates as follows:
 \[
 \begin{array}{ll}
\partial_{\psi}    \| \psi_J \|
&=  N(\psi_J).
\end{array}
  \]
Here $N$ maps $\reals^n$ into  subsets of $\reals^n$,
and is given by:
\[
N(\psi_J):= 
\begin{cases}
    \psi_J/\|\psi_J\|_2, &  \psi_J\neq 0 \\
    B_J, &  \psi_J = 0
  \end{cases}.
\]  
\end{lemma}

\begin{definition}
Consider a probability space containing just the possible  
outcomes  $J \subset [n]$.
Let $S_J$, $J \subset [n]$,  denote a closed compact subset of $\reals^n$.
Let $I$ be a random subset of
$[n]$ drawn at random from this probability space with probability $\pi_J = P\{ I = J \}$
of elementary event $J$.
Consider the set-valued random variable $S \equiv S_J$. We define its expectation as
\[
 {\Expect} S  := \sum_J \pi_J \cdot S_J.
\]
On the right side, we mean the closure of the 
compact set produced by all sums of the form
\[
 \sum_J \pi_J \cdot s_J
\]
where each $s_J \in S_J$.
\end{definition}

\begin{lemma}
For any $\psi \in \reals^n$, and $X = (X_i)_{i=1}^n$ with $X_i \sim_{iid} BG(p,0,1)$.
Let now $I \subset [n]$ be the random support of $X$.
Let $\partial_\psi $ denote the set-valued subgradient
operator.
  \[
 \begin{array}{ll}
\partial_\psi [E{|\sum_{i} \psi_i X_i| }  ]
&= \sqrt{\frac{2}{\pi}} \cdot {\Expect}_I [N(\psi_I)].
\end{array}
  \]
\end{lemma}

\paragraph{Subgradient and directional derivative at $e_0$}
Now we focus on $e_0$. Note that for any subset $J \subset [n]$, $(\be_0)_J$ is either the zero
vector or else $e_0$. Hence $N((\be_0)_J)$
is either $B_{J}$ or $\{\be_0\}$. What drives
this dichotomy is whether $0 \in J$ or not.

\begin{lemma}  
\[
 \begin{array}{ll}
\Expect_I [N((\be_0)_I) ]
&= p \be_0 + \sum_{J,0\not\in J } \pi_{J} B_{J} \\
& = p \be_0 + (1-p) {\cal B}_0,
\end{array}
  \]
where
\[
 \begin{array}{ll}
{\cal B}_0 = \Expect_{I}[ B_{I} | 0 \not \in I ] = 
(1-p)^{-1} \cdot \sum_{J,0\not\in J } {\pi_{J}} B_{J}.
\end{array}
\]
 
\end{lemma}

\begin{proof}
\begin{eqnarray*}
\Expect_I [N((\be_0)_I) ] &=& \sum_{J } \pi_{J} B_{J} \\
&=&  \sum_{J, 0 \in J}   \pi_{J} \be_0 +  \sum_{J,0\not\in J }  \pi_{J} B_{J}\\
&=& p \cdot \be_0 + (1-p) \cdot {\cB}_0 .
\end{eqnarray*}
 \end{proof}
To compute with $\Expect_I [N((\be_0)_I) ]$,
we need:
\begin{lemma} \label{eq:ComputeExpect}
For each fixed $\beta \in \reals^n$:
\[
  \sup_{b \in {\cal B}_0} \langle \beta, b \rangle
   = \Expect_{I}[ \|\beta_I\|_2 | 0 \not \in I  ].
  \]
\end{lemma}

\begin{proof}
Note that in the definition of the set 
\[
{\cal B}_0 = \Expect_{I}[ B_{I} | 0 \not \in I ],
\]
 each term $B_J$ obeys the bound
$\sup_{g_J \in B_J} \|g_J\|_2 = 1 $.
Now, given the fixed vector $\beta \in \reals^n$, define:
\begin{eqnarray*}
    b^* &=& \Expect_{I}[ \frac{\beta_I}{\| \beta_I \|_2} | 0 \not \in I  ]\\
     &=&  \sum_{0 \not \in J} \pi_J \frac{\beta_J }{\|\beta_J\|_2}.
 \end{eqnarray*}
 Since each term in this sum has Euclidean norm at most $1$, $b^* \in {\cB}_0$.
Now
\begin{eqnarray*}
\langle  \beta, b^* \rangle &=&
\langle \beta , \sum_{0 \not \in J} \pi_J \frac{\beta_J }{\|\beta_J\|_2} \rangle \\
&=& \sum_{0 \not \in J} \pi_J \langle \beta , \frac{\beta_J }{\|\beta_J\|_2} \rangle \\
&=& \sum_{0 \not \in J} \pi_J \|\beta_J\|_2 .
 \end{eqnarray*}
On the other hand, for any fixed 
vector $g = \sum \pi_J g_J $,
with each $g_J \in B_J$,
we have
\begin{eqnarray*}
\langle  \beta, g \rangle &=&  \sum_{0 \not \in J} \pi_J \langle \beta , g_J  \rangle \\
&\leq& \sum_{0 \not \in J} \pi_J  \|\beta_J \|_2 \|g_J\|_2 \\
& \leq & \sum_{0 \not \in J} \pi_J  \|\beta_J \|_2 
= \langle  \beta, b^* \rangle.
 \end{eqnarray*}
\end{proof}

Lemma (\ref{eq:ComputeExpect}) 
can be viewed as a special instance of
Theorem \textbf{A.15} from \cite{bai2018subgradient}:
\begin{lemma}[Interchangeability of set expectation and support function] 
Suppose a random compact set $S \subset \reals^n$ is integrably bounded and the underlying probability space is non-atomic, then $\E[S]$ is a convex set and for any fixed vector $\beta \in \reals^n,$
 \begin{equation} \label{eq:SetExpectEquality}
 \sup_{g\in \Expect S} \langle\beta, g\rangle =  \sup_{g\in  S} \Expect\langle\beta, g\rangle .
 \end{equation}
\end{lemma}

Define $\beta_{(0)} = \beta \cdot 1_{\{0\}^c}$ as the part of $\beta$ supported away from $0$.

\section{ Main Result $1$ and Its Proof: Phase Transition}
\subsection{KKT Condition for Exact Recovery}
\paragraph{KKT condition for $\be_0$ to be optimal solution.}
Given the tool defined above, we can calculate KKT rigorously. We first state the overview.

Let $\psi^\star$ be the solution of the optimization problem:
\[
 \begin{array}{ll}
 \underset{\psi}{\mbox{minimize}}   & \E_I \|\psi\cdot I\|_2\\
\mbox{subject to}  
& \langle   \tbe, \psi\rangle =1
\end{array}
\tag{$Q_1(\tbe)$ }
  \]
  
We claim that to prove that $\be_0$ solves ($Q_1(\tbe)$), we calculate the directional finite difference at $\be_0$. Then $\be_0$ solves this convex problem if the directional finite difference at $\psi = \be_0$ is non-negative at every direction $\beta$ on unit sphere where $\tbe^T \beta =0$:
	 \[
	 \Expect\|(\be_0+\beta)\cdot I\|_2 - \Expect\|(\be_0)\cdot I\|_2 \geq 0.
	 \]
We decompose the objective conditioning on whether $I_0 = 1_{\{X_0\neq 0\}}$ is zero or not:
	 \[
	 \begin{array}{ll}
	    \Expect\|(\be_0+\beta)\cdot I\|_2 - \Expect\|(\be_0)\cdot I\|_2
	    &= p (1+\beta_0)  + (1-p) \nabla_{\beta}{\Expect_I[\|(\be_0+\beta)'\|_{\ell_2(I-\{0\})} \mid I_0 =0]} -p\\
	    &= p \beta_0  + (1-p) \Expect_{I'}[\|\beta'\|_{\ell_2(I')} ].
	 \end{array}
	\]
This will be non-negative in case either $\beta_{0} > 0$, or else $\beta_0 < 0$ but
	\[
	\frac{p}{1-p} \leq  \frac{\Expect_{I'}\|\beta'\|_{\ell_2(I')}}{ |\beta_0|}
	\]
	for all $\beta $ that satisfy $ \tbe^T \beta =0$.
	
	In the following, we rigorously prove the last two claims this KKT condition using calculating directional derivative.
	\paragraph{KKT condition in the form of directional derivative and projected subgradient}
\begin{lemma}[Equivalent forms of KKT condition]
For the following optimization problem 
\[
 \begin{array}{ll}
 \underset{\psi}{\mbox{minimize}}   & \E_I \|\psi\cdot I\|_2\\
\mbox{subject to}  
& \langle   \tbe, \psi\rangle =1
\end{array}
\tag{$Q_1(\tbe)$ }
  \]
Let $P_{\tbe}^{\perp}$ be the projection onto the hyperplane as the orthogonal complement of  $\tbe$. The following are equivalent forms of KKT condition for $\be_0$ to be the optimal solution:   
  \BIT
  \item
  The directional derivative at $ \be_0$ along every direction $\beta$ on unit sphere where $\tbe^T \beta =0$ is non-negative:
	 \[
	\lim_{t\rightarrow 0^+} \frac{1}{t}[\Expect\|(\be_0+t\beta)\cdot I\|_2 - \Expect\|(\be_0)\cdot I\|_2] \geq 0.
	 \]
\item
  \[
 \begin{array}{ll}
0\in P_{\tbe}^{\perp} \Expect_I [N((\be_0)_I) ].
\end{array}
  \]
  \item
Equivalently, there exists a subgradient $g\in  P_{\tbe}^{\perp} \Expect_I [N((\be_0)_I) ]$ such that for all $\beta$ satisfying  $\tbe^T \beta = 0$,
\[
\beta^T g \geq 0.
\]
\item
Also equivalently, for all $\beta$,
\[
 \begin{array}{ll}
\sup_{g\in  [N((\be_0)_I)]}\Expect_I[ \langle   P_{\tbe}^{\perp}  \beta, g\rangle] \quad \geq 0.
\end{array}
  \]
  \EIT
\end{lemma}

\paragraph{KKT condition in directional derivative}
The following upper bound $ 1- \frac{|\tbe|_{(2)}}{|\tbe|_{(1)}}$ of $p^\star$ generalized the previous special case of exponential decay filter in theorem~\ref{thm: single root PT} with $p^\star = 1-|s|$:
\begin{lemma}
 The directional derivative at $e_0$ along $\beta$ 
 evaluates to the following:
 \BEAS
 \lim_{t\rightarrow 0^+}\frac{1}{t}(\Expect_I [\|(\be_0+t\beta)_I\|_2 - \|(\be_0)_I\|_2]) 
 &=&  p \cdot \langle\beta, \be_0\rangle + (1-p) \cdot \Expect_{I}[ \|\beta_{I}\|_2 | 0 \not \in I ]
 \EEAS
\end{lemma}
 \begin{proof}


Applying Lemma \ref{eq:ComputeExpect}, we proceed as follows:
\[
 \begin{array}{ll}
\lim_{t\rightarrow 0^+}\frac{1}{t}(\Expect_I [\|(\be_0+t\beta)_I\|_2 - \|(\be_0)_I\|_2])
&= \sup_{g\in \Expect_I [N((\be_0)_I)]} \langle\beta, g\rangle\\
&= p \cdot \langle\beta, \be_0\rangle + (1-p) \cdot  \sup_{g\in {\cal B}_0 } \langle\beta, g\rangle\\

&= p \cdot \langle\beta, \be_0\rangle + (1-p) \cdot \Expect_{I}[ \|\beta_{I}\|_2 | 0 \not \in I ] 
\end{array}
  \]
  \end{proof}

\subsection{Formula for Phase Transition Parameter}  
\paragraph{Reduction to $val(Q_1 (\tbe')).$}	
	We normalize the direction sequence $\beta$ so that $\beta_0= -1$; 
	using $\tilde{e}(0)=1$, we obtain a lower bound:
\[
 \inf_{\beta_0=-1, \langle  \tbe, \beta\rangle  =0} \Expect_{I'}\|\beta'\|_{\ell_2(I')}
 =\inf_{ \beta_0=-1, \beta_0 \tilde{e}_0-\langle  \tbe', \beta'\rangle  =0} \Expect_{I'}\|\beta'\|_{\ell_2(I')}
= \inf_{\langle   \tbe', \beta'\rangle  =1} \E_{I'} \|\beta'\|_{\ell_2(I')} = val(Q_1 (\tbe'))
\]	
Here $Q_1(\tbe')$ is the optimization problem:
\[
 \begin{array}{ll}
 \underset{\beta \in l_1(\integers)}{\mbox{minimize}}   &\E_{I'} \|\beta'\|_{\ell_2(I')}\\
\mbox{subject to}  
& \langle  \tbe', \beta'\rangle  =1
\end{array}
  \]
  
  Now we have rigorously proved that there exists a threshold $p^\star > 0$, so that for all 
\BIT
\item  $\bw^\star$ is $\ba^{-1}$ up to time shift and rescaling provided $p<p^\star$; and
\item   $\bw^\star$ is not $\ba^{-1}$ up to time shift and rescaling, provided $p> p^\star$.
\EIT
The threshold $p^\star$ obeys
 \[
\frac{p}{1-p}= val(Q_1(\tbe')).
\] 
  
  \section{Supplementary: Tight Upper and Lower Bound of Phase Transition Parameter}
  We have shown the existence of  $p^\star$ so that for all $p<p^\star$, the KKT condition is satisfied.   The threshold $p^\star$ determined by
 \[
\frac{p}{1-p}= val(Q_1(\tbe')).
\]
We have represented $p^\star$ as the optimal value of a derived optimization problem $val(Q_1 (\tbe'))$. From now on, we find upper and lower bounds of it.
\subsection{Upper and Lower Bound from Optimization Point of View}
	\begin{lemma}[Explicit phase transition condition with upper bound]
	\label{lemma:Condition phase transition upper}
$val(Q_1(\tbe'))$
obeys an upper bound and lower
\[
 \frac{p}{\|\tbe'\|_\infty}  \geq val(Q_1(\tbe'))  
\]
where $\|\tbe'\|_\infty  = \frac{|\tbe|_{(2)}}{|\tbe|_{(1)}}  $.
	Additionally, the upper bound 
is sharp if and only if
	\[
	\frac{p}{1-p} \leq val(Q_1(\frac{\tbe'}{\|\tbe'\|_\infty}))=val(Q_1(\tbe''))/\|\tbe'\|_\infty
	\]
	therefore, the upper bound 
holds with equality 
	\[
	p^\star= 1- \frac{|\tbe|_{(2)}}{|\tbe|_{(1)}}
	\]
if 
		\[
	p\leq 1- \frac{|\tbe|_{(3)}}{|\tbe|_{(2)}}
	\]
	Therefore, if 
		\[
	\frac{|\tbe|_{(3)}}{|\tbe|_{(2)}} \leq \frac{|\tbe|_{(2)}}{|\tbe|_{(1)}}
	\]
	then
	\[
	p^\star= 1- \frac{|\tbe|_{(2)}}{|\tbe|_{(1)}}.
	\]
	
	\end{lemma}

	\begin{lemma}[Explicit phase transition condition with lower bound]
\label{lemma:Condition phase transition lower}
\[
 val(Q_1(\tbe')) \geq  \E_I\frac{1}{\|\tbe\cdot I\|_{2} }\geq p^{-1/2} \|\tbe'\|_{2}^{-1}=p^{-1/2} \cot{\angle( \tbe, \be_0)}
\]
	\end{lemma}

  	\begin{proof}[Proof of Lemma~\ref{lemma:Condition phase transition upper}]
	
 Here $Q_1(\tbe')$ is the optimization problem:
\[
 \begin{array}{ll}
 \underset{\beta \in l_1(\integers)}{\mbox{minimize}}   &\E_{I'} \|\beta'_{I'}\|_{2}\\
\mbox{subject to}  
& \langle  \tbe', \beta'\rangle  =1
\end{array}
  \]
  
The upper bound is achieved at $\beta = e_{i_m}/|\tbe'_{i_m}| =e_{i_m}/\|\tbe'\|_\infty $, where $i_m = \arg\max_i |\tbe'|$.

The upper bound is tight (takes equality) if 
the projection pursuit problem $Q_1(\tbe')$  lead to one-sparse solution
$\beta = e_{i_m}/|\tbe'_{i_m}| =e_{i_m}/\|\tbe'\|_\infty $.
Using the previous condition, it require
\[
	val(Q_1(\frac{\tbe''}{\|\tbe'\|_\infty}))\geq \frac{p}{1-p}
	\]
	therefore, the upper bound 
holds with equality 
	\[
	p^\star= 1- \frac{|\tbe|_{(2)}}{|\tbe|_{(1)}}
	\]
if 
		\[
	p\leq val(Q_1(\frac{\tbe''}{\|\tbe'\|_\infty}))
	\]
	To simplify with upper bound on $val(Q_1(\frac{\tbe''}{\|\tbe'\|_\infty}))$, if 
		\[
	\frac{|\tbe|_{(3)}}{|\tbe|_{(2)}} \leq \frac{|\tbe|_{(2)}}{|\tbe|_{(1)}}
	\]
	then
	\[
	p^\star= 1- \frac{|\tbe|_{(2)}}{|\tbe|_{(1)}}.
	\]
		\end{proof}

	
		\begin{proof}[Proof of Lemma~\ref{lemma:Condition phase transition lower}]
In general, we define $Q_2(\tbe_J)$ for a fixed subset $J$:
\[
 \begin{array}{ll}
 \underset{\beta}{\mbox{minimize}}   &\|\beta\|_{2}\\
\mbox{subject to}  
&  \langle  \tbe_{J}, \beta\rangle  =1
\end{array}
\tag{$Q_2(\tbe_{J})$ }
  \] 	
then $val(Q_2(\tbe_{J})) = \|\tbe_J\|_2^{-1}$, where the optimal is achieved when 
  \[
  \beta_J = \tbe_J/\|\tbe_J\|_2^2 
  \]
then $val(Q_1 )$ has an lower bound:
\[
  val(Q_1(\tbe') ) \geq \E_I val(Q_2(\tbe'_I)) =\E_I\|\tbe_I\|_{2}^{-1} \geq p^{-1/2} \|\tbe'\|_{2}^{-1} 
  \]
  The last inequality is based on $V_{-1}\geq p^{-1/2}$ from theorem~\ref{thm:TaylorSqrtInvExpect}.
	\end{proof}
  \subsection{Upper and Lower Bound from Geometric Point of View}
  \paragraph{Geometric bound}
Let $\theta = \angle({\tbe, \be_0})$, for $\beta$ need to satisfy a constraint  $\tbe^T \beta =0$, we get
\[
 \begin{array}{ll}
\lim_{t\rightarrow 0^+}\frac{1}{t}(\Expect_I [\|(\be_0+t\beta)_I\|_2 - \|(\be_0)_I\|_2])
&= [p \cos{\angle(\beta, \be_0)} + (1-p)  V_1(\beta_{(0)}) \sin{\angle(\beta, \be_0)}] \|\beta\|_2.
\end{array}
  \]
Here $\angle(\beta, \be_0)\in [0, \pi], \sin{\angle(\beta, \be_0)}\in [0,1]$. 

\begin{figure}[!tb]
\centering
\includegraphics[width=.55\textwidth]{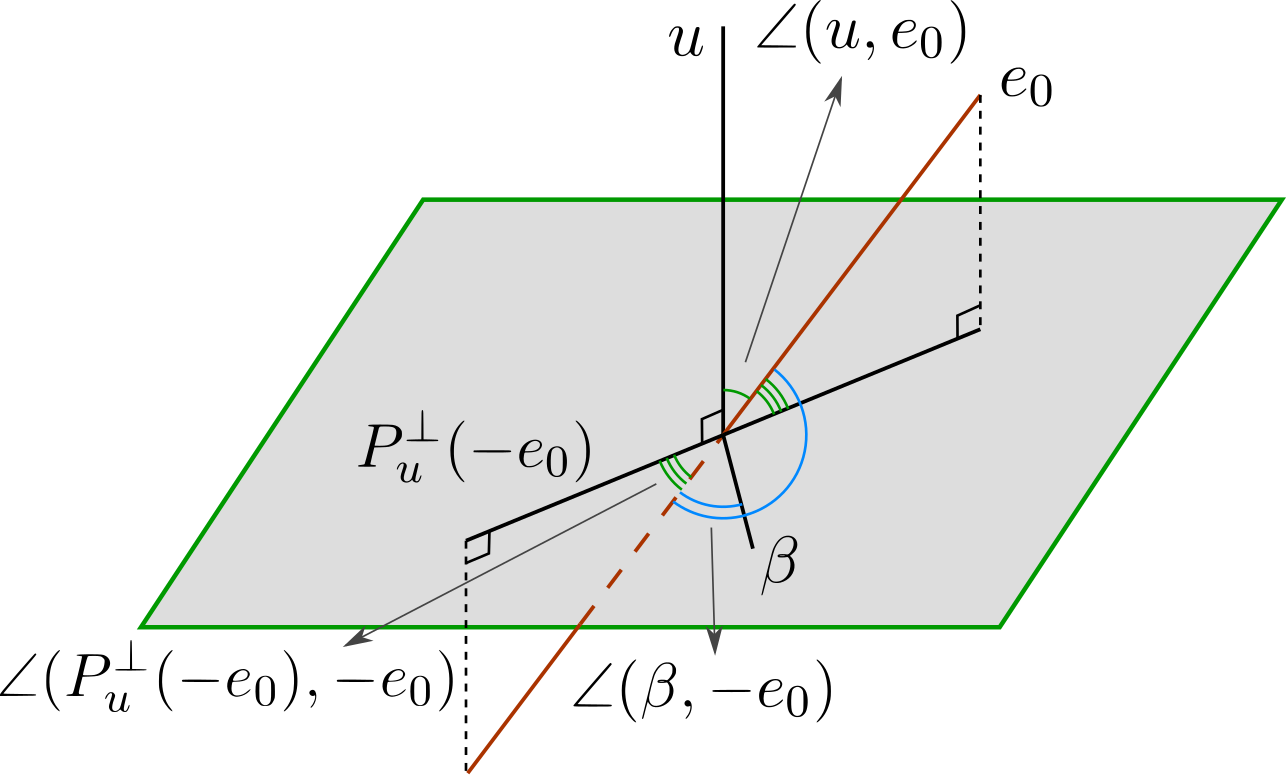}
\caption{Demonstration of the relation between
  $\angle( \beta, -\be_0)$ and ${\angle(\tbe, \be_0)}$. In the figure, $u=\tbe/\| \tbe\|_2$.}
\label{fig-Angle}
\end{figure}
  
\begin{lemma}
  \label{Angle_Inequality}
   \[
  \angle(  \beta, - \be_0) \geq \angle(P_{\tbe}^\perp (-\be_0), -\be_0)= \frac{\pi}{2} - \theta
  \]
  
    \[
  \tan(\angle(  \beta, - \be_0)) \geq \tan\angle(P_{\tbe}^\perp (-\be_0), -\be_0) =\cot{\theta}
  \]
\end{lemma}
  
\begin{proof}
   
We know that geometrically, using the property that projection of $e_0$ on the hyperplane with normal vector  $\tbe$ has the smallest angle among all the $\beta$ in that hyperplane, we have that if
   \[
  \angle(  \beta, - \be_0) \geq \angle(P_{\tbe}^\perp (-\be_0), -\be_0) = \frac{\pi}{2} - \angle(\tbe, \be_0) = \frac{\pi}{2} - \theta,
  \]
then
  \[
  \tan\angle(  \beta, - \be_0) \geq \tan\angle(P_{\tbe}^\perp (-\be_0), -\be_0)= \cot{\theta}.
  \]
 
\end{proof}
  

 
  
  

\begin{theorem}
 \label{thm:PT_lower}
The threshold $p^\star$ satisfies 
		\[ 
 p \cot\angle({\tbe, \be_0})   \leq val(Q_1(\tbe')) \leq	  \cot\angle({\tbe, \be_0})  V_1(\tbe').
		\]
\end{theorem}
 

\begin{proof}
  First, we prove the lower bound. Let $H_u=\{\beta : \quad u^T \beta =0 \}$,
    \BEAS
      val(Q_1(\tbe')) &=&  \inf_{\beta\in H_u} [V_1(\beta_{(0)}) \tan\angle(  \beta, - \be_0)]\\
      &\geq& \inf_{\beta\in H_u} V_1(\beta_{(0)}) \inf_{\beta\in H_u}\tan(\angle(\beta, -\be_0))\\
      &=& p \tan{(\angle(P_{\tbe}^\perp (-\be_0), -\be_0))}\\
       &=& p \cot\angle({\tbe, \be_0}) 
  \EEAS
 
Moreover, the lower bound is achieved when $\tbe'$ is one-sparse.  
 
For the upper bound, we plug in $\beta = P_{\tbe}^\perp (-\be_0)$, then
	\[ 
	\inf_{\| \beta\|_2=1, u^T \beta =0} 
	\tan(\angle( \beta, -\be_0))V_1(\beta_{(0)}) \leq
	\tan(\angle( P_{\tbe}^\perp (-\be_0), -\be_0))
	V_1(P_{\tbe}^\perp (-\be_0)).
	\]
It is worth commenting that since $V_1(P_{\tbe}^\perp (-\be_0))\leq \sqrt{p}$, and $\tan(\angle( P_{\tbe}^\perp (-\be_0), -\be_0))=\cot(\angle(\tbe, \be_0))=\cot\theta$, we have
	\[ 
	 \tan(\angle( P_{\tbe}^\perp (-\be_0), -\be_0))
 V_1(P_{\tbe}^\perp (-\be_0)) 
	\leq \cot(\angle(\tbe, \be_0)) \sqrt{p}.
	\]
 
  \end{proof}
 \subsection{Tighter Upper and Lower Bound from Refined Analysis} 
 \paragraph{Optimality by support}
Assume that $\tbe$ has $n_e$ non-zero entry on support $S_{\tbe}$, we can rank the absolute value of entries of $\tbe$ to be $|\tbe|_{(1)}, |\tbe|_{(2)}, |\tbe|_{(3)},\ldots,|\tbe|_{(n_e)}$, then the entries of $\tbe'$ will be ranked as $|\tbe|_{(2)}, |\tbe|_{(3)},\ldots,|\tbe|_{(n_e)}$. 

We know the optimal solution of $(\beta')^\star$ of $\inf_{\langle   \tbe', \beta'\rangle  =1} \E_{I'} \|\beta'\|_{\ell_2(I')}$ must have support $S_{\beta^\star}$ that satisfy $S_{\star}\subset S_{\tbe'}$. From the symmetry of objective, we know if $(\beta')^\star$ is $m$ sparse, then $m\leq n_e-1$ and, its support must be on the top $m$ entries $|\tbe|_{(2)}, |\tbe|_{(3)},\ldots,|\tbe|_{(m+1)}$, we call this support $S_m$, and we know the corresponding entries of $(\beta')^\star$ would have the same sign as entries of $\tbe$. 

We can define the $m$ sparse optimization problem for a random support function on the subset of $S_m$: $J_m \subset S_m$. Let $\beta$ be supported on $S_m$ and each entry non-negative, then
\[
val(Q_1^{S_m}) :=\inf_{ \beta: \sum_{j=1}^{m}  |\tbe|_{(j+1)} \beta_j  =1, \beta_j> 0} \E_{J_m} \|\beta \cdot J_m\|_{2} 
\]
Let $z_j:= \beta_j |\tbe|_{(j+1)}$, from the symmetry of objective and the order on $|\tbe|_{(j+1)}$, we know the solution must satisfy $0<z_m \leq z_{m-1} \leq \ldots \leq z_1$. Then we can recast the optimization problem as 
\[
val(Q_1^{S_m}) := 
\inf_{ z: \sum_{j=1}^{m}  z_j=1, 0<z_m \leq z_{m-1} \leq \ldots \leq z_1 } \E_{B} \sqrt{\sum_{j}B_j \frac{ 1}{|\tbe|^2_{(j+1)}} z_j^2 } 
\]
As a special case, when $m=1$, $val(Q_1^{S_1}) = \frac{p}{|\tbe|_{(2)}}$ as discussed above.

Then
\[
val(Q_1(\tbe'))= \inf_{ m\in \{1,2,\ldots, n_e-1\}}val(Q_1^{S_m})
\]
\paragraph{Tighter upper bound on $val(Q_1)$}
Now we can prove a more refine upper bound:
\begin{lemma}[Tighter upper bound on $val(Q_1)$]
    \[
val(Q_1(\tbe'))= \inf_{ m\in \{1,2,\ldots, n_e-1\}}val(Q_1^{S_m})
\leq \inf \{ \frac{p}{|\tbe|_{(2)}}, \frac{V_1(\tbe'_{S_2})}{\|\tbe'_{S_2}\|_2}, \frac{V_1(\tbe'_{S_3})}{\|\tbe'_{S_3}\|_2}, \ldots, \frac{V_1(\tbe')}{\|\tbe'\|_2}   \}
\]
where
\[
\cot\angle(\tbe_{S_m}, \be_0) =  \frac{1}{\|\tbe'_{S_m}\|_2}.
\]
\end{lemma}

\begin{proof}[Proof of lemma~\ref{lemma: tigher upper}]
From 
\[
val(Q_1^{S_m}) :=\inf_{ \beta: \sum_{j=1}^{m}  |\tbe|_{(j+1)} \beta_j  =1, \beta_j> 0} \E_{J_m} \|\beta\cdot J_m\|_{2} 
\]
we explore geometric point of view to find bounds.

For a fixed $m$, after re-ranking the entries by absolute value, let $S_m$ be the support so that only the top $m$ entries $|\tbe|_{(2)}, |\tbe|_{(3)},\ldots,|\tbe|_{(m+1)}$ are non-zero,  let 
\[
|\tbe'|_{S_m} = (0, |\tbe|_{(2)}, |\tbe|_{(3)},\ldots,|\tbe|_{(m+1)}, 0, \ldots, 0)
\]
\[
\|\tbe'_{S_m}\|_2^2 = |\tbe|^2_{(2)}+ |\tbe|^2_{(3)}+\ldots+|\tbe|^2_{(m+1)} 
\]
let the unit vector along the direction of $|\tbe'|_{S_m}$ be 
\[
u^m = |\tbe'|_{S_m}/\|\tbe'_{S_m} \|_2
\]
then 
\[
val(Q_1^{S_m}) := \frac{1}{\|\tbe'_{S_m}\|_2}\inf_{ \beta:  \beta_j> 0}  \E_{J_m} \| (\beta/\|\beta \|_{2}) \cdot J_m\|_{2} \cdot \frac{1}{\sum_{j=1}^{m}  u^m_j \frac{\beta_j }{\|\beta \|_{2}}}
\]
From previous definition, since $\beta$ is supported on $S_m$, we denote it as $\beta^m$, then
$V_1(\beta^m)= \E_{J_m} \| (\beta/\|\beta \|_{2}) \cdot J_m\|_{2} $,
$V_1(\tbe'_{S_m}) = V_1(|\tbe'|_{S_m}) = V_1(u^m)$.

Geometrically,
\[
\cos\angle(\beta^m, u^m) = \sum_{j=1}^{m}  u^m_j \frac{\beta_j }{\|\beta \|_{2}}
\]
\[
\cot\angle(\tbe_{S_m}, \be_0) =  \frac{1}{\|\tbe'_{S_m}\|_2}
\]
Since $\beta^m = u^m$ is a feasible point of the constraint, we have an upper bound 
\[
val(Q_1^{S_m}) \leq \frac{V_1(\tbe'_{S_m})}{\|\tbe'_{S_m}\|_2} 
= \cot\angle(\tbe_{S_m}, \be_0) V_1(\tbe'_{S_m})
\]
\BEAS
val(Q_1^{S_m}) 
&=& \cot\angle(\tbe_{S_m}, \be_0) \inf_{ \beta:  \beta_j> 0} V_1(\beta^m) \cdot \frac{1}{\cos\angle(\beta^m, u^m)}
\\
&=& \cot\angle(\tbe_{S_m}, \be_0)  V_1(|\tbe'|_{S_m}) \inf_{ \beta:  \beta_j> 0} \frac{V_1(\beta^m)}{V_1(u^m)} \cdot \frac{1}{\cos\angle(\beta^m, u^m)}
\EEAS
The lower bound is given by finding the lower bound of
\[
C(u^m):=\inf_{ \beta:  \beta_j> 0} \frac{V_1(\beta^m)}{V_1(u^m)} \cdot \frac{1}{\cos\angle(\beta^m, u^m)}
\]
We know upper bound $C(u^m)\leq 1$, and lower bound $C(u^m) \geq \frac{1}{p^{1/2}}$ based on the fact that $V_1 \in [p, \sqrt{p}]$.

\end{proof}

\section{Technical Tool: Tight Bound for Finite Difference of Objective}
We study the upper and lower bound of: 
\[
B(\be_0, \phi) := \frac{\E_I  \| (\be_0+\phi)_I \|_2 - \E_I  \| (\be_0)_I \|_2}{\|\phi\|_2} 
\]

This upper and lower bound allows us to connect objective $\E_I  \| \psi_I \|_2$ and the $2-$norm of $ \psi - \be_0$:
\[
\E_I  \| \psi_I \|_2 - \E_I  \| (\be_0)_I \|_2 = \E_I  \| \psi_I \|_2 - p=B(\be_0, \psi - \be_0) \| \psi - \be_0\|_2
\]

\paragraph{Bi-Lipschitzness of finite difference of objective near $e_0$ for linear constraint}
We normalize the problem by defining $t:=\|\phi\|_2, \beta:=\frac{\phi}{\|\phi\|_2}$.
After normalization, and taking into account the linear constraint  $\tbe^T \phi = 0$, we will study, for $\beta\in \mathcal{B}_t :=\{\|\beta\|_2=1, u^T \beta = 0,  \|e_0+t\beta\|_\infty\leq 1 \}$ for finite $t$, the upper and lower bound of 
\[
\frac{\E_I  \| (\be_0+t\beta)_I \|_2 - \E_I  \| (\be_0)_I \|_2}{t}.
\]

First, if $t\rightarrow 0^+$, then we get the directional derivative along direction of $\beta$. 
\BEAS
\lim_{t\rightarrow 0^+}\frac{\E_I  \| (\be_0+t \beta)_I \|_2 - \E_I  \| (\be_0)_I \|_2}{t}.
\EEAS

Due to convexity of the function $\beta \rightarrow \E_I  \| (\be_0+t \beta)_I \|_2$, we have the finite difference lower bounded by directional derivative:
\BEAS
\frac{\E_I  \| (\be_0+t \beta)_I \|_2 - \E_I  \| (\be_0)_I \|_2}{t} 
\geq 
 \lim_{t\rightarrow 0^+}\frac{\E_I  \| (\be_0+t \beta)_I \|_2 - \E_I  \| (\be_0)_I \|_2}{t} .
\EEAS


\begin{theorem}[Bi-Lipschitzness of finite difference of objective near $e_0$ for linear constraint]
We have upper and lower bound
\BEAS
0 \leq \frac{\E_I  \| (\be_0+t \beta)_I \|_2 - \E_I  \| (\be_0)_I \|_2}{t} -\lim_{t\rightarrow 0^+}\frac{\E_I  \| (\be_0+t \beta)_I \|_2 - \E_I  \| (\be_0)_I \|_2}{t}
\leq 
  \frac{t}{2}p\Big( 
   \beta_0^2 
 +
 p  (1-\beta_0^2) \Big) \leq \frac{pt}{2}.
\EEAS
This leads to
\[
0\leq B(\be_0, \phi) -	 \nabla_{\phi}{\E_I  \| (\be_0+\phi)_I \|_2}\mid_{\phi=0}\leq \frac{p}{2}\|\phi \|_2.
\]
\end{theorem}

\begin{proof}
As mentioned before, the lower bound is derived from convexity.

In the proof, for finite $t$, we calculate the finite difference condition on $I_0$:
\BEAS
\frac{\E_I  \| (\be_0+t \beta)_I \|_2 - \E_I  \| (\be_0)_I \|_2}{t}
&=&
(1-p)  \Expect_{I}[ \|\beta_{I}\|_2 | 0 \not \in I ] + p\big(\frac{\E_I [ \| \be_0+t\beta_I \|_2 -   1 | 0 \in I ]}{t} \big).
\EEAS
And
\BEAS
\lim_{t\rightarrow 0^+}\frac{\E_I  \| (\be_0+t \beta)_I \|_2 - \E_I  \| (\be_0)_I \|_2}{t}
&=&
(1-p)  \Expect_{I}[ \|\beta_{I}\|_2 | 0 \not \in I ] + p\beta_0.
\EEAS
In the following, we prove an upper bound 
\BEAS
 \frac{\E_I [ \| \be_0+t\beta_I \|_2 -   1 | 0 \in I ]}{t}- \beta_0
\leq   \frac{t}{2} p \| \beta_{{(0)}} \|_2^2  + \frac{t}{2} \beta_0^2.
\EEAS

For finite $t$,
\BEAS
\frac{\E_I [ \| \be_0+t\beta_I \|_2 -   1 | 0 \in I ]}{t}
&=&
\Expect_{I}[ \left( (\frac{1}{t}+ \beta_0)^2+\| \beta_{I_{(0)}} \|_2^2 \right)^{1/2} -\frac{1}{t}| 0 \in I ]\\
&=& \frac{1}{t}\Expect_{I}[ \left( 1+ 2t\beta_0 + t^2 \beta_0^2+t^2\| \beta_{I_{(0)}} \|_2^2 \right)^{1/2} -1| 0 \in I ].
\EEAS

Additionally, using concavity of the function $\sqrt{1+x}$, we have
\BEAS
 && \frac{1}{t}\Expect_{I}[ \left( 1+ 2t\beta_0 + t^2 \beta_0^2+t^2\| \beta_{I_{(0)}} \|_2^2 \right)^{1/2} -1| 0 \in I ]\\
&\leq&  \frac{1}{t}[ \left( 1+ 2t\beta_0 + t^2 \beta_0^2+t^2\Expect_{I}\| \beta_{I_{(0)}} \|_2^2| 0 \in I ] \right)^{1/2} -1]\\
&=&  \frac{1}{t}[ \left( 1+ 2t\beta_0 + t^2 \beta_0^2+t^2 p \| \beta_{{(0)}} \|_2^2 \right)^{1/2} -1]\\
&=&  \frac{1}{t}[ \left( 1+ 2t\beta_0 + t^2 \beta_0^2+t^2 p (1-\beta_0^2) \right)^{1/2} -1].
\EEAS
Now, apply the inequality: 
for any $x\geq -1$:
\[
\sqrt{1+x} -1 \leq \frac{x}{2},
\]
we have
\BEAS
 \frac{1}{t}[ \left( 1+ 2t\beta_0 + t^2 \beta_0^2+t^2 p \| \beta_{I_{(0)}} \|_2^2 \right)^{1/2} -1]
&\leq&  \beta_0 + \frac{t}{2} (p \| \beta_{{(0)}} \|_2^2  +  \beta_0^2).
\EEAS
Therefore, 
\BEAS
\frac{\E_I  \| (\be_0+t \beta)_I \|_2 - \E_I  \| (\be_0)_I \|_2}{t} -\lim_{t\rightarrow 0^+}\frac{\E_I  \| (\be_0+t \beta)_I \|_2 - \E_I  \| (\be_0)_I \|_2}{t}
\leq 
  \frac{t}{2}p\Big( 
   \beta_0^2 
 +
 p  (1-\beta_0^2) \Big) \leq \frac{pt}{2}.
\EEAS
The last inequality is due to $\beta_0^2 \leq 1$.

\end{proof}

\section{Supplementary: Main Result $2$ Proof: Guarantee for Finite Observation Window, Finite-Length Inverse }
\subsection{Phase Transition in  Finite Observation Window, Finite-Length Inverse Setting}
\paragraph{Finite sample setting}
\BIT
\item
Let $(X_t)_{t \in \integers}$ be IID sampled from  $pN(0,1) + (1-p) \delta_0$.
 \item
The filter $a\in l_1(\integers), (\mathcal{F} a)(\omega) \neq 0, \forall \omega\in T.$ Additionally, in finite sample setting, we assume $a^{-1}$ is zero outside of a centered window of radius $k$.
\item
Let $Y= a* X$ be a linear process, we are given a series of observations $(Y_t)_{t\in [-(T+k),T+k]}$ from a centered window of radius $T+k$.
\item
 We denote $X^{(i)}$ as the sequence being flipped and shifted for $i$ step from $X$ so that $X^{(i)}_k = X_{-k +i}$. Let $J_N$ be the support of $X$.
 We define the length $N=2T+1$ window as $\cT := [-T, T].$
\EIT
Now we want to find an inverse filter 
$w = (w_{-k}, \ldots, w_0, \ldots, w_k) \in V$ such that the convolution $w*Y$ is sparse,
\BEQ
 \begin{array}{ll}
 \underset{w\in V}{  \mbox{minimize} }   & \frac{1}{N}\|w*Y\|_{1,\cT}\\
 \mbox{subject to}  
& (\widetilde{a}* w)_0 =1.
\end{array}
\tag{${BD}_{\widetilde{a},N,k}$}
\EEQ
\paragraph{Finite sample directional derivative}
Now to study the finite sample convex problem, we need to calculate the directional derivative at $ \psi = e_0$.
Then
\begin{theorem}
\label{thm:fin_deriv}
For any $X$ sequence with support  $J_N$,
\BEAS
& & \lim_{t\rightarrow 0^+}\frac{1}{t}[(\hat{\E}  \| (e_0+t\beta)*X \|_1 - \hat{\E}  \| e_0*X \|_1)
\\
  &=& 
 \langle  \beta, {\frac{1}{N}\sum_{|i|\in \cT} \nabla_{\psi}L_i}|_{e_0} \rangle \\
 &=& 
 \frac{1}{N}\sum_{|i|\in \cT}  \mbox{sign}( X_{-i})X^{(i),T}\beta
 \\
  &=& 
 \frac{1}{N} \left(\|X\|_{1,J_N} \beta_0  +
 \|X*\beta_{(0)}\|_{1,\cT-J_N}+ 
\sum_{i\in J_N} (1_{X_{-i}> 0}-1_{X_{-i}< 0})  \cdot X^{(i),T}\beta_{(0)} \right).
 \EEAS
 
\end{theorem}

\subsection{Concentration of Objective}
 \paragraph{Concentration of $\|\psi*X\|_{1}$ }
 \begin{lemma}
 \label{lem:mu_min}
 Let 
 \[
\mu_{\mbox{min}}:=\frac{1}{\sqrt{2}\pi}\lambda_{m}(a)\sqrt{\frac{p}{k}} \|w\|_{1},
\]
then it is the lower bound of $\E \frac{1}{N}\|w*a*X\|_{1,\cT}$:
\[
\E \frac{1}{N}\|w*a*X\|_{1,\cT}  \geq \mu_{\mbox{min}}.
\]
 \end{lemma}

Now we define the intersection of $1-$norm ball of the $k-$dimensional subspace $V_k$ as $B_1(V_k)$, then we consider $w\in B_1(V_k).$ 
 
\begin{lemma}
\label{lem:tail_W}
Hence we consider $\psi = a*w \in a*B_1(V_k),$ and we define
\[
W:= \sup_{\psi \in a*B_1(V_k)}
|\frac{1}{N}\sum_{j\in [-T,T] } (|(\psi*X)_{j}| - \E |(\psi*X)_{j}|) |.
\]

For all $q\geq \max(2, \log(k)),$ there exists constant $C$,
\[
\E W^q\leq 2^{2q}(\frac{\lambda_{M}(a)}{N})^q   k C^q(\sqrt{Npq}+q)^q,
\]
so
\BEAS
 \|W\|_q \leq \frac{4eC}{N} \lambda_{M}(a)  
 (\sqrt{Np(q+\log(k))}+q +\log(k)),
\EEAS
and
\BEAS
P( W > \frac{4eC\lambda_{M}(a) }{N}  
 (\sqrt{Np\log(\frac{k}{\delta})}+\log(\frac{k}{\delta})))\leq \delta.
\EEAS
\end{lemma}

\begin{theorem}
\label{thm:concen_W}
Let the solution of finite sample convex optimization be $w^\star$, 
for small $\epsilon>0$ and concentration level $\delta>0$, 
there exists universal constant $C$, if there are 
\[
N \geq k\log(\frac{k}{\delta}) ( \frac{C\kappa_a}{\epsilon}  )^2,
\]
and  $ \frac{1}{N}\leq p$, we have:
\BEAS
P( W > \epsilon \mu_{\mbox{min}})\leq \delta.
\EEAS

  \end{theorem}
 
\subsection{Concentration for Directional Derivatives} 
\paragraph{Uniform bound of directional derivatives}
We  define the uniform bound of directional derivatives:
\[
W_D  =\sup_{\|\beta\|_2=1, u^T \beta =0, \beta +e_0 \in a*B_1(V_k)}|\langle  \beta, {\frac{1}{N}\sum_{|i|\in \cT} \nabla_{\phi}L_i}|_{\phi = 0} \rangle  -   \langle  \beta, \E \nabla_{\phi}L|_{\phi = 0} \rangle  |.
\]
Notice that
\BEAS
& &\langle  \beta, {\frac{1}{N}\sum_{|i|\in \cT} \nabla_{\phi}L_i}|_{\phi = 0} \rangle  -   \langle  \beta, \E \nabla_{\phi}L|_{\phi = 0} \rangle  
\\
&=& 
\left[\frac{1}{N} \left(\|X\|_{1,J_N} \beta_0  +
 \|X*\beta_{(0)}\|_{1,\cT-J_N}+ 
\sum_{i\in J_N} (1_{X_{-i}> 0}-1_{X_{-i}< 0})  \cdot X^{(i),T}\beta_{(0)}\right)\right]-
\\
& &
\left[\sqrt{\frac{2}{\pi}}(p \cdot \beta_0 + (1-p) \cdot \Expect_{I}[ \|\beta_{I}\|_2 | 0 \not \in I ])\right].
\EEAS

Combining the following three uniform bounds on
\[
|(\frac{1}{N} \|X\|_{1,J_N}-\sqrt{\frac{2}{\pi}}p)|,
\]
\[
|\frac{1}{N}\|X*\beta_{(0)}\|_{1,\cT-J_N}-\sqrt{\frac{2}{\pi}}(1-p) \cdot \Expect_{I}[ \|\beta_{I}\|_2 | 0 \not \in I ]|,
\]
\[
|\frac{1}{N}\sum_{i\in J_N} (1_{X_{-i}> 0}-1_{X_{-i}< 0})  \cdot X^{(i),T}\beta_{(0)}|,
\]
with high probability, we have  
\[
W_D  < \epsilon \mu_{\mbox{min}}.
\]

We define $\delta_p(N,\epsilon)$ the one side finite $N$ band such that for all $p<p^\star -\delta_p(N,\epsilon)$, 
\[
\langle  \beta, \E \nabla_{\phi}L|_{\phi = 0} \rangle >\epsilon \mu_{\mbox{min}}.
\]
for any direction $\beta \in V,$ the directional derivative at $\phi = 0$
\[
\langle  \beta, {\frac{1}{N}\sum_{|i|\in \cT} \nabla_{\phi}L_i}|_{\phi = 0} \rangle  >0,
\]
then using the convexity argument, $a*w^\star-e_0=0$. 

\paragraph{Finite sample guarantee}
From previous concentration of finite sample directional derivative, we have the following theorem.
	\begin{theorem}
	\label{thm:fin_guarantee_supp}
		When $a^{-1}$ is length $k$, for small constant $\epsilon>0, \delta>0$, when the number of observation $N$ satisfies
		\[
N \geq  k\log(\frac{k}{\delta}) ( \frac{C\kappa_a}{\epsilon}  )^2,
\]
then 
\[
W_D  < \epsilon \mu_{\mbox{min}}.
\]
For $\frac{1}{N}\leq p \leq p^\star - \delta_p(N,\epsilon)$,
\[
\langle  \beta, \E \nabla_{\phi}L|_{\phi = 0} \rangle >\epsilon \mu_{\mbox{min}}.
\]
therefore, 	with probability $1-\delta$, for any direction $\beta \in V,$ the directional derivative at $\phi = 0$
\[
\langle  \beta, {\frac{1}{N}\sum_{|i|\in \cT} \nabla_{\phi}L_i}|_{\phi = 0} \rangle  >0,
\]

		the solution of finite sample convex optimization $w^\star$ is $a^{-1}$ up to scaling and shift 
		with probability $1-\delta$.
	\end{theorem}

\subsection{Proof of Main Result $2$}	
	\begin{proof}[Proof of theorem~\ref{thm:fin_deriv}]
 \BEAS
& & \lim_{t\rightarrow 0^+}\frac{1}{t}[(\hat{\E}  \| (e_0+t\beta)*X \|_1 - \hat{\E}  \| e_0*X \|_1)
\\
  &=& 
 \langle  \beta, {\frac{1}{N}\sum_{|i|\in \cT} \nabla_{\psi}L_i}|_{e_0} \rangle \\
 &=& 
 \frac{1}{N}\sum_{|i|\in \cT}  \mbox{sign}( X_0^{(i)})X^{(i),T}\beta
 \\
 &=& 
 \frac{1}{N}\sum_{|i|\in \cT}  \mbox{sign}( X_{-i})X^{(i),T}\beta
 \\
 &=& 
\frac{1}{N}\sum_{|i|\in \cT} 1_{X_{-i}> 0}  \cdot X^{(i),T}\beta  +
\frac{1}{N}\sum_{|i|\in \cT} 1_{X_{-i}< 0}  \cdot (-X^{(i),T}\beta)  +
\frac{1}{N}\sum_{|i|\in \cT} 1_{X_{-i}= 0} \cdot |X^{(i),T}\beta|
 \\
 &=& 
\frac{1}{N}\sum_{|i|\in \cT} 1_{X_{-i}> 0}  \cdot X^{(i),T}\beta  +
\frac{1}{N}\sum_{|i|\in \cT} 1_{X_{-i}< 0}  \cdot (-X^{(i),T}\beta)  +
\frac{1}{N}\sum_{|i|\in \cT}  |X^{(i),T}\beta_{(0)}|
 \\
 &=& 
\frac{1}{N}\sum_{|i|\in \cT} 1_{X_{-i}> 0}  \cdot X_{-i}\beta_0  +
\frac{1}{N}\sum_{|i|\in \cT} 1_{X_{-i}< 0}  \cdot (-X_{-i}\beta_0)  + 
\frac{1}{N}  \|X*\beta_{(0)}\|_{1,\cT-J_N}+ 
\\
& &
\frac{1}{N}\sum_{|i|\in \cT} 1_{X_{-i}> 0}  \cdot X^{(i),T}\beta_{(0)}  +
\frac{1}{N}\sum_{|i|\in \cT} 1_{X_{-i}< 0}  \cdot (-X^{(i),T}\beta_{(0)})  
 \\
  &=& 
\frac{1}{N}\sum_{|i|\in \cT} 1_{X_{-i}\neq 0}  \cdot |X_{-i}|\beta_0  +
\frac{1}{N}  \|X*\beta_{(0)}\|_{1,\cT-J_N}+ 
\frac{1}{N}\sum_{i\in J_N} (1_{X_{-i}> 0}-1_{X_{-i}< 0})  \cdot X^{(i),T}\beta_{(0)}   
 \\
  &=& 
 \frac{1}{N} \left(\|X\|_{1,J_N} \beta_0  +
 \|X*\beta_{(0)}\|_{1,\cT-J_N}+ 
\sum_{i\in J_N} (1_{X_{-i}> 0}-1_{X_{-i}< 0})  \cdot X^{(i),T}\beta_{(0)} \right).
 \EEAS
\end{proof}

\begin{proof}[Proof of lemma~\ref{lem:mu_min}]
Without loss of generality, we rescale $w$ so that $\|w\|_1 =1$, now
\BEAS
\E_X Z 
&=& \frac{1}{N}\sum_{j\in [-T,T] }\E_X |\sum_{i=1}^k w_i (a*X)_{j-i}|\\
&\geq& \inf_{\|w\|_{1} = 1}\frac{1}{N}\sum_{j\in [-T,T] }\E_X |\sum_{i=1}^k w_i (a*X)_{j-i}|.
\EEAS
By symmetry, we know the minimizer is $w = (\frac{1}{k}, \ldots,\frac{1}{k}),$
then we use the symmetrization trick to insert a sequence of $k$ IID Rademacher ($\pm 1$)  random variables 
$\varepsilon_1, \ldots, \varepsilon_k$ in the second inequality, and then apply Khintchine inequality as the third inequality,
\BEAS
\E_X Z 
&\geq&  \frac{1}{k} \frac{1}{N}\sum_{j\in [-T,T] }\E_X |\sum_{i=1}^k (a*X)_{j-i}|\\
&\geq& \frac{1}{2} \frac{1}{k} \frac{1}{N}\sum_{j\in [-T,T] }\E_X |\sum_{i=1}^k \varepsilon_i (a*X)_{j-i}|\\
&\geq&  \frac{1}{2} \frac{1}{k} \frac{1}{N}\sum_{j\in [-T,T] }\frac{1}{\sqrt{2}}\E_X \sqrt{ \sum_{i=1}^k  (a*X)^2_{j-i}}\\
&\geq&  \frac{1}{2} \frac{1}{k} \frac{1}{N}\sum_{j\in [-T,T] }\lambda_{m}(a)\frac{1}{\sqrt{2}}\E_X \sqrt{ \sum_{i=1}^k  X^2_{j-i}}\\
&\geq& \frac{1}{2}  \frac{1}{k} \frac{1}{\sqrt{2}}\frac{1}{N}\sum_{j\in [-T,T] }\lambda_{m}(a)\frac{2}{\pi}\sqrt{pk}\\
&=&  \frac{1}{\sqrt{2}\pi}\lambda_{m}(a)\sqrt{\frac{p}{k}}.\\
\EEAS
\end{proof}
 
\begin{proof}[Proof of Lemma~\ref{lem:tail_W}]
Let $\varepsilon_1, \ldots, \varepsilon_N$ be a sequence of IID Rademacher ($\pm 1$)  random variables independent of $X$.  By the symmetrization inequality, see e.g. Lemma~$6.3$ of book \cite{ledoux2013probability}, we have the first inequality. Then since the function $t \rightarrow |t| $ is a contraction, an application of Talagrand’s contraction principle (see Lemma~$8$ of \cite{adamczak2016note}) with $F (x) = |x|^q$ 
conditionally on $X$ gives the second inequality.
\BEAS
\E W^q 
&\leq& 2^q \E \sup_{\psi \in a*B_1(V_k)} 
|\frac{1}{N}\sum_{j\in [-T,T] } \varepsilon_j |(\psi*X)_{j}| |^q\\
&\leq&
2^{2q}\E \sup_{\psi \in a*B_1(V_k)} |\frac{1}{N}\sum_{j\in [-T,T] } \varepsilon_j (\psi*X)_{j} |^q\\
&=&
2^{2q}\frac{1}{N^q}\E \sup_{\psi \in a*B_1(V_k)} |\sum_{i\in [-T,T] }\psi_{i}  \sum_{j\in [-T,T] } \varepsilon_j X_{j-i} |^q\\
&=&
2^{2q}\frac{1}{N^q}\E \sup_{w \in B_1(V_k)} |\sum_{i\in [-T,T] }(w*a)_{i}  \sum_{j\in [-T,T] } \varepsilon_j X_{j-i} |^q\\
&\leq&
2^{2q}\frac{1}{N^q}\E \max_{s\in [k]}|\sum_{i\in [-T,T] }a_{i-s}( \sum_{j\in [-T,T] } \varepsilon_j X_{j-i}) |^q\\
&\leq&
2^{2q}(\frac{\lambda_{M}(a)}{N})^q   \sum_{s\in [k] }\E| \sum_{j\in [-T,T] } \varepsilon_j X_{j-s} |^q.
\EEAS
Now by the moment version of Bernstein’s inequality (see Lemma~$7$, equation $(18)$ of~\cite{adamczak2016note}), we know that there exists a universal constant $C$,
\BEAS
\E| \sum_{j\in [-T,T] } \varepsilon_j X_{j-i} |^q \leq C^q(\sqrt{Npq}+q)^q.
\EEAS
Therefore, 
\BEAS
\E W^q &\leq& 2^{2q}(\frac{\lambda_{M}(a)}{N})^q    
\sum_{s\in [k] }\E| \sum_{j\in [-T,T] } \varepsilon_j X_{j-s} |^q \\
&\leq& 2^{2q}(\frac{\lambda_{M}(a)}{N})^q   k C^q(\sqrt{Npq}+q)^q.
\EEAS
Then
\BEAS
 \|W\|_q \leq \frac{4eC}{N} \lambda_{M}(a)  
 (\sqrt{Np(q+\log(k))}+q +\log(k)).
\EEAS
Therefore, we could get the tail bound using the Chebyshev inequality for the moments,
\BEAS
P( W > \frac{4eC\lambda_{M}(a) }{N}  
 (\sqrt{Np(q+\log(k))}+q +\log(k)))\leq e^{-q}.
\EEAS
We set $q = \log(\frac{1}{\delta})$, obtaining an upper bound for $W$ which is satisfied with probability at least $1-\delta$:
\BEAS
P( W > \frac{4eC\lambda_{M}(a) }{N}  
 (\sqrt{Np\log(\frac{k}{\delta})}+\log(\frac{k}{\delta})))\leq \delta.
\EEAS
\end{proof}

\begin{proof}[Proof of Theorem~\ref{thm:concen_W}]
From previous lemma, we know that with the conditions, 
\BEAS
P( W > \epsilon \mu_{\mbox{min}})\leq \delta.
\EEAS
Therefore, with probability at least $1-\delta$,
\BEAS
W = \sup_{\psi \in a*B_1(V_k)}
|\frac{1}{N}\sum_{j\in [-T,T] } (|(\psi*X)_{j}| - \E |(\psi*X)_{j}|) | \leq \epsilon \mu_{\mbox{min}}.
\EEAS
When this is true, for all $\psi \in a*B_1(V_k)$, we have a uniform bound:
\[
 \E |(\psi*X)_{j}|-\epsilon\mu_{\mbox{min}} \leq \frac{1}{N}\sum_{j\in [-T,T] } (|(\psi*X)_{j}| \leq \E |(\psi*X)_{j}|+\epsilon\mu_{\mbox{min}}.
\]
\end{proof}

\begin{proof}[Proof of Theorem~\ref{thm:fin_guarantee}]
Notice that
\BEAS
& &\langle  \beta, {\frac{1}{N}\sum_{|i|\in \cT} \nabla_{\phi}L_i}|_{\phi = 0} \rangle  -   \langle  \beta, \E \nabla_{\phi}L|_{\phi = 0} \rangle  
\\
&=& 
\left[\frac{1}{N} \left(\|X\|_{1,J_N} \beta_0  +
 \|X*\beta_{(0)}\|_{1,\cT-J_N}+ 
\sum_{i\in J_N} (1_{X_{-i}> 0}-1_{X_{-i}< 0})  \cdot X^{(i),T}\beta_{(0)}\right)\right]-
\\
& &
\left[\sqrt{\frac{2}{\pi}}(p \cdot \beta_0 + (1-p) \cdot \Expect_{I}[ \|\beta_{I}\|_2 | 0 \not \in I ])\right].
\EEAS
From the previous uniform bound, with high probability, we have the following three uniform bounds on
\[
|(\frac{1}{N} \|X\|_{1,J_N}-\sqrt{\frac{2}{\pi}}p)|,
\]
\[
|\frac{1}{N}\|X*\beta_{(0)}\|_{1,\cT-J_N}-\sqrt{\frac{2}{\pi}}(1-p) \cdot \Expect_{I}[ \|\beta_{I}\|_2 | 0 \not \in I ]|.
\]
The uniform bound on 
\[
|\frac{1}{N}\sum_{i\in J_N} (1_{X_{-i}> 0}-1_{X_{-i}< 0})  \cdot X^{(i),T}\beta_{(0)}|
\]
comes from symmetric distribution assumption.

Combining all three uniform bounds, with high probability, we have:
\[
W_D  < \epsilon \mu_{\mbox{min}}.
\]

We define $\delta_p(N,\epsilon)$ the one-side finite $N$ band such that for all $p<p^\star -\delta_p(N,\epsilon)$, 
\[
\langle  \beta, \E \nabla_{\phi}L|_{\phi = 0} \rangle >\epsilon \mu_{\mbox{min}}.
\]
for any direction $\beta \in V,$ the directional derivative at $\phi = 0$
\[
\langle  \beta, {\frac{1}{N}\sum_{|i|\in \cT} \nabla_{\phi}L_i}|_{\phi = 0} \rangle  >0,
\]
then using the convexity argument, $a*w^\star-e_0=0$.

\end{proof}





\section{Supplementary: Main Result 3 Proof: Stability Guarantee with Finite Length Approximation to Infinite Length Inverse  }
\subsubsection{Finite Length Approximation to Infinite Length Inverse Filter}
Now we consider a setting where there is a kernel whose corresponding inverse kernel has infinite support, and we give a finite-length approximation. 

Within the space $V_{-\infty,\infty}$ of bilaterally infinite real-valued sequences $(h(i))_{i \in \integers}$, consider the affine subspace $V_{N_{-}, N_{+}} = \{ (...,0,0,h_{-N_{-}},\ldots, h_{-1}, 1, h_1, h_2, \ldots, h_{N_{+}},0,0,...)\}$, an $N_{-}+N_{+}$-dimensional subspace of bilateral sequences with support at most $N_{-}+N_{+}+1$.
The coordinate that is fixed to one is 
located at index $i=0$.
Each coordinate is zero outside of a window
of size $N_{-}$ on the left of zero 
and  size $N_{+}$ on the right of zero.
The special sequence $e_0 = (\dots,0,0,1,0,0,\dots)$,
vanishing everywhere except the origin,
belongs to $V_{-\infty,\infty}$ and to every
$V_{N_-,N_+}$.

Let $a \in V_{N_-,N_+}$.
Then $a =  (...,0,0,a_{-N_{-}},\ldots, a_{-1}, 1, a_1, a_2, \ldots, a_{N_{+}},0,0,...)$.
We also write $a = e_0 + a_{(0)}$,
where $a_{(0)} \equiv (1-e_0) \cdot a$ 
denotes the `part of $a$ supported away from location
$i=0$'. We also write $a = e_0 + a_L + a_R$,
where $a_L = (a_L(i))_{i \in \integers} 
= (a(i) 1_{\{i < 0\}})_{i \in \integers}$ 
denotes the `part of $a$ supported to
the left of $i=0$' and 
where $a_R = (a_R(i))_{i \in \integers} 
= (a(i) 1_{\{i > 0\}})_{i \in \integers}$
denotes the `part of $a$
supported to the right of $i=0$'.
Finally, we say that $a \in V_{N_{-}, N_{+}}$ 
is a length $ L= N_{-}+N_{+}+1$ filter.

\textbf{First example}: let's look at a simple example of infinite length inverse filter approximation.
Let $s\in (-1,1)$. For $a=( 0, 1, -s)$, then $a^{-1} = (\ldots, 0, 1, s, s^2, s^3, \ldots)$ is an infinite length inverse filter. 

If we choose a length $r$ approximation $w^{r}= (\ldots, 0, 1, s, s^2, s^3, \ldots, s^{r-1})$ then
\[
\|a*w^{r}-e_0\|_2^2 = |s|^{2r}.
\]

Now we consider general finite-length forward filter with infinite-length inverse filter.

\subsection{Finite Length Approximation based on Z Transform}
We construct the finite-length approximation filter explicitly by truncation of Z-transform.

Let the Z-transform of $a$ be 
\[
 \begin{array}{ll}
A(z) =  \sum_{i=-s}^t a_{i} z^{-i}.
\end{array}
  \]
Then Z-transform of the inverse kernel $a^{-1}$ is $1/A(z)$.

\begin{theorem}
Assuming we have a finite length forward filter $a$ with its Z-transform having roots inside the unit circle, namely $s_{k}:= e^{-\rho_k + i \varphi_k }$ with $|s_k|<1$ and $\rho_k>0$
for $k \in \{-N_{-}, \ldots -1, 1, \ldots N_{+}\}$. Let $\cI = \{-N_{-},\ldots, -1, 1, \ldots, N_{+} \}$ as the set of all the possible indicies.
\[
 \begin{array}{ll}
A(z) &=  \sum_{i=-N_{-}}^{N_{+}} a_{i} z^{-i}
\\
&= c_0\prod_{j=1}^{N_{-}} (1- s_{-j} z) \prod_{i=1}^{N_{+}} (1- s_{i} z^{-1}),
\end{array}
  \]
Where $c_0$ is a constant to make sure that the coefficient $a_0=1$.

Then for a vector index $r=(r_{-N_-}, \ldots, r_{-1}, r_{1},\ldots, r_{N_+} ) $, we could construct an approximate inverse filter $w^r$ with Z-transform
\[
 \begin{array}{ll}
W(z) 
&=  \frac{1}{c_0}\prod_{j=1}^{N_{-}} (\sum_{\ell_j=0}^{r_{-j}-1} s^{\ell_j}_{-j} z^{-\ell_j} ) \prod_{i=1}^{N_{+}} (\sum_{\ell_i=0}^{r_{-j}-1} s^{\ell_i}_{i} z^{\ell_i} )\\
&=  \frac{1}{c_0}\prod_{j=1}^{N_{-}} (1- (s_{-j} z)^{r_{-j}})(1- s_{-j} z)^{-1} \prod_{i=1}^{N_{+}}(1- (s_{i}  z^{-1})^{r_i}) (1- s_{i} z^{-1})^{-1}.
\end{array}
\]
Let $\phi^r = w^r*a-e_0$, then
\[
 \begin{array}{ll}
\|\phi^r \|_2^2 = \|a*w^r-e_0\|_2^2
&= \sum_{n=1}^{|\cI|}\sum_{ k_1, \ldots, k_n  \in \cI } 
\exp \Big(- 2 \big(r_{k_1} \rho_{k_1}+ \ldots+  r_{k_n} \rho_{k_n} \big) \Big).
\end{array}
  \]
When $\min_j r_j \rightarrow \infty$, it converges to zero at an exponential rate. The convergence rate is determined by the slowest decaying exponential term as a function of $\min_i (r_i \rho_i)$.
\[
 \begin{array}{ll}
\|\phi^r \|_2
&= O(
\exp \Big(-  \min_i (r_i \rho_i) \Big) ), \quad \min_i r_i \rightarrow \infty.
\end{array}
  \]
 \end{theorem}
\subsection{Stability Theorem}
 \begin{theorem} 
 Let $a\in V_{N_-, N_+} $ be a forward filter with all the roots of Z-transform strictly in the unit circle. Let $w^\star \in V_{(r-1)N_-, (r-1)N_+} $ be the solution of the convex optimization problem. Let $w^r$ be the constructed filter in previous theorem with a uniform vector index $(r, \ldots, r, r,\ldots, r)$.
  Let $\phi^\star = w^\star*a-e_0$, 
 then the solution satisfy
  \[
 \begin{array}{ll}
  B(e_0,\phi^\star) \|\phi^\star \|_2 \leq B(e_0,\phi^r) \|\phi^r \|_2.
\end{array}
  \]
  Additionally, as $p<p^\star$, $B(e_0,\phi^\star)$ and $B(e_0,\phi^r)$ are both upper and lower bounded. 
  Therefore, using the previous asymptotic exponential convergence bound on $\|\phi^r \|_2$ as 
  $r \rightarrow \infty$, it converges to zero at an exponential rate
  \[
 \begin{array}{ll}
   \|\phi^\star \|_2 \leq \frac{B(e_0,\phi^r)}{B(e_0,\phi^\star)} \|\phi^r \|_2
   \leq
   O(
\exp \Big(-  r \min_i (\rho_i) \Big) ), \quad r \rightarrow \infty.
\end{array}
  \]
 \end{theorem}

\subsection{Proof of Stability Theorem}	

\begin{proof}[Proof of Theorem~\ref{thm:RootNorm}]
Now let $\psi^r$ have Z-transform $\Psi(z)$, $\phi = a*w-e_0$ with Z-transform $\Phi(z)$,
\[
 \begin{array}{ll}
\Psi(z): = A(z)W(z) 
&=   \prod_{j=1}^{N_{-}} (1- (s_{-j} z)^{r_{-j}}) \prod_{i=1}^{N_{+}}(1- (s_{i} z^{-1})^{r_i}).
\end{array}
\]
Therefore,
\[
 \begin{array}{ll}
\Phi(z): = A(z)W(z) -1 
&=  \prod_{j=1}^{N_{-}} (1- (s_{-j} z)^{r_{-j}}) \prod_{i=1}^{N_{+}}(1- (s_{i} z^{-1})^{r_i}) -1.
\end{array}
  \]
Let $\cI = \{-N_{-},\ldots, -1, 1, \ldots, N_{+} \}$ as the set of all the possible indicies.
And use polar representation of complex roots:
for any $k \in \cI = \{-N_{-},\ldots, -1, 1, \ldots, N_{+} \}$
\[
s_{k} := |s_{k}| e^{i \varphi_k } = e^{-\rho_k + i \varphi_k },
\]
where 
\[
 \begin{array}{ll}
\varphi_k := \operatorname{Im}(\log( s_{k} )) = -i \log(\frac{s_{k}}{|s_{k}|}),
\end{array}
\]
and since $|s_{k}|<1$ we have 
\[
\rho_k = -\operatorname{Re}(\log( s_{k} ))>0.
\]

Now we consider all $z = e^{2 \pi i t}$ on the unit circle for $t\in [\frac{-1}{2},\frac{1}{2}]$,
then for any $k \in \cI$
\[
 \begin{array}{ll}
- (s_{k} z^{-\mbox{sign}(k)})^{r_{k}} 
&= - |s_{k}|^{r_{k}} \exp \big(i [\varphi_k-  2 \pi t \mbox{sign}(k)]\cdot r_{k} \big)\\
&=   \exp \Big( -r_{k} \rho_k+ i \cdot \big( r_{k} (\varphi_k- 2 \pi t \mbox{sign}(k) ) + \pi  \big) \Big).
\end{array}
  \]
 Now we simplify the notation by defining 
 \[
 \upsilon_k(t) :=r_{k} (\varphi_k- 2 \pi t \mbox{sign}(k) ) + \pi,
 \]
then
 \[
 \begin{array}{ll}
\Phi(z)
&=  \prod_{j=1}^{N_{-}} (1- (s_{-j} z)^{r_{-j}}) \prod_{i=1}^{N_{+}}(1- (s_{i} z^{-1})^{r_{i}}) -1\\
&=  \prod_{k \in \cI} \Big(1- |s_{k}|^{r_{k}} \exp \big(i [\varphi_k-  2 \pi t \mbox{sign}(k)]\cdot r_{k} \big)\Big)  -1\\
&=  \prod_{k \in \cI} \bigg( 1+ \exp \Big( -r_{k} \rho_k+ i \cdot \big( r_{k} (\varphi_k- 2 \pi t \mbox{sign}(k) ) + \pi  \big) \Big) \bigg)     -1\\
&=  \prod_{k \in \cI} \bigg( 1+ \exp \Big( -r_{k} \rho_k+ i \cdot \upsilon_k(t) \Big) \bigg) -1.
\end{array}
  \]
  
 For any $z = \exp{(2 \pi i t )}$ on the unit circle,
 \[
 \begin{array}{ll}
\Phi\big(\exp{(2 \pi i t )}\big)
&=  \prod_{k \in \cI} \bigg( 1+ \exp \Big(- r_{k} \rho_k+ i \cdot \upsilon_k(t) \Big) \bigg)     -1\\
&= \sum_{n=1}^{|\cI|}\sum_{ k_1, \ldots, k_n  \in \cI } 
\exp \Big( -\big(r_{k_1} \rho_{k_1}+ \ldots+  r_{k_n} \rho_{k_n} \big)
+ i \cdot \big( \upsilon_{k_1}(t) + \ldots+  \upsilon_{k_n}(t)  \big) \Big).
\end{array}
  \]
 
Additionally, for any $t \in [\frac{-1}{2},\frac{1}{2}]$,
\[
 \begin{array}{ll}
 |\Phi\big(\exp{(2 \pi i t )}\big)|^2  
= 
\sum_{n=1}^{|\cI|}\sum_{ k_1, \ldots, k_n  \in \cI } 
\exp \Big( -2 \big(r_{k_1} \rho_{k_1}+ \ldots+  r_{k_n} \rho_{k_n} \big) \Big)
\end{array}
  \]
is independent of $t$, since the oscillation integral over the imaginary part is 
  \[
 \begin{array}{ll}
\int_{\frac{-1}{2}}^\frac{1}{2}
\exp  \bigg( i \cdot \Big(  \big(\upsilon_{k_1}(t) - \upsilon_{k_1}(t) \big) + \ldots+  \big(\upsilon_{k_n}(t) -\upsilon_{k_n}(-t)\big)  \Big) \bigg)  dt
= 1.
\end{array}
  \]
Therefore, using Fourier isometry,
 \[
 \begin{array}{ll}
 \|\phi\|_2^2 &= \int_{\frac{-1}{2}}^\frac{1}{2} |\Phi\big(\exp{(2 \pi i t )}\big)|^2 dt \\
&= \sum_{n=1}^{|\cI|}\sum_{ k_1, \ldots, k_n  \in \cI } 
\exp \Big( -2 \big(r_{k_1} \rho_{k_1}+ \ldots+  r_{k_n} \rho_{k_n} \big) \Big).
\end{array}
  \]
When $\min_j r_j \rightarrow \infty$, it converges to zero at an exponential rate. The convergence rate is determined by the slowest decay exponential term as a function of $\min_i (r_i \rho_i)$.
\[
 \begin{array}{ll}
\|\phi^r \|_2
&= O(
\exp \Big(-  \min_i (r_i \rho_i) \Big) ), \quad \min_i r_i \rightarrow \infty.
\end{array}
  \]
  
\end{proof}

\begin{proof}[Proof of Theorem~\ref{thm:bd_stab}]
Following the previous phase transition analysis, 
let the directional finite difference of the objective at $e_0$ be $D(w) = |(w*Y)_0| - |X_0|$, $I$ be the support of $\tilde{X}$, let $\phi = w*a-e_0$,  then we have 
\[
 \begin{array}{ll}
\Expect D(w) &= \sqrt{\frac{2}{\pi}}\Expect_I [\|(w*a)_I\|_2 - \|(e_0)_I\|_2]\\
&= \sqrt{\frac{2}{\pi}}\Expect_I [\|(e_0+\phi)_I\|_2 - \|(e_0)_I\|_2]\\
&= \sqrt{\frac{2}{\pi}}\|{\phi}\|_2 B(p, \phi) . \\
\end{array}
  \]
Now using the optimality of $w^\star \in V_{(r-1)N_-, (r-1)N_+} $, for all $w^r in V_{(r-1)N_-, (r-1)N_+}$,
we have
 \[
 \begin{array}{ll}
\Expect_I [\|(w^\star*a)_I\|_2 \leq \Expect_I [\|(w^r*a)_I\|_2,
\end{array}
  \]
  therefore,
  \[
 \begin{array}{ll}
B(p, \phi^\star) \|{\phi^\star}\|_2 \leq B(p, \phi)\|{\phi}\|_2.
\end{array}
  \]
Due to the bi-Lipschitz property of $\Expect_I \|(e_0+\phi)_I\|_2$ around $e_0$, we know when $p<p^\star$, $B(p, \phi)$ is positive and bounded by a constant.

Therefore, as $p<p^\star$, $B(e_0,\phi^\star)$ and $B(e_0,\phi^r)$ are both upper and lower bounded. 
 \[
 \begin{array}{ll}
   \|\phi^\star \|_2 \leq \frac{B(e_0,\phi^r)}{B(e_0,\phi^\star)} \|\phi^r \|_2.
\end{array}
  \]
 Using the previous theorem, 
 \[
 \begin{array}{ll}
   \|\phi^\star \|_2 \leq \frac{B(e_0,\phi^r)}{B(e_0,\phi^\star)} \sum_{n=1}^{|\cI|}\sum_{ k_1, \ldots, k_n  \in \cI } 
\exp \Big( -2 r\big(\rho_{k_1}+ \ldots+ \rho_{k_n} \big) \Big),
\end{array}
  \]
 using asymptotic exponential convergence bound on $\|\phi^r \|_2$ as 
  $r \rightarrow \infty$, it converges to zero at an exponential rate
  \[
 \begin{array}{ll}
   \|\phi^\star \|_2 
   \leq
   O(
\exp \Big(-  r \min_i (\rho_i) \Big) ), \quad r \rightarrow \infty.
\end{array}
  \]
\end{proof}

\section{Supplementary: Main Result $4$ Proof: Robustness Guarantee against Stochastic and Adversarial  Noises }
\subsection{Robustness Theorem against Stochastic Noise: Moving Average Gaussian Noise }
Consider a convoluted Gaussian noise with an average standard deviation $\sigma$ and a forward moving average filter $b$ with unit norm $\|b\|_2=1$, then
 \[
\begin{array}{ll}
Y&= a* (X+ \sigma b*G).
\end{array}
\]

This model would include the case of IID mixture of sparse Gaussian and small Gaussian $X_t \sim p N(0,1) +(1-p) N(0,\sigma_1)$ with $\sigma = \sigma_1$ and $b=e_0 $.
It also includes the Gaussian observation noise model
\[
\begin{array}{ll}
Y&= a* X+  \sigma_2 G_2 = a* (X+  \sigma_2 a^{-1}*G_2),
\end{array}
\]
with $b = a^{-1}/\|a^{-1}\|_2$, and $\sigma = \sigma_2 \|a^{-1}\|_2$.

 \begin{theorem}
 Let $\psi^\star$ be the solution of the convex optimization problem in Eq.(4.1) for the moving average random noisy model
  \[
\begin{array}{ll}
Y&= a* (X+ \sigma b*G),
\end{array}
\]
 then
\[
B(e_0, \psi^\star-e_0) \|\psi^\star -e_0\|_2 = \E_I  \| (\psi^\star)_I \|_2 -\E_I  \| (e_0)_I \|_2 \leq  (1-p) \sigma +
 p(\sqrt{1+ \sigma^2 }-1).
\]
When $\sigma\leq 1,$
\[
(1-p) \sigma +
 p(\sqrt{1+ \sigma^2 }-1)  \leq \sigma,
\]
therefore, when $p<p^\star$, $\sigma\leq 1,$ there exists a constant $C$,
\[
 \|\psi^\star -e_0\|_2 \leq C \sigma.
\]
\end{theorem}  

\subsection{Robustness Theorem against Adversarial Noise }
In the adversarial noise setting, we observe
\[
	\begin{array}{ll}
	Y&= a* (X+  \zeta) 
	\end{array}
	\]
where $\zeta$ is a sequence chosen by an adversary under constraint:	
\[
	\| \zeta\|_\infty \leq \eta.
	\]
 \begin{theorem}
	Under adversarial noise, 
Let $w^\star$ be the solution of the population convex optimization
	\[
 \begin{array}{ll}
 \underset{w }{\mbox{minimize}}   &\E \frac{1}{N}\|w*Y\|_{\ell_1(\cT_N)}\\
\mbox{subject to}  
& (\alpha*w)_0 =1.
\end{array}
  \]
Define $\psi^\star = a*w^\star$; then $\psi^\star$ satisfies the following bound:
 \[
 \begin{array}{ll}
B(e_0, \psi^\star-e_0) \|\psi^\star -e_0\|_2 = \E_I  \| (\psi^\star)_I \|_2 -\E_I  \| (e_0)_I \|_2 \leq 
 p\sqrt{\frac{2}{\pi}}(\mathcal{R}(\eta) - 1)+(1-p)\eta.
\end{array}
  \]  
  Here $\mathcal{R}(\eta)$
is the folded Gaussian mean, for standard Gaussian $G$:
\[
\mathcal{R}(\eta) :=  \sqrt{\frac{\pi}{2}} E_G |\eta+  G| = \exp\{-\eta^2/2\} + \sqrt{\frac{\pi}{2}}\eta \left(1 - 2\Phi\left(-\eta\right) \right).
\]  
$\mathcal{R}(\eta)-1$ is an even function that is monotonically non-decreasing for $\eta \geq 0$ with quadratic upper and lower bound: there exists constants $C_1\leq C_2$,
  \[
C_1 \eta^2 \leq \mathcal{R}(\eta)-1 \leq C_2 \eta^2, \forall \eta.
\]
  
Therefore, when $p<p^\star$, $B(e_0, \psi^\star-e_0)$ is a bounded positive constant, there exists a constant $C$, so that
 \[
 \|\psi^\star -e_0\|_2  \leq C  \eta, \forall \eta > 0.
 \]
\end{theorem}


\subsection{Proof of Robustness Theorem against Stochastic Noise }
\begin{proof}[Proof of Theorem~\ref{thm:rand_exact_rob}]
Consider the noisy model
 \[
\begin{array}{ll}
Y&= a* (X+  \sigma b*G).
\end{array}
\]
 First, 
\[
 \begin{array}{ll}
 F_{\sigma}(\psi)= \Expect_{G} \Expect_X|\psi^T X+\sigma \psi^T (b* G) | = 
   \Expect_{G} \Expect_I | \sqrt{ \|\psi_I\|_2^2  + \sigma^2\|C_{b}^T \psi\|_2^2}G |,
\end{array}
  \]
where $C_{b}$ is the Topelitz matrix with the first column being $b$.
  
Let $\psi^\star $ be the optimization solution, when $p<p^\star$,  we have  a chain of inequality:
    \[
 \begin{array}{ll}
F_{\sigma}(e_0)
 \geq F_{\sigma}(\psi^\star ) 
 \geq F_{0}(\psi^\star )
 \geq F_{0}(e_0).
\end{array}
  \]
  We have 
\[
 \begin{array}{ll}
 F_{\sigma}(e_0)
 &= 
\Expect_{G} \Expect_I | \sqrt{ \|(e_0)_I\|_2^2  + \sigma^2\|b\|_2^2}G | \\
&= 
\Expect_{G} \Expect_I | \sqrt{ \|(e_0)_I\|_2^2  + \sigma^2}G | \\
  &=
 \sqrt{\frac{2}{\pi}}[ 
 (1-p) \sigma +
 p(\sqrt{1+ \sigma^2 }) ].
\end{array}
  \]
  
Therefore, when $p<p^\star$
   \[
 \begin{array}{ll}
F_{\sigma}(e_0)-F_{0}(e_0) 
 \geq F_{0}(\psi^\star )-F_{0}(e_0) \geq 0.
\end{array}
  \]
  \[
B(e_0, \psi^\star-e_0) \|\psi^\star -e_0\|_2 = \E_I  \| (\psi^\star)_I \|_2 -\E_I  \| (e_0)_I \|_2 \leq  (1-p) \sigma +
 p(\sqrt{1+ \sigma^2 }-1). 
\]
When $\sigma\leq 1,$
\[
(1-p) \sigma +
 p(\sqrt{1+ \sigma^2 }-1)  \leq \sigma,
\]
therefore, when $p<p^\star$, $\sigma\leq 1,$ there exists a constant $C$,
\[
 \|\psi^\star -e_0\|_2 \leq C \sigma.
\]
\end{proof}

\subsection{Proof of Robustness Theorem against  Adversarial Noise }
Our data generative model is
\[
	\begin{array}{ll}
	Y&= a* (X+  \zeta),
	\end{array}
	\]
where 	
\[
	\| \zeta\|_\infty \leq \eta.
	\]


\begin{proof}[Proof of Theorem~\ref{thm:adv_exact_rob}]
The population convex optimization is
	\[
 \begin{array}{ll}
 \underset{w, z }{\mbox{minimize}}   & \E \frac{1}{N}\|w*Y\|_{\ell_1(\cT_N)}\\
\mbox{subject to}  
& (\alpha*w)_0 =1.
\end{array}
  \]
 By change of variable $\tilde{\psi} = w*a$, it could be reduced to a simpler problem 
  \[
 \begin{array}{ll}
 \mbox{minimize}_{w,s}   & \Expect|\psi^T (X+  \zeta)| \\
\mbox{subject to}  
& u^T \psi =1.
\end{array}
  \]
 Let $\psi^\star $ be the optimization solution of the worst-case objective  over all possible $\zeta$, defined as $F_{\eta}(\psi)$:
 \[
 \begin{array}{ll}
 F_{\eta}(\psi)
 &= \sup_{\zeta: \| \zeta\|_\infty \leq \eta } \Expect|\psi^T (X+ \zeta)| \\
 &= \sup_{\zeta: \| \zeta\|_\infty \leq \eta } \Expect_I  | \|\psi_I\|_2 G  + \psi^T \zeta  |,
\end{array}
  \] 
   when $p<p^\star$,  we have a chain of inequality:
    \[
 \begin{array}{ll}
F_{\eta}(e_0)
 \geq F_{\eta}(\psi^\star ) 
 \geq F_{0}(\psi^\star )
 \geq F_{0}(e_0).
\end{array}
  \]
   Therefore, when $p<p^\star$
   \[
 \begin{array}{ll}
F_{\eta}(e_0)-F_{0}(e_0) 
\geq F(\psi(\zeta),\zeta)-F_{0}(e_0)
 \geq F_{0}(\psi^\star )-F_{0}(e_0) \geq 0.
\end{array}
  \]
  Therefore,
 \[
 \begin{array}{ll}
B(e_0, \psi^\star-e_0) \|\psi^\star -e_0\|_2 = \E_I  \| (\psi^\star)_I \|_2 -\E_I  \| (e_0)_I \|_2 .
\end{array}
  \]   
  
From the folded Gaussian mean formula, let $G$ be scalar standard Gaussian, we have 
\[
 \begin{array}{ll}
 F_{\eta}(e_0)
  &= \sup_{\zeta: \| \zeta\|_\infty \leq \eta } \Expect_I  \E_G| \|(e_0)_I\|_2 G  +  \zeta_0  |\\
&= 
\sup_{\zeta: \| \zeta\|_\infty \leq \eta} [p\E_G | G  +  \zeta_0  |+ (1-p)|\zeta_0|]\\
  &=
\sup_{\zeta: \| \zeta\|_\infty \leq \eta}[  p \sqrt{\frac{2}{\pi}} \mathcal{R}(\zeta_0) + (1-p)|\zeta_0|]
\\
 &=
p \sqrt{\frac{2}{\pi}} \mathcal{R}(\eta) + (1-p)\eta.
\end{array}
  \]
 The last inequality comes from the fact that $\mathcal{R}(\zeta_0)$ is an even function that is monotonically non-decreasing for $\zeta_0 \geq 0$. 
 We will prove this conclusion below:
 \[
 \begin{array}{ll}
B(e_0, \psi^\star-e_0) \|\psi^\star -e_0\|_2 = \E_I  \| (\psi^\star)_I \|_2 -\E_I  \| (e_0)_I \|_2 \leq 
 (1-p)\eta + p \sqrt{\frac{2}{\pi}} (\mathcal{R}(\eta) - 1).
\end{array}
  \]

\end{proof}

\paragraph{Tool: folded Gaussian mean formula}
From general theory of folded Gaussian,
 $ |\mu+ \sigma G |, \quad  G \sim N(0, 1).$
Then its mean is
\beq
\label{FoldedGaussianMean}
E_G |\mu+ \sigma G| =  \sqrt{\frac{2}{\pi}} \sigma \exp\{-\mu^2/(2\sigma^2)\} + \mu \left(1 - 2\Phi\left(\frac{-\mu}{\sigma}\right) \right),
\eeq
where $\Phi$  is the normal cumulative distribution function:
\[
\Phi(x) = \frac 1 {\sqrt{2\pi}} \int_{-\infty}^x e^{-t^2/2} \, dt.
\]
As a related remark, its variance is
\beq
\label{FoldedGaussianVar}
\Var_G |\mu+ \sigma G| = E_G( |\mu+ \sigma G|-E_G |\mu+ \sigma G|)^2 =\sigma^2+ \mu^2 -(E_G |\mu+ \sigma G|)^2.
\eeq  

 \begin{lemma}
\label{lem:FoldedGaussian}
We define the ratio of folded Gaussian mean as 
\[
\mathcal{R}(\gamma) := \frac{E_G |\mu+ \sigma G|}{E_G |\sigma G|} = \frac{E_G |\mu+ \sigma G|}{\sqrt{\frac{2}{\pi}}\sigma} = \sqrt{\frac{\pi}{2}} E_G |\gamma+  G| = \exp\{-\gamma^2/2\} + \sqrt{\frac{\pi}{2}}\gamma \left(1 - 2\Phi\left(-\gamma\right) \right).
\]

$\mathcal{R}(\gamma)$ is an even function that is monotonically non-decreasing for $\gamma \geq 0$.

We have three different expressions (asymptotic expansion around  $\gamma$) for $\mathcal{R}(\gamma)$:
\BEA
\label{FoldedGaussianMean_ratio}
\mathcal{R}(\gamma)
&=&  \exp\{-\gamma^2/2\} + \sqrt{\frac{\pi}{2}}\gamma \left(1 - 2\Phi\left(-\gamma\right) \right)\\
&=&  \exp\{-\gamma^2/2\}\{1 + \frac{1}{2}\gamma \left[ \gamma + \frac{ \gamma^3}{3} + \frac{ \gamma^5}{3\cdot 5} + \cdots + \frac{ \gamma^{2n+1}}{(2n+1)!!} + \cdots\right]\}\\
&=& 1 + \frac{3}{2}\gamma^2 + \frac{1}{24} \gamma^4   - \frac{ \gamma^6}{120}  + O(\gamma^8).
\EEA
When $\gamma<1,$ 
\[
 (\frac{3}{2} + \frac{1}{24}) \gamma^2 \geq  \mathcal{R}(\gamma) -1\geq  \frac{3}{2}\gamma^2.
\]
When $\gamma<M$ for $M>1$, this lemma can be generalized:
there exists constants $C'\leq C$,
\[
C' \gamma^2 \leq \mathcal{R}(\gamma)-1 \leq C \gamma^2.
\]
 \end{lemma}
 
 \subsection{Technical Tool: Folded Gaussian Mean Formula }
\paragraph{Proof for folded Gaussian mean}
\begin{proof}[Proof of lemma~\ref{lem:FoldedGaussian}]
Plug in the formula for mean of folded Gaussian \ref{FoldedGaussianMean}, 
we have
\[
E_G |\mu+ \sigma G| = \sqrt{\frac{2}{\pi}} \sigma \exp\{-\gamma^2/2\} + \mu \left(1 - 2\Phi\left(-\gamma\right) \right).
\]
Divide it by 
\[
E_G |\sigma G| =   \sqrt{\frac{2}{\pi}} \sigma,
\]
we get the first equality.

Since the CDF of the standard normal distribution can be expanded by integration by parts into a series:
\[
\Phi( \gamma)=\frac{1}{2} + \frac{1}{\sqrt{2\pi}}\cdot e^{- \gamma^2/2} \left[ \gamma + \frac{ \gamma^3}{3} + \frac{ \gamma^5}{3\cdot 5} + \cdots + \frac{ \gamma^{2n+1}}{(2n+1)!!} + \cdots\right],
\]
where $!!$ denotes the double factorial. 

\[
1 - 2\Phi(-\gamma)  = \frac{1}{\sqrt{2\pi}}\cdot e^{- \gamma^2/2}\left[ \gamma + \frac{ \gamma^3}{3} + \frac{ \gamma^5}{3\cdot 5} + \cdots + \frac{ \gamma^{2n+1}}{(2n+1)!!} + \cdots\right].
\]
This gives the second inequality.

Additionally,
 \[
 \begin{array}{ll}
  e^{- \gamma^2/2} &=\sum_{n=0}^{\infty} \frac{ \gamma^{2n}}{(-2)^n (n)!}\\
  &=  1- \frac{ \gamma^2}{2} + \sum_{n=2}^{\infty} \frac{ \gamma^{2n}}{(-2)^n (n)!},
\end{array}
  \] 
then
 \[
 \begin{array}{ll}
& \exp\{-\gamma^2/2\}\{1 + \frac{1}{2}\gamma \left[ \gamma + \frac{ \gamma^3}{3} + \frac{ \gamma^5}{3\cdot 5} + \cdots + \frac{ \gamma^{2n+1}}{(2n+1)!!} + \cdots\right]\}\\
=& (\sum_{n=0}^{\infty} \frac{ \gamma^{2n}}{(-2)^n (n)!})\{1 + \frac{1}{2}\gamma \left[ \sum_{n=0}^{\infty} \frac{ \gamma^{2n+1}}{(2n+1)!!} \right]\}\\

=& 1 + \frac{3}{2}\gamma^2 + \frac{1}{24} \gamma^4   - \frac{ \gamma^6}{120}  + O(\gamma^8). 
\end{array}
  \] 
  
Let $C_\Phi( \gamma^2)$ be a function of $ \gamma$, then it is a composition with the inner function being $ \gamma^2,$ defined as follows: 
\[
 \begin{array}{ll}
C_\Phi( \gamma^2) 
&:= (2\Phi( \gamma)-1)/\left(\frac{ \gamma}{\sqrt{2\pi}}\right) \\
&= e^{- \gamma^2/2} \left[1 + \frac{ \gamma^2}{3} + \frac{ \gamma^4}{3\cdot 5} + \cdots + \frac{ \gamma^{2n}}{(2n+1)!!} + \cdots\right]\\
&= [\sum_{n=0}^{\infty} \frac{ \gamma^{2n}}{(2n+1)!!} ]/[\sum_{n=0}^{\infty} \frac{ \gamma^{2n}}{2^n (n)!}] \\
&= [\sum_{n=0}^{\infty} \frac{ \gamma^{2n}}{(2n+1)!!} ]/[\sum_{n=0}^{\infty} \frac{ \gamma^{2n}}{(2n)!!}] \\
&\leq 1.
\end{array}
\]
We remark further that
\[
 \begin{array}{ll}
C_\Phi( \gamma^2) 
&= [\sum_{n=0}^{\infty} \frac{ \gamma^{2n}}{(2n+1)!!} ]/[\sum_{n=0}^{\infty} \frac{ \gamma^{2n}}{(2n)!!}] \\
&= [1 + \frac{ \gamma^2}{3}+ +  \frac{ \gamma^4}{15}+ O( \gamma^6) ]/[1 + \frac{ \gamma^2}{2}+ +  \frac{ \gamma^4}{8}+O( \gamma^6)] \\
&= 1 -\frac{ \gamma^2}{6}+  \frac{ \gamma^4}{40}+ O( \gamma^6).
\end{array}
\]
Then
 \[
 \begin{array}{ll}
\mu \left(1 - 2\Phi\left(-\gamma\right) \right)
&= \frac{1}{\sqrt{2\pi}}\cdot \mu^2 \frac{1}{\sigma} C_\Phi(\gamma^2)\\
&= \frac{1}{\sqrt{2\pi}}\cdot \mu^2 \frac{1}{\sigma} (1-\gamma^2/6 +O(\gamma^4)).\\
\end{array}
\] 
Combining the above:
\[
 \begin{array}{ll}
E_G |\mu+ \sigma G|
&=   \left[ \sqrt{\frac{2}{\pi}} \sigma \exp\{-\gamma^2/2\} + \frac{1}{\sqrt{2\pi}}\cdot \mu^2 \frac{1}{\sigma} C_\Phi(\gamma^2)\right]\\
&=   \left[ \sqrt{\frac{2}{\pi}} \sigma \{ 1- \frac{\gamma^2}{2} + \sum_{n=2}^{\infty} \frac{\gamma^{2n}}{(-2)^n (n)!}\} + \frac{1}{\sqrt{2\pi}}\cdot \mu^2 \frac{1}{\sigma} C_\Phi(\gamma^2)\right]\\
&= \sqrt{\frac{2}{\pi}}   \left[ \sigma 
\{ 1 + (1+\frac{1}{2}\cdot C_\Phi(\gamma^2))\gamma^2 +   \sum_{n=2}^{\infty} \frac{\gamma^{2n}}{(-2)^n (n)!}\} \}\right].\\
\end{array}
\] 
Define the residual of this expansion as
\[
 \begin{array}{ll}
\mbox{Res}(\gamma^2) 
&:= \frac{ (2\Phi( \gamma)-1)/\left(\frac{ \gamma}{\sqrt{2\pi}}\right)-1}{2}\cdot \gamma^2 +  (\exp\{-\gamma^2/2\} - 1 + \gamma^2/2)
\\
&= \frac{ C_\Phi(\gamma^2)-1}{2}\cdot \gamma^2 +  (\exp\{-\gamma^2/2\} - 1 + \gamma^2/2) \\
&=    \frac{ C_\Phi(\gamma^2)-1}{2}\cdot \gamma^2  +   \{ \sum_{n=2}^{\infty} \frac{\gamma^{2n}}{(-2)^n (n)!}\} \} \\
&=    \frac{-\gamma^4 }{12}  +    \frac{ \gamma^6}{80} +\frac{\gamma^{4}}{(-2)^2 (2)!} +\frac{\gamma^{6}}{(-2)^3 (3)!}+
O(\gamma^8)\\
&=    \frac{\gamma^4 }{24}  - \frac{ \gamma^6}{120}  + O(\gamma^8).
\end{array}
\] 
We know that if $\gamma^2<1,$ 
\[
 \begin{array}{ll}
\mbox{Res}(\gamma^2) 
&=    \frac{\gamma^4 }{24}  - \frac{ \gamma^6}{120}  + O(\gamma^8)\\
&\geq 0.
\end{array}
\] 

When $\gamma<1,$ 
\[
 (\frac{3}{2} + \frac{1}{24}) \gamma^2 \geq  \mathcal{R}(\gamma) -1\geq  \frac{3}{2}\gamma^2.
\]
When $\gamma<M$ for $M>1$, this lemma can be generalized:
there exist constants $C'\leq C$,
\[
C' \gamma^2 \leq \mathcal{R}(\gamma)-1 \leq C \gamma^2.
\]
\end{proof}



\bibliographystyle{apalike}
 
 \bibliography{BD}
\end{document}